\newif\ifcomments   % this enables comments
\newif\ifanon       % anonymous
\newif\ifcrypto     % toggles where things go if it is the crypto version vs full version
\newif\ifllncs      % blarf

\commentstrue
\ifllncs
  \documentclass[runningheads,a4paper]{llncs}

\else
  \documentclass[11pt,pdfa]{article}
  \usepackage[in]{fullpage}
\fi
\usepackage{comment}
\usepackage{iftex}
\ifPDFTeX
  \usepackage[utf8]{inputenc}
  \usepackage[noTeX]{mmap}
  \usepackage[T1]{fontenc}
\fi
\ifLuaTeX
  \usepackage{luatex85}
  \usepackage[noTeX]{mmap}
\fi

\usepackage[shortlabels]{enumitem}
\usepackage{bm}
\usepackage{amsfonts}
\usepackage{amsmath}
\usepackage{amssymb}
\usepackage{amsthm}
\usepackage{xcolor}
\usepackage{graphicx}
\usepackage{qtree}
\usepackage{tree-dvips}
\usepackage{float}
\usepackage{hyperref}
\usepackage[nameinlink]{cleveref}
\usepackage{braket}
\usepackage{mathrsfs}
\usepackage{tikz}
\usepackage{qcircuit}

\ifllncs
  
\else
  \hypersetup{colorlinks={true},linkcolor={blue},citecolor=magenta}

  \theoremstyle{plain}
  \newtheorem{theorem}{Theorem}[section]
  \newtheorem{lemma}[theorem]{Lemma}
  \newtheorem{claim}{Claim}
  \newtheorem{corollary}[theorem]{Corollary}
  \newtheorem{proposition}[theorem]{Proposition}
  \newtheorem{definition}[theorem]{Definition}
  \newtheorem{remark}[theorem]{Remark}

  \usepackage[style=alphabetic,minalphanames=3,maxalphanames=4,maxnames=99,backref=true]{biblatex}

  \DeclareFieldFormat{eprint:iacr}{Cryptology ePrint Archive: \href{https://ia.cr/#1}{\texttt{#1}}}
  \DeclareFieldFormat{eprint:iacrarchive}{Cryptology ePrint Archive: \href{https://eprint.iacr.org/archive/#1}{\texttt{#1}}}
  \AtEveryBibitem{% Clean up the bibtex rather than editing it
     \clearlist{address}
     \clearfield{date}
     \clearfield{isbn}
     \clearfield{issn}
     \clearlist{location}
     \clearfield{month}
     \clearfield{series}
     
     \ifentrytype{book}{}{% Remove publisher and editor except for books
      \clearlist{publisher}
      \clearname{editor}
     }
  }
  
  \newcommand{\email}[1]{\href{mailto:#1}{\texttt{#1}}}
\fi
\newtheorem{fact}[theorem]{Fact}

\ifcomments
  \newcommand{\todo}[1]{\textcolor{red}{TODO: #1}}
  \ifllncs
    \renewcommand{\note}[1]{\textcolor{blue}{NOTE: #1}}
  \else
    \newcommand{\note}[1]{\textcolor{blue}{NOTE: #1}}
  \fi
  \newcommand{\pnote}[1]{}
  \newcommand{\ag}[1]{}
  \newcommand{\luowen}[1]{}
  \newcommand{\hnote}[1]{}
\else
  \newcommand{\todo}[1]{}
  \ifllncs\renewcommand{\note}[1]{}\else\newcommand{\note}[1]{}\fi
  \newcommand{\pnote}[1]{}
  \newcommand{\ag}[1]{}
  \newcommand{\luowen}[1]{}
  \newcommand{\hnote}[1]{}  
\fi

\newcommand{\ketbra}[2]{\ket{#1}\!\bra{#2}}
\renewcommand{\cal}[1]{\mathcal{#1}}
\newcommand{\R}{\mathbb{R}}
\newcommand{\C}{\mathbb{C}}
\newcommand{\N}{\mathbb{N}}
\newcommand{\E}{\mathop{\mathbb{E}}}
\newcommand{\Tr}{\mathrm{Tr}}

\newcommand{\sym}{\mathsf{sym}}
\newcommand{\eps}{\varepsilon}

\newcommand{\Haar}{\mathscr{H}}

\DeclareMathOperator{\TD}{TD}
\newcommand{\TraceDist}[2]{{\TD\left(#1, #2\right)}}
\newcommand{\norm}[1]{{\left\lVert#1\right\rVert}}
\newcommand{\diamondnorm}[1]{{\norm{#1}_\diamond}}
\newcommand{\operatornorm}[1]{{\norm{#1}}}

\newcommand{\bit}{{\{0, 1\}}}
\newcommand{\SR}{\mathsf{SR}}

\newcommand{\dmx}{\mathcal{D}}
\newcommand{\bfX}{\mathbf{X}}
\newcommand{\bfY}{\mathbf{Y}}

\newcommand{\ip}[2]{\langle #1 | #2 \rangle}
\newcommand{\secparam}{\lambda}

\newcommand{\commit}{\mathsf{Commit}}
\newcommand{\reveal}{\mathsf{Reveal}}

\DeclareFontFamily{U}{skulls}{}
\DeclareFontShape{U}{skulls}{m}{n}{ <-> skull }{}

\date{}
\title{Pseudorandom (Function-Like) Quantum State Generators: \\ New Definitions and Applications}
\ifanon
  
  \author{}
  \ifllncs
    \institute{}
  \fi
\else
  \ifllncs
    \author{Prabhanjan Ananth\inst{1} \and
        Aditya Gulati\inst{1} \and
        Luowen Qian\inst{2}\thanks{Supported by DARPA under Agreement No. HR00112020023.} \and
        Henry Yuen\inst{3}\thanks{Supported by AFOSR award FA9550-21-1-0040 and NSF CAREER award CCF-2144219. }}
    \institute{University of California, Santa Barbara, USA,  \email{prabhanjan@cs.ucsb.edu},  \email{adityagulati@ucsb.edu}\and Boston University, USA,  \email{luowenq@bu.edu} \and Columbia University, USA,  \email{hyuen@cs.columbia.edu}}
  \else
    \author{
    Prabhanjan Ananth\thanks{\email{prabhanjan@cs.ucsb.edu}} \vspace{0.3em}\\ \small{UCSB}
    \and Aditya Gulati\thanks{\email{adityagulati@ucsb.edu}} \vspace{0.3em}\\
    \small {UCSB}
    \and Luowen Qian\thanks{\email{luowenq@bu.edu}. Supported by DARPA under Agreement No. HR00112020023.} \vspace{0.3em} \\
    \small{Boston University}
    \and Henry Yuen\thanks{\email{hyuen@cs.columbia.edu}. Supported by AFOSR award
    FA9550-21-1-0040 and NSF CAREER award CCF-2144219. }\vspace{0.3em}\\
    \small{Columbia University}\vspace*{10pt}
    }
  \fi
\fi
\begin{document}

\maketitle

%\hnote{feel free to change the title; this is just an attempt to try to encapsulate what the paper is about} \pnote{looks good; maybe we can change it a bit to "Pseudorandom (Function-Like) Quantum States: New Definitions and Applications" to also capture the PRFS part?}\hnote{thumbs up}

\begin{abstract}
\noindent Pseudorandom quantum states (PRS) are efficiently constructible states that are computationally indistinguishable from being Haar-random, and have recently found cryptographic applications. We explore new definitions, new properties and applications of pseudorandom states, and present the following contributions:
\begin{enumerate}
    \item {\bf New Definitions}: We study variants of pseudorandom \emph{function-like} state (PRFS) generators, introduced by Ananth, Qian, and Yuen (CRYPTO'22), where the pseudorandomness property holds even when the generator can be queried adaptively or in superposition. We show feasibility of these variants assuming the existence of post-quantum one-way functions. %\luowen{What about PRU? Should mention this if we want to say this is incomparable to OWF. Could also mention quantum small range as a contribution and potential application in proving security of other pseudorandom quantum unitaries} \hnote{we can mention this in the ``our results'' section but I think may be too nuanced for the abstract.}
    % as well as applications to generically stretching the length of PRS generators. \hnote{do we have this last application?}
    \item {\bf Classical Communication}: We show that PRS generators with logarithmic output length imply commitment and encryption schemes with \emph{classical communication}. Previous constructions of such schemes from PRS generators required quantum communication. %We accomplish this using a new primitive we call \emph{verifiable tomography}. %\hnote{Add something about verifiable tomography as a new conceptual contribution}
    \item {\bf Simplified Proof}: We give a simpler proof of the Brakerski--Shmueli (TCC'19) result that polynomially-many copies of uniform superposition states with random binary phases are indistinguishable from Haar-random states.
    \item {\bf Necessity of Computational Assumptions}: We also show that a secure PRS with  output length logarithmic, or larger, in the key length necessarily requires computational assumptions. 
    %is a sharp threshold beyond which PRS generators start requiring computational assumptions. \ag{This isn't a sharp threshold, we just know it's between $c\log(\secparam)$ and $\log(\secparam)$.}%\hnote{not the best way of saying it; do we want to mention this in the abstract?}
\end{enumerate}

\end{abstract}

\newcommand{\tomography}{\mathsf{Tomography}}
\newcommand{\verify}{\mathsf{Verify}}
\newcommand{\enc}{\mathsf{Enc}}
\newcommand{\dec}{\mathsf{Dec}}
\newcommand{\majority}{\mathsf{majority}}
\newcommand{\unanimous}{\mathsf{unanimous}}
\newcommand{\valid}{\mathsf{Valid}}
\newcommand{\invalid}{\mathsf{Invalid}}
\newcommand{\prob}{\mathsf{Pr}}

\ifllncs\else
\newpage
\tableofcontents
\newpage
\fi

\section{Introduction}
\label{sec:intro}

The study of pseudorandom objects is central to the foundations of cryptography. After many decades, cryptographers have developed a deep understanding of the zoo of pseudorandom primitives such as one-way functions (OWF), pseudorandom generators (PRG), and pseudorandom functions (PRF)~\cite{goldreich1986construct,haastad1999pseudorandom}.

The study of pseudorandomness in the quantum setting, on the other hand, is just getting started.
Objects such as state and unitary $k$-designs have been studied extensively, but these are best thought of as quantum analogues of $k$-wise independent hash functions~\cite{ambainis2007quantum,dankert2009exact}. There are unconditional constructions of state and unitary designs and they do not imply any computational assumptions~\cite{ambainis2007quantum,roy2009unitary}. 

Quantum pseudorandomness requiring computational assumptions, in contrast, has been studied much less. Ji, Liu, and Song introduced the notion of \emph{pseudorandom quantum states (PRS)} and \emph{pseudorandom quantum unitaries (PRU)}~\cite{ji2018pseudorandom}. At a high level, these are efficiently sampleable distributions over states/unitaries that are computationally indistinguishable from being sampled from the Haar distribution (i.e., the uniform measure over the space of states/unitaries). Ji, Liu, and Song as well as Brakerski and Shmueli have presented constructions of PRS that are based on quantum-secure OWFs~\cite{ji2018pseudorandom,BS19,BrakerskiS20}. Kretschmer showed, however, that PRS do not necessarily imply OWFs; there are oracles relative to which PRS exist but OWFs don't~\cite{Kretschmer21}. This was followed by recent works that demonstrated the cryptographic utility of PRS: basic cryptographic tasks such as bit commitment, symmetric-key encryption, and secure multiparty computation can be accomplished using only PRS as a primitive~\cite{AQY21,MY21}. It is an intriguing research direction to find more cryptographic applications of PRS and PRU.

The key idea in~\cite{AQY21} that unlocked the aforementioned applications was the notion of a \emph{pseudorandom function-like state (PRFS) generator}. To explain this we first review the definition of PRS generators. A quantum polynomial-time (QPT) algorithm $G$ is a PRS generator if for a uniformly random key $k \in \{0,1\}^\secparam$ (with $\secparam$ being the security parameter), polynomially-many copies of the state $\ket{\psi_k} = G(k)$ is indistinguishable from polynomially-many copies of a state $\ket{\vartheta}$ sampled from the Haar measure by all QPT algorithms. One can view this as a quantum analogue of classical PRGs. Alternately, one could consider a version of PRS where the adversary only gets one copy of the state. However, as we will see later, the multi-copy security of PRS will play a crucial role in our applications.

The notion of PRFS generator introduced by~\cite{AQY21} is a quantum analogue of PRF (hence the name \emph{function-like}): in addition to taking in a key $k$, the generator $G$ also takes an \emph{input} $x$ (just like a PRF takes a key $k$ and an input $x$). Let $\ket{\psi_{k,x}} = G(k,x)$. The pseudorandomness property of $G$ is that for all sequences of inputs $(x_1,\ldots,x_s)$ for polynomially large $s$, averaged over the key $k$, the collection of states $\ket{\psi_{k,x_1}}^{\otimes t}, \ldots, \ket{\psi_{k,x_s}}^{\otimes t}$ for polynomially large $t$ is computationally indistinguishable from $\ket{\vartheta_1}^{\otimes t},\ldots,\ket{\vartheta_s}^{\otimes t}$ where the $\ket{\vartheta_i}$'s are sampled independently from the Haar measure. In other words, while PRS generators look like (to a computationally bounded distinguisher) they are sampling a \emph{single} state from the Haar measure, PRFS generators look like they are sampling \emph{many} (as compared to the key length) states from the Haar measure. Importantly, this still holds true even when the distinguisher is given the inputs $x_1,\ldots,x_s$.

As mentioned, this (seemingly) stronger notion of quantum pseudorandomness provided a useful conceptual tool to perform cryptographic tasks (encryption, commitments, secure computation, etc) using pseudorandom states alone. Furthermore, \cite{AQY21} showed that for a number of applications, PRFS generators with logarithmic input length suffices and furthermore such objects can be constructed in a black-box way from PRS generators.\footnote{However, unlike the equivalence between PRG and PRF in the classical setting~\cite{goldreich1986construct}, it is not known whether \emph{every} PRFS generator can be constructed from PRS generators in a black-box way.} %\pnote{the last line is a great question but we are not addressing it in this work?}

%In this paper we continue exploring the landscape of quantum pseudorandomness, with the goal of identifying natural concepts and understanding their properties more deeply. Ultimately, we hope to develop a theory of pseudorandom states and unitaries that is as rich as the theory of one-way functions, pseudorandom generators, and pseudorandom functions in classical cryptography. \\

%\noindent \pnote{A suggestion: maybe we can end this part of intro with something like this: }

Despite exciting progress in this area in the last few years, there is still much to understand about the properties, relationships, and applications of pseudorandom states. In this paper we explore a number of natural questions about pseudorandom states:
%are still many interesting directions that remain to be explored. 
\begin{itemize}
    \item {\em Feasibility of Stronger Definitions of PRFS}: In the PRFS definition of~\cite{AQY21}, it was assumed that the set of inputs on which the adversary obtains the outputs are determined ahead of time. Moreover, the adversary could obtain the output of PRFS on only classical inputs. This is often referred to as \emph{selective security} in the cryptography literature. For many interesting applications, this definition is insufficient\footnote{For example, the application of private-key encryption from PRFS as described in~\cite{AQY21} is only selectively secure. This is due to the fact that the underlying PRFS is selectively secure.}. This leads us to ask: \emph{is it feasible to obtain strengthened versions of PRFS that maintain security in the presence of adaptive and superposition queries?} %What applications are unlocked by such strengthened pseudorandomness guarantees?}
    
    \item {\em Necessity of Assumptions}: In the classical setting, essentially all cryptographic primitives require computational assumptions, at the very least $\sf P \neq {\sf NP}$. What computational assumptions are required by pseudorandom quantum states? The answer appears to depend on the output length of the PRS generator. Brakerski and Shmueli~\cite{BrakerskiS20} constructed PRS generators with output length $c \log \secparam$ for some $c > 0$ satisfying statistical security (in other words, the outputs are statistically close to being Haar-random). On the other hand, Kretschmer showed that the existence of PRS generators with output length $\secparam$ implies that ${\sf BQP} \neq {\sf PP}$~\cite{Kretschmer21}. This leads to an intriguing question: {\em is it possible to unconditionally show the existence of $n(\secparam)$-length output PRS, for some $n(\secparam) \geq \log(\secparam)$?}

%    It was observed by Brakeski and Shmueli~\cite{BrakerskiS20} that PRS with $c \cdot \log(\secparam)$-outputs can be achieved information-theoretically, for some constant $c > 0$. 

    \item {\em Necessity of Quantum Communication}: A common theme in all the different PRS-based cryptographic constructions of~\cite{AQY21,MY21} is that the parties involved in the system perform quantum communication. Looking forward, it is conceivable that quantum communication will be a much more expensive resource than having access to a quantum computer. Achieving quantum cryptography with classical communication has been an important direction, dating back to Gavinsky~\cite{Gav12}. We ask the following question: {\em is quantum communication inherent in  the cryptographic constructions based on PRS}?
 
\end{itemize}

%We believe these results contribute towards a growing theory of pseudorandom states and unitaries.
%that is as rich as the theory of one-way functions, pseudorandom generators, and pseudorandom functions in classical cryptography.

\subsection{Our Results}

We explore the aforementioned questions. Our results include the following.

\paragraph{Adaptive-Secure and Quantum-Accessible PRFS.} As mentioned earlier, the notion of PRFS given by~\cite{AQY21} has \emph{selective security}, meaning that the inputs $x_1,\ldots,x_s$ are fixed ahead of time. Another way of putting it is, the adversary can only make non-adaptive, classical queries to the PRFS generator (where by query we mean, submit an input $x$ to the generator and receive $\ket{\psi_{k,x}} = G(k,x)$ where $k$ is the hidden, secret key). 

We study the notion of \emph{adaptively secure PRFS}, in which the security holds with respect to adversaries that can make queries to the generator adaptively. We consider two variants of this: one where the adversary is restricted to making classical queries to the generator (we call this a \emph{classically-accessible adaptively secure PRFS}), and one where there are no restrictions at all; the adversary can even query the generator on a \emph{quantum superposition of inputs} (we call this a \emph{quantum-accessible adaptively secure PRFS}). These definitions can be found in \Cref{sec:adaptive-security}. 

We then show feasibility of these definitions by constructing classically- and quantum-accessible adaptively secure PRFS generators from the existence of post-quantum one-way functions. These constructions are given in \Cref{sec:classical:aprfs} and \Cref{sec:cons:qrfs} respectively.

\paragraph{A Sharp Threshold For Computational Assumptions.} In \Cref{sec:computational-assumptions} we show that there is a sharp threshold between when computational assumptions are required for the existence of PRS generators: we give a simple argument that demonstrates that PRS generators with $\log \secparam$-length outputs require computational assumptions on the adversary\footnote{We also note that there is a much more roundabout argument for a quantitatively weaker result: ~\cite{AQY21} constructed bit commitment schemes from $O(\log \secparam)$-length PRS. If such PRS were possible to construct unconditionally, this would imply information-theoretically secure bit commitment schemes in the quantum setting. However, this contradicts the famous results of~\cite{lo1997quantum,mayers1997unconditionally}, which rules out this possibility. Our calculation, on the other hand, directly shows that $\log \secparam$ (without any constants in front) is a sharp threshold.}. This complements the aforementioned result of Brakerski and Shmueli~\cite{BrakerskiS20} that shows $c \log \secparam$-length PRS for some $c > 0$ do not require computational assumptions. We also note that the calculations of~\cite{Kretschmer21} can be refined to show that the existence of $(1 + \epsilon)\log \secparam$-length PRS for all $\epsilon > 0$ implies that ${\sf BQP} \neq {\sf PP}$. 
%\pnote{I had the same confusion yesterday. I convinced myself that our result might be stronger. We currently know how to construct bit-commitments only from $O(\log(\secparam))+\omega(\log\log(\secparam))$-PRS. We don't know what constants are hidden in $O(\cdot)$ and $\omega(\cdot)$ factors whereas in our negative result, we can show the necessity of computational assumptions for any output length $\geq \log(\secparam)$. } \hnote{Agreed; this roundabout thing doesn't prove the tighter bound, but I thought it would be good to point out because someone might notice this as well. Added word ``weaker''} \pnote{for sure, I think this is a very good point.} \luowen{Hmm? What is the confusion? AQY showed commmitments from $2 \log\lambda + \omega(\log\log\lambda)$ so the constant is known. (hidden constants are not well defined for $\omega(\cdot)$.)}

\paragraph{PRS-Based Constructions With Classical Communication.}\ We show that bit commitments and pseudo one-time pad schemes can be achieved using only classical communication based on the existence of PRS with $\secparam$-bit keys and $O(\log(\secparam))$-output length. This improves upon the previous result of~\cite{AQY21} who achieved bit commitments and pseudo one-time pad schemes from PRS using quantum communication. However, we note that~\cite{AQY21} worked with a wider range of parameters, while our constructions are based on PRS with $O(\log(\secparam))$-output length.
%based on $O(\log(\secparam))$-output length PRS. 
\par En route, we use quantum state tomography (or tomography for short), a well studied concept in quantum information. Roughly speaking, tomography, allows for obtaining a classical string $u$ that captures some properties of an unknown quantum state $\rho$, given many copies of this state. 
\par We develop a new notion called \emph{verifiable tomography} that might particularly be useful in cryptographic settings. Verifiable tomography allows for verifying whether a given string $u$ is consistent (according to some prescribed verification procedure) with a quantum state $\rho$. We present the definition and instantiations of verifiable tomography in~\Cref{sec:tomography}. In~\Cref{sec:apps}, we use verifiable tomography to achieve the aforementioned applications. At a high level, our constructions are similar to the ones in~\cite{AQY21}, except that verifiable tomography is additionally used to make the communication classical. %We note that while~\cite{AQY21} work with a wider range of parameters, our constructions are based on PRS with $O(\log(\secparam))$-output length.

\paragraph{A Simpler Analysis of Binary-Phase PRS.} 
% analyzed the following construction of PRS. 
Consider the following construction of PRS. Let $\{F_k: \{0,1\}^n \to \{0,1\} \}_{k \in \{0,1\}^\secparam}$ denote a (quantum-secure) pseudorandom function family. Then $\{ \ket{\psi_k} \}_k$ forms a PRS, where $\ket{\psi_k}$ is defined as
\begin{equation}
    \label{eq:binary-phase}
    \ket{\psi_k} = 2^{-n/2} \sum_{x \in \{0,1\}^n} (-1)^{F_k(x)} \ket{x}~.
\end{equation}
In other words, the pseudorandom states are \emph{binary phase states} where the phases are given by a pseudorandom function. This is a simpler construction of PRS than the one originally given by~\cite{ji2018pseudorandom}, where the phases are pseudorandomly chosen $N$-th roots of unity with $N = 2^n$. Ji, Liu, and Song conjectured that the binary phase construction should also be pseudorandom, and this was confirmed by Brakerski and Shmueli~\cite{BS19}. 

We give a simpler proof of this in \Cref{sec:simpleranalysis:prs}, which may be of independent interest.

\section{Technical Overview}
%\subsection{Adaptive-Secure and Quantum-Accessible PRFS}

\subsection{Threshold For Computational Assumptions}
We show that PRS generators with $\secparam$-bit keys and $\log \secparam$-length outputs cannot be statistically secure. To show this we construct an inefficient adversary, given polynomially many copies of a state, can distinguish whether the 
state was sampled from the output distribution of a $\log \secparam$-length PRS generator or sampled from the Haar distribution on $\log \secparam$-qubit states with constant probability. 

\paragraph{Simple Case: PRS output is always pure.} Let us start with a simple case when the PRS generator is such that each possible PRS state is pure. Consider the subspace spanned by all possible PRS outputs. The dimension of the subspace spanned by these states is atmost $2^\secparam$: the reason being that there are at most $2^{\secparam}$ keys. Now, consider the subspace spanned by $t$-copies of PRS states. The dimension of this subspace is still at most $2^{\secparam}$ and in particular, independent of $t$. Define $P^{(t)}$ to be a projector (which could have an inefficient implementation) onto this subspace. By definition, the measurement of $t$ copies of the output of a PRS generator with respect to $P^{(t)}$ always succeeds. 

Recall that the subspace spanned by $t$-copies of states sampled from the Haar distribution (of length $\log \secparam$) is a symmetric subspace of dimension $\binom{2^\secparam + t - 1}{t}$. By choosing $t$ as an appropriate polynomial (in particular, set $t\gg \secparam$), we can make $\binom{2^\secparam + t - 1}{t} \gg 2^\secparam$, such that a measurement with $P^{(t)}$ on $t$ copies of states sampled from the Haar distribution fails with constant probability. Hence, an adversary, who just runs $P$, can successfully distinguish between $t$ copies of the output of a $\log \secparam$-length PRS generator and $t$ copies of a sample from a Haar distribution with constant probability. 

\paragraph{General Case.} Now let us focus on the case when the PRS generator can also output mixed states. Then we have 2 cases:
\begin{itemize}
    \item \emph{The majority of outputs of the PRS generator are negligibly close to a pure state:} In this case, we show that the previous approach still works. We replace the projector $P^{(t)}$ with a projection onto the space spanned by states closest to the output states of the PRS generator and we can show that modified projector still succeeds with constant probability.
    \item \emph{The majority of outputs of the PRS generator are not negligibly close to a pure state:} In this case, most PRS outputs have purity\footnote{A density matrix $\rho$ has purity $p$ if $\Tr(\rho^2)=p$.} non-negligibly away from $1$. Thus, we can violate the security of PRS as follows: run polynomially (in $\secparam$) many SWAP tests to check if the state is mixed or not. When the input state is from a Haar distribution, the test will always determine the input state to be pure. On the other hand, if the input state is the output of a PRS generator, the test will determine the input to be pure with probability that is non-negligibly bounded away from 1. Thus, this case cannot happen if the PRS generator is secure. %Hence, checking purity by enough swap-tests can successfully distinguish between the output of a $\log \secparam$-length PRS generator and samples from a Haar distribution with constant probability.
\end{itemize}

\noindent Details can be found in \Cref{sec:computational-assumptions}.

\subsection{Cryptographic Applications With Classical Communication}
We show how to construct bit commitments and pseudo one-time encryption schemes from $O(\log(\secparam))$-output PRS with classical communication. Previously,~\cite{AQY21} achieved the same result for a wider range of parameters. In this overview, we mainly focus on bit commitments since the main techniques used in constructing commitments will be re-purposed for designing pseudo one-time encryption schemes. 
%\par We construct a bit commitment scheme and a psuedo one-time pad scheme with only classical communication using $O\left(\log \secparam\right)$-length PRS. 
\par We use the construction of bit commitments from~\cite{AQY21} as a starting point. Let $d=O\left(\log \secparam\right)$, $n=O\left(\log \secparam\right)$ and $G$ is a $(d,n)$-PRFS generator\footnote{This in turn can be built from $O(\log(\secparam))$-output PRS as shown in~\cite{AQY21}.}. The commitment scheme from~\cite{AQY21} is as follows:
\begin{itemize}
    \item In the commit phase, the receiver sends a random $2^d n$-qubit Pauli $P=P_1\otimes P_2\otimes\cdots\otimes P_{2^d-1}$ to the sender, where each $P_i$ is an $n$-qubit Pauli. The sender on input bit $b$, samples a key $k$ uniformly at random from $\{0,1\}^\secparam$. The sender then sends the state $\rho = \bigotimes_{x\in[2^d]}P^b_x\sigma_{k,x}P^b_x$, where $\sigma_{k,x} = G(k,x)$ to the receiver.
    \item In the reveal phase, the sender sends $(k,b)$ to the receiver. The receiver accepts if $P^b\rho P^b$ is a tensor product of the PRFS evaluations of $(k,x)$, for all $x=0,\ldots,2^d-1$. 
\end{itemize}

\noindent To convert this scheme into one that only has classical comunication, we need a mechanism to generate classical information $c$ from $\rho$, where $\rho$ is generated from $(k,b)$ as above, that have the following properties:
\begin{enumerate}
    \item {\em Classical Description}: $c$ can be computed efficiently and does not leak any information about $b$.
    \item {\em Correctness}: $(k,b)$ is accepted as a valid opening for $c$, 
    \item {\em Binding}: $(k',b')$, for $b \neq b'$, is rejected as an opening for $c$
\end{enumerate}

\paragraph{State Tomography.} To design such a mechanism, we turn to quantum state tomography. Quantum state tomography is a process that takes as input multiple copies of a quantum state $\sigma$ and outputs a string $u$ that is close (according to some distance metric) to a classical description of the state $\sigma$. In general, tomography procedures require exponential in $d$ number of copies of a state and also run in time exponential in $d$, where $d$ is the dimension of the state. Since the states in question are $O(\log(\secparam))$-output length PRFS states, all the algorithms in the commitment scheme would still be efficient. 
\par Since performing tomography on a PRFS state does not violate its pseudorandomness property, the hiding property is unaffected.  
For achieving correctness and binding properties, we need to also equip the tomography process with a verification algorithm, denoted by ${\sf Verify}$. A natural verification algorithm that can be associated with the tomography procedure is the following: to check if $u$ is a valid classical description of a state $\sigma$, simply run the above tomography procedure on many copies of $\sigma$ and check if the output density matrix is close to $u$. 

\par More formally, we introduce a new tomography called verifiable tomography and we present a generic transformation that converts a specific tomography procedure into one that is also verifiable. We will see how verifiable tomography helps us achieve both correctness and binding. Before we dive into the new notion and understand its properties, we will first discuss the specific tomography procedure that we consider. 

%$(\sigma',u)$ can be done by running tomography on the state $\sigma'$ and checking if the output is close to $u$ (treating both as density matrices). Hence using an efficient $\tomography$ scheme at the end of the commit phase and using the verification method discribed above seems to work as long as we have an honest committer. For most tomography procedures to work efficiently, we need to fix the output length of the PRS to be $O(\log \secparam)$. 

\paragraph{Instantiation.} We develop a tomography procedure based on \cite{Lowe21} that outputs a denisity matrix close (constant distance away) to the input with $1-\mathsf{negl}(\secparam)$ probability. This is an upgrade to the tomography procedure in \cite{Lowe21}, the expected distance of whose output was a constant. To  achieve this, we make use of the fact that if we repeat \cite{Lowe21}'s tomography procedure polynomially many times, most output states cluster around the input at a constant distance with $1-\mathsf{negl}(\secparam)$ probability. We believe this procedure might be of independent interest. Details about this procedure can be found in \Cref{sec:tomography:procedures}.

\paragraph{Verifiable Tomography.} %For a malicious committer, we need something stronger. In particular, we want $c$ to fail the verification on every $(k',b')$ if $b\neq b'$. To do this we define a new primitive named \emph{verifiable tomography}. 
Verifiable tomography is a pair of efficient algorithms $(\tomography,\verify)$ associated with a family of channels $\Phi_\secparam$ such that the following holds:
\begin{itemize}
    \item \emph{Same-input correctness:} Let $u_1 = \tomography(\Phi_\secparam (x))$, then $\verify(u_1,\Phi_\secparam (x))$ accepts with high probability.
    \item \emph{Different-input correctness:} Let $u_1 = \tomography(\Phi_\secparam (x_1))$, and $x_1\neq x_2$, then $\verify(u_1,\Phi_\secparam (x_2))$ rejects with high probability.
\end{itemize}
\noindent The family of channels we consider corresponds to the PRFS state generation. That is, $\Phi_{\secparam}(x=(k,i))$ outputs $G(k,i)$. As mentioned earlier, we can generically convert the above instantiation into a verifiable tomography procedure. Let us see how the generic transformation works. 
\par For simplicity, consider the case when the underlying PRFS has perfect state generation, i.e., the output of PRFS is always a pure state. In this case, the verification algorithm is the canonical one that we described earlier: on input $u$ and PRFS key $k$, input $i$, it first performs tomography on many copies of $G(k,i)$ to recover $u'$ and then checks if $u$ is close to $u'$ or not. The same-input correctness follows from the tomography guarantee of the instantiation. To prove the different-input correctness, we use the fact that PRFS outputs are close to uniformly distributed and the following fact \cite[Fact 6.9]{AQY21}: for two  arbitrary $n$-qubit states $\ket{\psi}$ and $\ket{\phi}$, $$\E_{P\xleftarrow{\$}\mathcal{P}_n} \left[\left|\bra{\psi}P\ket{\phi}\right|^{2}\right]=2^{-n}.$$
\noindent Thus, if $x_1 \neq x_2$ then $u_1$ and $u_2$ are most likely going to be far and thus, differing-input correctness property is satisfied as well. 
\par The proofs get more involved when the underlying PRFS does not satisfy perfect state generation. We consider PRFS generators that satisfy recognisable abort; we note that this notion of PRFS can be instantiated from PRS, also with $O(\log(\secparam))$ outpout length, using~\cite{AQY21}. A $(d(\secparam),n(\secparam))$-PRFS generator $G$ has the \emph{strongly recognizable abort property} if its output can be written as follows: $G_\secparam(k,x) = Tr_{\mathcal{A}}\left(\eta \ketbra{0}{0} \otimes \ketbra{\psi}{\psi} + (1-\eta) \ketbra{\bot}{\bot}\right)$, where $\mathcal{A}$ is the register with the first qubit. Moreover, $\ket{\bot}$ is of the form $\ket{1}\ket{\widehat{\bot}}$ for some $n(\secparam)$-qubit state state $\ket{\widehat{\bot}}$ so that, $(\bra{0} \otimes \bra{\psi})(\ket{\bot})=0$. The same-input correctness essentially follows as before; however arguing differing-input correctness property seems more challenging.
\par Consider the following degenerate case: suppose $k$ be a key and $x_1,x_2$ be two inputs such that PRFS on input $(k,x_1)$ and PRFS on $(k,x_2)$ abort with very high probability (say, close to 1). Note that the recognizable abort property does not rule out this degenerate case. Then, it holds that the outputs $u_1 = \tomography(\Phi_\secparam (x_1))$ and $u_2 = \tomography(\Phi_\secparam (x_2))$ are close. $\verify(u_1,\Phi_\secparam (x_2))$ accepts and thus, the different-input correctness is not satisfied. To handle such degenerate cases, we incorporate the following into the verification procedure: on input $(u_1,u_2)$, reject if either $u_1$ or $u_2$ is close to an abort state. Checking whether a classical description of a state is close to an abort state can be done efficiently.

\paragraph{From Verifiable Tomography to Commitments.} Incorporating verifiable tomography into the commitment scheme, we have the following:
\begin{itemize}
    \item The correctness follows from the same-input correctness of the tomography procedure.
    \item The binding property follows from the different-input correctness of the tomography procedure.
    \item The hiding property follows from the fact that the output of a PRFS generator is indistinguishable from Haar random, even given polynomially many copies of the state. % and the tomography procedure runs is QPT.
\end{itemize}

%The psuedo one-time pad scheme incorporates the same techniques as above. The definitions as well as instantiations of the verifiable tomography procedure can be found in \Cref{sec:tomography}. Details along with complete proofs of the commitment scheme and the psuedo one-time pad scheme can be found in \Cref{sec:commitment} and \Cref{sec:encryption}, respectively.

\begin{comment}
\subsection{A Simpler Analysis of Binary-Phase PRS}
We consider the binary-phase PRS from~\Cref{eq:binary-phase}. We give a simpler proof of this construction being a PRS. To do this, we define some new notation for states which are superposition of a subset of basis states. We find a new charecterisation of the symmertic subspace and multiple copies of a Haar-random state in terms of the new notation introduced. We find that the binary-phase PRS state has a nice representation in terms of this representation as well. Details of the prrof can be found in~\Cref{sec:simpleranalysis:prs}.
\end{comment}

\section{Preliminaries}

We present the preliminaries in this section.
%Most of the text in this section is copied verbatim from~\cite{AQY21}.
We use $\secparam$ to denote the security parameter. We use the notation ${\sf negl}(\cdot)$ to denote a negligible function. %We use the abbreviation PPT to denote a probabilistic polynomial-time algorithm. 
\par We refer the reader to~\cite{nielsen_chuang_2010} for a comprehensive reference on the basics of quantum information and quantum computation. We use $I$ to denote the identity operator. We use $\dmx(\cal{H})$ to denote the set of density matrices on a Hilbert space $\cal{H}$.

\paragraph{Haar Measure.} The Haar measure over $\C^d$, denoted by $\Haar(\C^d)$ is the uniform measure over all $d$-dimensional unit vectors. One useful property of the Haar measure is that for all $d$-dimensional unitary matrices $U$, if a random vector $\ket{\psi}$ is distributed according to the Haar measure $\Haar(\C^d)$, then the state $U\ket{\psi}$ is also distributed according to the Haar measure. For notational convenience we write $\Haar_m$ to denote the Haar measure over $m$-qubit space, or $\Haar((\C^2)^{\otimes m})$.

\begin{fact}
\label{fact:haar-avg}
We have
\[
    \E_{\ket{\psi} \leftarrow \Haar(\C^d)} \ketbra{\psi}{\psi} = \frac{I}{d}~.
\]
\end{fact}

\subsection{Distance Metrics and Matrix Norms} 

\paragraph{Trace Distance.} Let $\rho,\sigma \in \dmx(\cal{H})$ be density matrices. We write $\TD(\rho,\sigma)$ to denote the trace distance between them, i.e.,
\[
    \TD(\rho,\sigma) = \frac{1}{2} \| \rho - \sigma \|_1
\]
where $\norm{X}_1 = \Tr(\sqrt{X^\dagger X})$ denotes the trace norm.
We denote $\norm{X} := \sup_{\ket\psi}\{\braket{\psi|X|\psi}\}$ to be the operator norm where the supremum is taken over all unit vectors.
For a vector $x$, we denote its Euclidean norm to be $\norm{x}_2$.

\paragraph{Frobenius Norm.} The Frobenius norm of a matrix $M$ is
$$\| M\|_F = \sqrt{\sum_{i,j} |M_{i,j}|^2} = \sqrt{\Tr\left(M M^{\dagger} \right)},$$
where $M_{i,j}$ denotes the $(i,j)^{th}$ entry of $M$.

\par We state some useful facts about Frobenius norm below. 

\begin{fact}\label{fact:frab:square}
For all matrices $A,B$ we have $\|A-B\|_F^2 = \|A\|_F^2 + \|B\|_F^2 - 2\Tr(A^{\dagger}B)$. 
\end{fact}

\begin{fact}\label{fact:frob_inequality}
Let $M_0,M_1$ be density matricies and $\ket{\psi}$ be a pure state such that $\bra{\psi}M_0\ket{\psi}\leq\alpha$ and $\| M_0-M_1\|_{F}^2 \leq \beta$, where $\beta + 2\alpha < 1$ then $$\bra{\psi}M_1\ket{\psi}\leq\alpha+\sqrt{\beta}+\sqrt{\left(2-2\alpha\right)\beta}.$$
\end{fact}
\begin{proof}
From~\cref{fact:frab:square}, we have the following: 
\begin{align*}
    \|M_0-\ketbra{\psi}{\psi}\|_F &= \sqrt{\|M_0\|_F^2 + \|\ketbra{\psi}{\psi}\|_F^2 - 2\Tr(M_0^{\dagger}\ketbra{\psi}{\psi})} \\
    &=\sqrt{\|M_0\|_F^2+1-2\bra{\psi}M_0\ket{\psi}}\\
    &\geq \sqrt{\|M_0\|_F^2+1-2\alpha}.
    \end{align*}
By triangle inequality, we know $$\|M_1\|_F\leq\|M_0\|_F+\|M_0-M_1\|_F\leq ||M_0||_F+\sqrt{\beta}.$$
Similarly by~\cref{fact:frab:square}, $$\|M_1-\ketbra{\psi}{\psi}\|_F = \sqrt{1+\|M_1\|_F^2-2\bra{\psi}M_1\ket{\psi}}\leq \sqrt{1+\left(\|M_0\|_F+\sqrt{\beta}\right)^2-2\bra{\psi}M_1\ket{\psi}}.$$
By triangle inequality, we know $\|M_0-\ketbra{\psi}{\psi}\|_{F}\leq \|M_1-\ketbra{\psi}{\psi}\|_{F} + \|M_0-M_1\|_{F}.$
Hence, $$\sqrt{1+\|M_0\|_F^2-2\alpha} \leq \sqrt{1+\left(\|M_0\|_F+\sqrt{\beta}\right)^2-2\bra{\psi}M_1\ket{\psi}}+\sqrt{\beta}.$$ By rearranging the terms, we get $$\bra{\psi}M_1\ket{\psi}\leq\alpha+\|M_0\|_F^2\sqrt{\beta}+\sqrt{\left(1+\|M_0\|_F^2-2\alpha\right)\beta}\leq \alpha+\sqrt{\beta}+\sqrt{\left(2-2\alpha\right)\beta}.$$ 
\end{proof}

\begin{fact}
\label{fact:frobenius:haar}
For any $0 \leq \eps \leq 1$, $$\prob_{\ket{\psi_1},\ket{\psi_2} \leftarrow \Haar_n}\left[ \|\ketbra{\psi_1}{\psi_1} - \ketbra{\psi_2}{\psi_2} \|_F^2 \leq  \eps \right] \leq \frac{1}{e^{2^n(1 - \frac{\eps}{2})}}.$$
\end{fact}
\begin{proof}
From~\Cref{fact:frab:square},
\begin{eqnarray*}
\|\ketbra{\psi_1}{\psi_1} - \ketbra{\psi_2}{\psi_2} \|_F^2 & = & \|\ketbra{\psi_1}{\psi_1}\|_F^2 +  \|\ketbra{\psi_2}{\psi_2} \|_F^2 - 2\Tr\left( \ketbra{\psi_1}{\psi_1} \ketbra{\psi_2}{\psi_2} \right) \\
& = & 2 - 2|\ip{\psi_1}{\psi_2}|^2
\end{eqnarray*}
\par Thus, we have the following: 
\begin{eqnarray*}
\prob_{\ket{\psi_1},\ket{\psi_2} \leftarrow \Haar_n}\left[ \|\ketbra{\psi_1}{\psi_1} - \ketbra{\psi_2}{\psi_2} \|_F^2 \leq  \eps \right] & = & \prob_{\ket{\psi_1},\ket{\psi_2} \leftarrow \Haar_n}\left[ |\ip{\psi_1}{\psi_2}|^2 \geq  1 - \frac{\eps}{2} \right] \\
& \leq & \frac{1}{e^{2^n(1 - \frac{\eps}{2})}},
\end{eqnarray*}
where the last inequality was shown in~\cite{BHH16} (Equation 14). 
\end{proof}

% \paragraph{General Measurements.} A \emph{general measurement} on a Hilbert space $\cal H$ is a set $M = \{ M_a \}_{a \in A}$ of operators acting on $\cal H$ indexed by some finite set $A$ of outcomes satisfying the completeness relation
% \[
%     \sum_{a \in A} M_a^\dagger M_a = I~.
% \]
% Applying the measurement $M$ to a density matrix $\rho \in \dmx(\cal H)$ corresponds to the following operation: outcome $a$ is obtained with probability $\Tr(M_a^\dagger M_a \rho)$, and the post-measurement state is defined to
% \[
%     \rho \mapsto \frac{ M_a \rho M_a^\dagger }{ \Tr(M_a^\dagger M_a \rho)}~.
% \]

% The following result, known as \emph{L\'{e}vy's Lemma}, expresses strong concentration of measure for the Haar measure. 

% \begin{fact}[L\'{e}vy's Lemma~\cite{hayden2006aspects}]
% \label{fact:levy}
% Let $f:\C^d \to \R$ be a function such that for all unit vectors $\ket{\psi},\ket{\phi}$ we have
% \[
%     | f(\ket{\psi}) - f(\ket{\phi}) | \leq K \cdot \norm{\ket{\psi} - \ket{\phi}}_2
% \]
% for some number $K > 0$. Then there exists a universal constant $C > 0$ such that
% \[
%     \Pr_{\ket{\psi} \leftarrow \Haar(\C^d)} \left [ | f(\ket{\psi}) - \E f | \geq \delta \right ] \leq \exp \Big ( - \frac{C d \delta^2}{K^2} \Big)
% \]
% where $\E f$ denotes the average of $f$ over the Haar distribution $\Haar(\C^d)$. 
% \end{fact}

\subsection{Quantum Algorithms}
\label{sec:algorithms}

A quantum algorithm $A$ is a family of generalized quantum circuits $\{A_\lambda\}_{\lambda \in \N}$ over a discrete universal gate set (such as $\{ CNOT, H, T \}$). By generalized, we mean that such circuits can have a subset of input qubits that are designated to be initialized in the zero state, and a subset of output qubits that are designated to be traced out at the end of the computation. Thus a generalized quantum circuit $A_\lambda$ corresponds to a \emph{quantum channel}, which is a is a completely positive trace-preserving (CPTP) map. When we write $A_\lambda(\rho)$ for some density matrix $\rho$, we mean the output of the generalized circuit $A_\lambda$ on input $\rho$. If we only take the quantum gates of $A_\lambda$ and ignore the subset of input/output qubits that are initialized to zeroes/traced out, then we get the \emph{unitary part} of $A_\lambda$, which corresponds to a unitary operator which we denote by $\hat{A}_\lambda$. The \emph{size} of a generalized quantum circuit is the number of gates in it, plus the number of input and output qubits.

We say that $A = \{A_\lambda\}_\lambda$ is a \emph{quantum polynomial-time (QPT) algorithm} if there exists a polynomial $p$ such that the size of each circuit $A_\lambda$ is at most $p(\lambda)$. Furthermore we say that $A$ is \emph{uniform} if there exists a deterministic polynomial-time Turing machine $M$ that on input $1^n$ outputs the description of $A_\lambda$. 

We also define the notion of a \emph{non-uniform} QPT algorithm $A$ that consists of a family $\{(A_\lambda,\rho_\lambda) \}_\lambda$ where $\{A_\lambda\}_\lambda$ is a polynomial-size family of circuits (not necessarily uniformly generated), and for each $\lambda$ there is additionally a subset of input qubits of $A_\lambda$ that are designated to be initialized with the density matrix $\rho_\lambda$ of polynomial length. This is intended to model non-uniform quantum adversaries who may receive quantum states as advice.
Nevertheless, the reductions we show in this work are all uniform.

The notation we use to describe the inputs/outputs of quantum algorithms will largely mimic what is used in the classical cryptography literature. For example, for a state generator algorithm $G$, we write $G_\lambda(k)$ to denote running the generalized quantum circuit $G_\lambda$ on input $\ketbra{k}{k}$, which outputs a state $\rho_k$. 
%\luowen{We could drop $1^n$ when $n$ could be implied by the length of the input?} \note{HY: dropping $1^n$, using $A_\lambda$ instead; that's what other papers seem to use.}

Ultimately, all inputs to a quantum circuit are density matrices. However, we mix-and-match between classical, pure state, and density matrix notation; for example, we may write $A_\lambda(k,\ket{\theta},\rho)$ to denote running the circuit $A_\lambda$ on input $\ketbra{k}{k} \otimes \ketbra{\theta}{\theta} \otimes \rho$. In general, we will not explain all the input and output sizes of every quantum circuit in excruciating detail; we will implicitly assume that a quantum circuit in question has the appropriate number of input and output qubits as required by context.

\subsection{Pseudorandomness Notions}
\label{sec:prs_notions}
\noindent Next, we recall the different notions of pseudorandomness. First, in~\Cref{sec:prf}, we recall (classical) pseudorandom functions (prfs) and consider two notions of security associated with it. Then in~\Cref{sec:prs}, we define pseudorandom quantum state (PRS) generators, which are a quantum analogue of pseudorandom generators (PRGs). Finally in~\Cref{sec:prfs}, we define pseudorandom function-like quantum state (PRFS) generators, which are a quantum analogue of pseudorandom functions. To make it less confusing to the reader, we use the abbreviation ``prfs'' (small letters) for classical pseudorandom functions and ``PRFS'' (all caps) for pseudorandom function-like states.

\subsubsection{Pseudorandom Functions}
\label{sec:prf}
We present two security notions of pseudorandom functions. First, we consider the notion of post-quantum security, defined below. 

\begin{definition}[Post-quantum pseudorandom functions]
\label{def:pqprfs}
\label{def:aprf}
We say that a deterministic polynomial-time algorithm $F:\{0,1\}^{\secparam} \times \{0,1\}^{d(\secparam)} \rightarrow \{0,1\}^{n(\secparam)}$ is a {\em post-quantum secure pseudorandom function (pq-prf)} if for all QPT (non-uniform) distinguishers $A=(A_{\secparam},\rho_{\secparam})$ there exists a negligible function $\eps(\cdot)$ such that the following holds:
 \[
    \left|\Pr_{k \leftarrow \{0,1\}^{\secparam} }\left[A_\lambda^{{\cal O}_{\sf prf}(k,\cdot)}(\rho_\lambda) = 1\right] - \Pr_{{\cal O}_{\sf Rand}}\left[A_\lambda^{{\cal O}_{\sf Rand}(\cdot)}(\rho_\lambda) = 1\right]\right| \le \eps(\lambda),
  \]
  where:
  \begin{itemize}
      \item ${\cal O}_{\sf prf}(k,\cdot)$, modeled as a classical algorithm, on input $x \in \{0,1\}^{d(\secparam)}$,  outputs $F(k,x)$.
      \item ${\cal O}_{\sf Rand}(\cdot)$, modeled as a classical algorithm, on input $x \in \{0,1\}^{d(\secparam)}$, outputs $y_x$, where $y_x \xleftarrow{} \{0,1\}^{n(\secparam)}$. 
  \end{itemize}
Moreover, the adversary $A_{\secparam}$ only has classical access to ${\cal O}_{\sf prf}(k,\cdot)$ and ${\cal O}_{\sf Rand}(\cdot)$. That is, any query made to the oracle is measured in the computational basis. 
\par We also say that $F$ is a $(d(\lambda),n(\lambda))$-pq-prf to succinctly indicate that its input length is $d(\lambda)$ and its output length is $n(\lambda)$.
\end{definition}

\noindent Next, we consider the quantum-query security, as considered by Zhandry~\cite{Zhandry12}. In this security notion, the adversary has superposition access to either ${\cal O}_{\sf prf}$ or ${\cal O}_{\sf Rand}$. By definition, quantum-query security implies post-quantum security. 
\par Unlike all the other pseudorandom notions considered in this section, we are going to use a different convention and allow the key length to be a polynomial in $\secparam$, instead of it being just $\secparam$. We also parameterize the advantage of the adversary. The motivation behind these changes in the definition will become clear in~\Cref{sec:cons:qrfs}. 

\begin{definition}[Quantum-query secure pseudorandom functions]
\label{def:qqprfs}
We say that a deterministic polynomial-time algorithm $F:\{0,1\}^{\ell(\secparam)} \times \{0,1\}^{d(\secparam)} \rightarrow \{0,1\}^{n(\secparam)}$ is a {\em quantum-query $\varepsilon$-secure pseudorandom function (qprf)} if for all QPT (non-uniform) distinguishers $A=(A_{\secparam},\rho_{\secparam})$ there exists a function $\eps(\cdot)$ such that the following holds:
 \[
    \left|\Pr_{k \leftarrow \{0,1\}^{\ell(\secparam}) }\left[A_\lambda^{\ket{{\cal O}_{\sf prf}(k,\cdot)}}(\rho_\lambda) = 1\right] - \Pr_{{\cal O}_{\sf Rand}}\left[A_\lambda^{\ket{{\cal O}_{\sf Rand}(\cdot)}}(\rho_\lambda) = 1\right]\right| \le \eps(\lambda),
  \]
  where:
\begin{itemize}
    \item ${\cal O}_{\sf prf}(k,\cdot)$ on input a $(d+n)$-qubit state on registers ${\bf X}$ (first $d$ qubits) and ${\bf Y}$, applies an $(n+d)$-qubit unitary $U$ described as follows: $U\ket{x}\ket{a} = \ket{x}\ket{a \oplus F(k,x)}$. It sends back the registers ${\bf X}$ and ${\bf Y}$. 
    \item ${\cal O}_{\sf Rand}(\cdot)$ on input a $(d+n)$-qubit state on registers ${\bf X}$ (first $d$ qubits) and ${\bf Y}$, applies an $(n+d)$-qubit unitary $R$ described as follows: $R \ket{x}\ket{a} = \ket{x}\ket{a \oplus y_x}$, where $y_x \leftarrow \{0,1\}^{n(\secparam)}$. It sends back the registers ${\bf X}$ and ${\bf Y}$.   
\end{itemize}
Moreover, $A_{\secparam}$ has superposition access to ${\cal O}_{\sf prf}(k,\cdot)$ and ${\cal O}_{\sf Rand}(\cdot)$. We denote the fact that $A_{\secparam}$ has quantum access to an oracle ${\cal O}$ by $A_{\secparam}^{\ket{{\cal O}}}$. 
\par We also say that $F$ is a $(\ell(\secparam),d(\lambda),n(\lambda),\eps)$-qprf to succinctly indicate that its input length is $d(\lambda)$ and its output length is $n(\lambda)$. When $\ell(\secparam)=\secparam$, we drop $\ell(\secparam)$ from the notation. Similarly, when $\eps(\secparam)$ can be any negligible function, we drop $\eps(\secparam)$ from the notation.    
\end{definition}

\noindent Zhandry~\cite{Zhandry12} presented a construction of quantum-query secure pseudorandom functions from one-way functions. 

\begin{lemma}[Zhandry~\cite{Zhandry12}]
Assuming post-quantum one-way functions, there exists quantum-query secure pseudorandom functions.
\end{lemma}

\paragraph{Useful Lemma.} We will use the following lemma due to Zhandry~\cite{C:Zhandry12}. The lemma states that any $q$-query algorithm cannot distinguish (quantum) oracle access to a random function versus a $2q$-wise independent hash function. We restate the lemma using our notation.  

\begin{lemma}[{\cite[Theorem 3.1]{C:Zhandry12}}]
\label{lem:qwise}
Let $A$ be a $q$-query algorithm. Then, for any $d,n \in \mathbb{N}$, every $2q$-wise independent hash function $H:\{0,1\}^{\ell(q)} \times \{0,1\}^{d} \rightarrow \{0,1\}^{n}$ satisfies the following: 
\[
    \left|\Pr_{k \leftarrow \{0,1\}^{\ell(q)} }\left[A_\lambda^{\ket{{\cal O}_{\sf H}(k,\cdot)}}(\rho_\lambda) = 1\right] - \Pr_{{\cal O}_{\sf Rand}}\left[A_\lambda^{\ket{{\cal O}_{\sf Rand}(\cdot)}}(\rho_\lambda) = 1\right]\right| = 0,
  \]
  where ${\cal O}_{\sf Rand}$ is as defined in~\Cref{def:qqprfs} and ${\cal O}_{\sf H}$ is defined similarly to ${\cal O}_{\sf prf}$ except that the unitary $U$ uses $H$ instead of $F$. 
\end{lemma}

\hnote{The citation of Lemma 2.8 is to Zha15b, but the text above says it's due to Zha12. It should be Zha12, right?}
\hnote{What is $\eps(\secparam)$? It should be a function of $d,q,n$ instead, right?}
\luowen{Fixed}

\subsubsection{Pseudorandom Quantum State Generators}
\label{sec:prs}
We move onto the pseudorandom notions in the quantum world. The notion of pseudorandom states were first introduced by Ji, Liu, and Song in~\cite{ji2018pseudorandom}. We reproduce their definition here:

\begin{definition}[PRS Generator~\cite{ji2018pseudorandom}]
\label{def:vanilla-prs}
We say that a QPT algorithm $G$ is a \emph{pseudorandom state (PRS) generator} if the following holds. 
\begin{enumerate}
    %\item \textbf{Key Generation}. For all $\lambda$, $K(1^n)$ outputs a classical key $k$. %the circuit $K_n$ takes no inputs and outputs an $n$-qubit classical density matrix. \note{HY: maybe it's better to define $K$ as a classical algorithm, this might be too general}
    
    \item \textbf{State Generation}. For all $\lambda$ and for all $k \in \{0,1\}^\lambda$, the algorithm $G$ behaves as
    \[
        G_\lambda(k) = \ketbra{\psi_{k}}{\psi_{k}}.
    \]
    for some $n(\lambda)$-qubit pure state $\ket{\psi_k}$. 
    
    %On input $1^n$, $k$ from the support of $K(1^n)$, and sufficiently many zero ancillas, the algorithm $G$ acts as a unitary and outputs $\ket{1^n, k} \otimes \ket{\psi_k} \otimes \ket{0 \cdots 0}$ for some $m(\lambda)$-qubit pure state $\ket{\psi_k}$. \note{HY: need to be more precise about what ``acts as a unitary'' means}
    
    % There exists a negligible function $\eps(\lambda)$ such that for all $n \in \N$ and $k \in \{0,1\}^n$ in the support of $K(1^n)$,
    % \[
    %     \TD \left ( G(1^n,k), \ketbra{k}{k} \otimes \ketbra{\psi_k}{\psi_k} \right) \leq \eps(\lambda)
    % \]
    % for some pure-state $\ket{\psi_k}$.
    %on input $1^n$, $k$, and sufficiently many zero ancillas, the algorithm $G$ acts as a unitary and outputs $\ket{1^n, k} \otimes \ket{\psi_k} \otimes \ket{0 \cdots 0}$ for some $n$-qubit pure state $\ket{\psi_k}$. 
    
    % \item \textbf{Testing}. There exists a negligible function $\eps(\lambda)$ such that for all $\lambda$, for all keys $k \in \{0,1\}^\lambda$, we have that
    % \[
    % \Pr \left [ T(1^n, k, \ket{\psi_k}) = 1 \right] \geq 1 - \eps(\lambda)~,
    % \]
    % and for all states $\ket{\theta}$ such that $|\langle \theta | \psi_{k} \rangle|^2 \leq \delta$, we have
    % \[
    % \Pr \left [ T(1^n, k, \ket{\theta}) = 1 \right] \leq \eps(\lambda) + \delta~.
    % \]
    
    \item \textbf{Pseudorandomness}. For all polynomials $t(\cdot)$ and QPT (nonuniform) distinguisher $A$ there exists a negligible function $\eps(\cdot)$ such that for all $\lambda$,  we have
    \[
        \left | \Pr_{k \leftarrow \{0,1\}^\lambda} \left [ A_\lambda (G_\lambda(k)^{\otimes t(\lambda)}) = 1 \right] - \Pr_{\ket{\vartheta} \leftarrow \Haar_{n(\lambda)}} \left [ A_\lambda (\ket{\vartheta}^{\otimes t(\lambda)}) = 1 \right] \right | \leq \eps(\lambda)~.
    \]
    %where $m(\lambda)$ denotes the output length of $G(1^n,\cdot)$ and $\Haar(m(\lambda))$ denotes the Haar measure on $m(\lambda)$ qubits.
\end{enumerate}
We also say that $G$ is a $n(\lambda)$-PRS generator to succinctly indicate that the output length of $G$ is $n(\lambda)$.
\end{definition}

\noindent Ji, Liu, and Song showed that post-quantum one-way functions can be used to construct PRS generators. 

\begin{theorem}[\cite{ji2018pseudorandom,BrakerskiS20}]
  If post-quantum one-way functions exist, then there exist PRS generators for all polynomial output lengths.
  %\luowen{This is actually shown by , the original work implicitly only showed this for $n = \omega(\log\lambda)$}
\end{theorem}

\begin{comment}
\paragraph{Stretch PRS.} Morimae and Yamakawa~\cite{MY21} recently considered the notion of PRS where the output length is larger than the key length. While they considered a single-copy version we present below a more general definition where the adversary can even receive multiple copies of the PRS state. 

\begin{definition}[Stretch PRS Generator] 
\label{def:stretchPRS}
We say that  $\{G_{\secparam}(\cdot)\}$ is a {\em stretch} $n(\secparam)$-PRS generator if $n(\secparam) > \secparam$. 
\end{definition}
\end{comment}

% \begin{definition}[Stretch PRS Generator]
% \label{def:stretch-prs}
% Let $(K,G)$ denote a PRS generator. We say that $(K,G)$ is furthermore a \emph{stretch PRS generator} if for all $\lambda$, the number of keys $k$ in the image of $K(1^n)$ is less than $2^n$. 
% \end{definition}

%\noindent \pnote{Does it make sense to look at a definition of stretch PRS where the adversary only gets as input one copy of the state?}

\subsubsection{Pseudorandom Function-Like State (PRFS) Generators}
\label{sec:prfs}
In this section, we recall the definition of pseudorandom function-like state (PRFS) generators by Ananth, Qian and Yuen~\cite{AQY21}. PRFS generators generalize PRS generators in two ways: first, in addition to the secret key $k$, the PRFS generator additionally takes a (classical) input $x$. The second way in which this definition generalizes the definition of PRS generators is that the output of the generator need not be a pure state.
\par However, they considered the  weaker selective security definition (stated below) where the adversary needs to choose all the inputs to be queried to the PRFS ahead of time. Later we will introduce the stronger and the more useful definition of adaptive security.

\begin{definition}[Selectively Secure PRFS generator]
\label{def:prfs}
We say that a QPT algorithm $G$ is a (selectively secure) \emph{pseudorandom function-like state (PRFS) generator} if for all polynomials $s(\cdot), t(\cdot)$, QPT (nonuniform) distinguishers $A$ and a family of indices $\left(\{x_1,\ldots, x_{s(\lambda)}\} \subseteq \{0,1\}^{d(\lambda)}\right)_\lambda$, there exists a negligible function $\eps(\cdot)$ such that for all $\lambda$,
    \begin{align*}
        &\Big | \Pr_{k \leftarrow \{0,1\}^\lambda} \left [ A_\lambda( x_1,\ldots,x_{s(\lambda)},G_\lambda(k,x_1)^{\otimes t(\lambda)},\ldots, G_\lambda(k,x_{s(\lambda)})^{\otimes t(\lambda)}) = 1 \right] \\
        & \qquad \qquad - \Pr_{\ket{\vartheta_1}, \ldots,\ket{\vartheta_{s(\lambda)}} \leftarrow \Haar_{n(\lambda)}} \left [ A_\lambda( x_1,\ldots,x_{s(\lambda)}, \ket{\vartheta_1}^{\otimes t(\lambda)},\ldots, \ket{\vartheta_{s(\lambda)}}^{\otimes t(\lambda)}) = 1 \right] \Big | \leq \eps(\lambda)~.
    \end{align*}
%\pnote{I think we should either remove $\rho$ and say the advice is implicit or we should explicitly add to the expression.}\note{made the advice implicit}
We say that $G$ is a $(d(\lambda),n(\lambda))$-PRFS generator to succinctly indicate that its input length is $d(\lambda)$ and its output length is $n(\lambda)$.
\end{definition}

Our notion of security here can be seen as a version of \emph{(classical) selective security}, where the queries to the PRFS generator are fixed before the key is sampled.  

\paragraph{State Generation Guarantees.} Towards capturing a natural class of PRFS generators,~\cite{AQY21} introduced the concept of \emph{recognizable abort}. At a high level, recognizable abort is the property that the output of PRFS can be written as a convex combination of a pure state and a known abort state, denoted by $\ket{\bot}$. In more detail,  the PRFS generator works in two stages. In the first stage it either generates a valid PRFS state $\ket{\psi}$ or it aborts. If it outputs a valid PRFS state then the first qubit is set to $\ket{0}$ and if it aborts, the entire state is set to $\ket{\bot}$. We have the guarantee that $\ket{0}\ket{\psi}$ is orthogonal to $\ket{\bot}$. In the next stage, the PRFS generator traces out the first qubit and outputs the resulting state. Our definition could be useful to capture many generators that don't always succeed in generating the pseudorandom state; for example, Brakerski and Shmueli~\cite{BrakerskiS20} design generators that doesn't always succeed in generating the state. 
\par We formally define the notion of recognizable abort\footnote{We note that~\cite{AQY21} define a slightly weaker definition of recognizable abort. However, the definitions and results considered in~\cite{AQY21} also work with our (stronger) definition of recognizable abort. \luowen{I still don't buy this. If $\ket\bot$ is not a known state, it wouldn't be called ``recognizable''. That is what I meant to write in AQY anyways.} } below. 

\begin{definition}[Recognizable abort]
\label{def:classicalgen:strongrec}
 A $(d(\secparam),n(\secparam))$-PRFS generator $G$ has the \emph{strongly recognizable abort property} if there exists an algorithm $\widehat{G}$ and a special $(n(\secparam)+1)$-qubit state $\ket{\bot}$ such that $G_{\secparam}(k,x)$ has the following form: it takes as input $k \in \bit^\lambda$, $x \in \bit^{d(\secparam)}$ and does the following,
 \begin{itemize}
     \item Compute $\widehat{G}_{\secparam}(k,x)$ to obtain an output of the form $\eta \left(\ketbra{0}{0} \otimes \ketbra{\psi}{\psi}\right) + (1-\eta) \ketbra{\bot}{\bot}$ and moreover, $\ket{\bot}$ is of the form $\ket{1}\ket{\widehat{\bot}}$ for some $n(\secparam)$-qubit state state $\ket{\widehat{\bot}}$. As a consequence, $(\bra{0} \otimes \bra{\psi})(\ket{\bot})=0$. 
     \item Trace out the first bit of $\widehat{G}_{\secparam}(k,x)$ and output the resulting state. 
 \end{itemize}
\end{definition}

\noindent As observed by~\cite{AQY21}, the definition alone does not have any constraint on $\eta$ being close to $1$. The security guarantee of a PRFS generator implies that $\eta$ will be negligibly close to 1 with overwhelming probability over the choice of $k$ \cite[Lemma 3.6]{AQY21}.

\section{Adaptive Security}
\label{sec:adaptive-security}

\noindent The previous work by~\cite{AQY21} only considers PRFS that is selectively secure. That is, the adversary needs to declare the input queries ahead of time. For many applications, selective security is insufficient. For example, in the application of PRFS to secret-key encryption (satisfying multi-message security), the resulting scheme was also only proven to be selectively secure, whereas one could ask for security against adversaries that can make \emph{adaptive} queries to the PRFS generator. Another drawback of the notion considered by~\cite{AQY21} is the assumption that the adversary can make classical queries to the challenger who either returns PRFS states or independent Haar random states, whereas one would ideally prefer security against adversaries that can make \emph{quantum superposition} queries.
%Another drawback of the notion considered by~\cite{AQY21} is the assumption that the adversary can only make classical queries to the oracle, \hnote{added: } whereas one would ideally prefer security against adversaries that can make \emph{quantum superposition} queries. \pnote{to avoid ambiguity about oracles, I modified the above sentence 

\par \hnote{Edited:} In this work, we consider stronger notions of security for PRFS. We strengthen the definitions of~\cite{AQY21} in two ways. First, we allow the the adversary to make adaptive queries to the PRFS oracle, and second, we allow the adversary to make \emph{quantum} queries to the oracle. The oracle model we consider here is slightly different from the usual quantum query model. In the usual model, there is an underlying function $f$ and the oracle is modelled as a unitary acting on two registers, a \emph{query} register $\bfX$ and an \emph{answer} register $\bfY$ mapping basis states $\ket{x}_{\bfX} \otimes \ket{y}_{\bfY}$ to $\ket{x}_{\bfX} \otimes \ket{y \oplus f(x)}_{\bfY}$ (in other words, the function output is XORed with answer register in the standard basis). The query algorithm also acts on the query and answer registers; indeed, it is often useful in quantum algorithms to initialize the answer register to something other than all zeroes. 

In the PRS/PRFS setting, however, there is no underlying classical function: the output of the PRFS generator $G$ could be an entangled pseudorandom state far from any standard basis state; it seems unnatural to XOR the pseudorandom the state with a standard basis state. Instead we consider a model where the query algorithm submits a query register $\bfX$ to the oracle, and the oracle returns the query register $\bfX$ as well as an answer register $\bfY$. If the algorithm submits query $\ket{x}_{\bfX}$, then the joint state register $\bfX \bfY$ after the query is $\ket{x}_{\bfX} \otimes \ket{\psi_x}_{\bfY}$ for some pure state $\ket{\psi_x}$. Each time the algorithm makes a query, the oracle returns a fresh answer register. Thus, the number of qubits that the query algorithm acts on grows with the number of queries.\footnote{Alternatively, one can think of answer registers $\bfY_1,\bfY_2,\ldots$ as being initialized in the zeroes state at the beginning, and the query algorithm is only allowed to act nontrivially on $\bfY_i$ after the $i$'th query.}
\luowen{Should mention the word ``isometry'', which summarizes this whole paragraph} \hnote{Whether it's an isometry depends on whether it's quantum or classical accessible. added a footnote below}

How the oracle behaves when the query algorithm submits a superposition $\sum_x \alpha_x \ket{x}_{\bfX}$ in the query register is a further modeling choice. In the most general setting, the oracle behaves as a unitary on registers $\bfX \bfY$,\footnote{Alternatively, one can think of the oracle as an \emph{isometry} mapping register $\bfX$ to registers $\bfX \bfY$.} and the resulting state of the query and answer registers is $\sum_x \alpha_x \ket{x}_{\bfX} \otimes \ket{\psi_x}_{\bfY}$. That is, queries are answered in superposition. We call such an oracle \emph{quantum-accessible}. 
\luowen{Should also mention at some point the ability to uncompute is not considered in this work and left for future work (but construction wise everything in Appendix A will extend immediately to that setting as well, given that the unitary is properly defined)} \hnote{The way it's described here, it's not clear what it means to invert. Does the algorithm get access to the $\bfY$ register after uncomputing?}

We also consider the case where the queries are forced to be \emph{classical}, which may already be useful for some applications. Here, the oracle is modeled as a channel (instead of a unitary) that first measures the query register in the standard basis before returning the corresponding state $\ket{\psi_x}$. In other words, if the query is $\sum_x \alpha_x \ket{x}_{\bfX}$, then the resulting state becomes the mixed state $\sum_x |\alpha_x|^2 \ketbra{x}{x}_{\bfX} \otimes \ketbra{\psi_x}{\psi_x}_{\bfY}$. This way, the algorithm cannot take advantage of quantum queries -- but it can still make queries adaptively. We call such an oracle \emph{classically-accessible}. 

To distinguish between classical and quantum access to oracles, we write $A^{\cal O}$ to denote a quantum algorithm that has classical access to the oracle $\cal O$, and $A^{\ket{\cal O}}$ to denote a quantum algorithm that has quantum access to the oracle $\cal O$.

\subsection{Classical Access}
\noindent We define adaptively secure PRFS, where the adversary is given \emph{classical access} to the PRFS/Haar-random oracle.
%we still restrict the adversary to make classical queries to the oracle. On the other hand, we do allow the adversary to make adaptive queries. 

\begin{definition}[Adaptively-Secure PRFS]
\label{def:aprfs}
We say that a QPT algorithm $G$ is an {\em adaptively secure pseudorandom function-like state (APRFS)} generator if for all QPT (non-uniform) distinguishers $A$, there exists a negligible function $\eps$, such that for all $\lambda$, the following holds:
 \[
    \left|\Pr_{k \leftarrow \{0,1\}^{\secparam} }\left[A_\lambda^{{\cal O}_{\sf PRFS}(k,\cdot)}(\rho_\lambda) = 1\right] - \Pr_{{\cal O}_{\sf Haar}}\left[A_\lambda^{{\cal O}_{\sf Haar}(\cdot)}(\rho_\lambda) = 1\right]\right| \le \eps(\lambda),
  \]
  where:
  \begin{itemize}
      \item ${\cal O}_{\sf PRFS}(k,\cdot)$, %modeled as a channel, 
      on input $x \in \{0,1\}^{d(\secparam)}$,  outputs $G_{\secparam}(k,x)$.
      \item ${\cal O}_{\sf Haar}(\cdot)$, %modeled as a channel, 
      on input $x \in \{0,1\}^{d(\secparam)}$, outputs $\ket{\vartheta_x}$, where, for every $y \in \{0,1\}^{d(\secparam)}$,  $\ket{\vartheta_y} \leftarrow \Haar_{n(\lambda)}$. %\luowen{From this notation, it is unclear if a new state is sampled every time the channel is invoked.} \pnote{I did take this into consideration and I think as written it should not mean what you said.}
  \end{itemize}
Moreover, the adversary $A_{\secparam}$ has classical access to ${\cal O}_{\sf PRFS}(k,\cdot)$ and ${\cal O}_{\sf Haar}(\cdot)$. That is, we can assume without loss of generality that any query made to either oracle is measured in the computational basis. 
\par We say that $G$ is a $(d(\lambda),n(\lambda))$-APRFS generator to succinctly indicate that its input length is $d(\lambda)$ and its output length is $n(\lambda)$.
\end{definition}

\noindent Some remarks are in order. 

\paragraph{Instantiation.} For the case when $d(\secparam)=O(\log(\secparam))$, selectively secure PRFS is equivalent to adaptively secure PRFS. The reason being that we can assume without loss of generality, the selective adversary can query on all possible inputs (there are only polynomially many) and use the outputs to simulate the adaptive adversary. As a consequence of the result that log-input selectively-secure PRFS can be built from PRS~\cite{AQY21}, we obtain the following. 

\begin{lemma}
For $d=O(\log(\secparam))$ and $n=d+\omega(\log \log \secparam)$, assuming the existence of $(d+n)$-PRS, there exists a $(d,n)$-APRFS.   
\end{lemma}

\noindent In the case when $d(\secparam)$ is an arbitrary polynomial in $\secparam$, we present a construction of APRFS from post-quantum one-way functions in~\Cref{sec:classical:aprfs}. 

\paragraph{Test procedure.} It was shown by~\cite{AQY21} that a PRFS admits a Test procedure (See Section 3.3 in~\cite{AQY21}). The goal of a Test procedure is to determine whether the given state is a valid PRFS state or not. Having a Test procedure is useful in applications. For example,~\cite{AQY21} used a Test procedure in the construction of a bit  commitment scheme. We note that the same Test procedure also works for adaptively secure PRFS.

\paragraph{Multiple copies.} In the definition of PRS (\Cref{def:vanilla-prs}) and selectively-secure PRFS (\Cref{def:prfs}), the adversary is allowed to obtain multiple copies of the same pseudorandom (or haar random) quantum state. While we do not explicitly state it, even in~\Cref{def:aprfs}, the adversary can indeed obtain multiple copies of a (pseudorandom or haar random) quantum state. To obtain $t$ copies of the output of $G_{\secparam}(k,x)$ (or $\ket{\vartheta_x}$), the adversary can query the same input $x$, $t$ times, to the oracle ${\cal O}_{\sf PRFS}(k,\cdot)$ (or ${\cal O}_{\sf Haar}(\cdot)$).

\subsection{Quantum Access}
\noindent We further strengthen our notion of adaptively secure PRFS by allowing the adversary to make superposition queries to either ${\cal O}_{\sf PRFS}(k,\cdot)$ or ${\cal O}_{\sf Haar}(\cdot)$. Providing superposition access to the adversary not only makes the definition stronger\footnote{It is stronger in the sense that an algorithm that has quantum query access to the oracle can simulate an algorithm that only has classical query access.} than~\Cref{def:aprfs} but is also arguably more useful for a larger class of applications. To indicate quantum query access, we put the oracle inside the ket notation: $A^{\ket{\cal O}}$ (whereas for classical query access we write $A^{\cal O}$). 

\par We provide the formal definition below. %In the definition below, the adversary has quantum access to either the oracle that computes the PRFS honestly or the oracle that outputs a Haar random state. In either of the two cases, the adversary submits a register ${\bf X}$, the oracle controlled on this register either computes the PRFS output or the sampled Haar random state and stores the resulting state in a new register ${\bf Y}$. It sends back both the registers ${\bf X}$ and ${\bf Y}$. The goal of the adversary is to distinguish whether the oracle is computing the PRFS honestly or whether it outputs a Haar random state.  

\begin{definition}[Quantum-accessible Adaptively-secure PRFS]
\label{def:qaprfs}
We say that a QPT algorithm $G$ is a  quantum-accessible adaptively secure pseudorandom function-like state (QAPRFS) generator if for all QPT (non-uniform) distinguishers $A$ if there exists a negligible function $\eps$, such that for all $\lambda$, the following holds:
 \[
    \left|\Pr_{k \leftarrow \{0,1\}^{\secparam} }\left[A_\lambda^{\ket{{\cal O}_{\sf PRFS}(k,\cdot)}}(\rho_\lambda) = 1\right] - \Pr_{{\cal O}_{\sf Haar}}\left[A_\lambda^{\ket{{\cal O}_{\sf Haar}(\cdot)}}(\rho_\lambda) = 1\right]\right| \le \eps(\lambda),
  \]
  where:
\begin{itemize}
    \item ${\cal O}_{\sf PRFS}(k,\cdot)$, %modeled as a channel, 
    on input a $d$-qubit register ${\bf X}$, does the following: it  applies a \luowen{isometry} channel that controlled on the register ${\bf X}$ containing $x$, it creates and stores  $G_{\secparam}(k,x)$ in a new register ${\bf Y}$.  It outputs the state on the registers ${\bf X}$ and ${\bf Y}$. 
    \item ${\cal O}_{\sf Haar}(\cdot)$, modeled as a  channel, on input a $d$-qubit register ${\bf X}$, does the following: it applies a channel that controlled on the register ${\bf X}$ containing $x$, stores $\ketbra{\vartheta_x}{\vartheta_x}$ in a new register ${\bf Y}$, where $\ket{\vartheta_x}$ is sampled from the Haar distribution. It outputs the state on the registers ${\bf X}$ and ${\bf Y}$. 
\end{itemize}
Moreover, $A_{\secparam}$ has superposition access to ${\cal O}_{\sf PRFS}(k,\cdot)$ and ${\cal O}_{\sf Haar}(\cdot)$. %We denote the fact that $A_{\secparam}$ has quantum access to an oracle ${\cal O}$ by $A_{\secparam}^{\ket{{\cal O}}}$. 
\par We say that $G$ is a $(d(\lambda),n(\lambda))$-QAPRFS generator to succinctly indicate that its input length is $d(\lambda)$ and its output length is $n(\lambda)$.
\end{definition}

\noindent We present a construction satisfying the above definition in~\Cref{sec:cons:qrfs}.  
\par Unlike~\Cref{def:aprfs}, it is not without loss of generality that $A_{\secparam}$ can get multiple copies of a quantum state. To illustrate, consider an adversary that submits a state of the form $\sum_{x} \alpha_x \ket{x}$ to the oracle. It then gets back $\sum_{x} \alpha_x \ket{x} \ket{\psi_x}$ (where $\ket{\psi_x}$ is either the output of PRFS\footnote{In this illustration, we are pretending that the PRFS satisfies perfect state generation property. That is, the output of PRFS is always a pure state.} or it is Haar random) instead of $\sum_{x} \alpha_x \ket{x} \ket{\psi_x}^{\otimes t}$, for some polynomial $t$. On the other hand, if the adversary can create multiple copies of $\sum_{x} \alpha_x \ket{x}$, the above definition allows the adversary to obtain $\left( \sum_{x} \alpha_x \ket{x} \ket{\psi_x}\right)^{\otimes t}$ for any  polynomial $t(\cdot)$ of its choice.

\section{Constructions of Pseudorandom Primitives}
We present improved and/or new constructions of pseudorandom primitives. 
\begin{enumerate}
    \item In~\Cref{sec:simpleranalysis:prs}, we present a simpler proof of binary phase PRS. The current known proof of binary phase PRS by Brakerski and Shmueli~\cite{BS19} is arguably more involved.
    \item In~\Cref{sec:classical:aprfs}, we present a construction of APRFS (\Cref{def:aprfs}) from post-quantum one-way functions. Recall that in this definition, the adversary only has classical access to the oracle. 
    \item In~\Cref{sec:cons:qrfs}, we present two constructions of QAPRFS (\Cref{def:qaprfs}) from post-quantum one-way functions. Recall that in this definition, the adversary has quantum access to the oracle. %\hnote{This last sentence seems to be a typo?} \pnote{corrected}
%    \item \hnote{Added:} In~\Cref{sec:cons:qrfs2}, we present an alterate construction of QAPRFS from polynomially-secure post-quantum one-way functions. 
    %\item Finally, in~\Cref{sec:stretchprs}, we present a construction of stretch PRS (\Cref{def:stretchPRS}) from QAPRFS. Previously, stretch PRS was known based on the existence of post-quantum one-way functions. 
\end{enumerate} 

\subsection{Simpler Analysis of Binary Phase PRS}
\label{sec:simpleranalysis:prs}

In this section we give a simpler analysis\footnote{The current analysis presented here making use of binary
type vectors was suggested by Fermi Ma. Our original proof involved a
more direct analysis of the symmetric subspace projector.} of the binary phase PRS construction suggested by~\cite{ji2018pseudorandom} and later analyzed by~\cite{BS19}. 
%In fact, we prove something stronger: we show that if $ G: \{0,1\}^\secparam \times \{0,1\}^n \to \C^{2^n}$ is a $(d,n)$-PRFS generator \hnote{need to specify quantum accessible?, and also that it outputs a pure state?}, then the following states are pseudorandom:
%\begin{equation}
%     \label{eq:prfs-to-prs}
%     \ket{\psi_k} = 2^{-n/2} \sum_{x \in \{0,1\}^n} \ket{x} \otimes \ket{G(k,x)}~.
% \end{equation}
% Note that we have achieved a limited form of stretch: the given PRFS generator $G$ outputs pseudorandom states on $n$ qubits, and the states defined above are on $n + d$ qubits.  \hnote{say something about here we don't assume variable length PRFS?}

% This subsumes the binary phase construction because we can view a PRF as a $(d,n)$-PRFS generator with $n = 0$. \hnote{Does this need more elaboration?} In other words, the output of the PRFS generator is a scalar value. 

\begin{theorem}
\label{thm:prfs-to-prs}
Let $F: \{0,1\}^\secparam \times \{0,1\}^n \to \{0,1\}$ is a quantum-query secure PRF (\Cref{def:qqprfs}), then $G: \{0,1\}^\secparam \to \C^{2^n}$ defined by $G(k) = \ket{\psi_k} = 2^{-n/2} \sum_x (-1)^{F(k,x)} \ket{x}$ is a $n$-qubit PRS generator.
\end{theorem}

To prove this, we establish some basic facts about the \emph{symmetric subspace}. 

\paragraph{The Symmetric Subspace and Its Properties.} The symmetric subspace of the tensor product space $(\C^N)^{\otimes t}$ is the subspace of states invariant under all permutations of the $t$ tensor factors. We write $\Pi^{N,t}_\sym$ to denote the projector onto the symmetric subspace of $(\C^N)^{\otimes t}$. The next fact gives an equivalent definition of the projector. For proofs of these facts, please see~\cite{harrow2013church}. 

\begin{fact}
\label{fact:sym}
For a permutation $\sigma \in S_t$, let $P_{N}(\sigma)$ denote the permutation on $(\C^N)^{\otimes t}$ that permutes the $t$ tensor factors according to $\sigma$. Hence, $$P_N(\sigma) = \sum_{x_1,\ldots,x_t\in[N]}\ketbra{x_{\sigma^{-1}(1)},\ldots,x_{\sigma^{-1}(t)}}{x_1,\ldots,x_t}.$$ Then we have
\[
    \Pi^{N,t}_\sym = \frac{1}{t!} \sum_{\sigma \in S_t} P_{N}(\sigma)~.
\]
\end{fact}

\begin{fact}[Average of copies of Haar-random states]
\label{fact:avg-haar-random}
For all $N,t \in \N$, we have
\[
    \E_{\ket{\vartheta} \leftarrow \Haar(\C^N)} \ketbra{\vartheta}{\vartheta}^{\otimes t} = \frac{\Pi^{N,t}_\sym}{\Tr(\Pi^{N,t}_\sym)}~.
\]
\end{fact}

\begin{fact}
\label{fact:trace-sum}
    Let $\rho_1,\rho_2$ be density matrices such that $\rho_2 = \alpha\rho_1 + \beta\rho_1^{\perp}$ where $\rho_1\rho_1^{\perp} = 0$, $\rho_1^{\perp}\rho_1 = 0$ and $\alpha,\beta\in[0,1]$, $\alpha + \beta = 1$, then $$\TD\left(\rho_1,\rho_2\right)=\beta.$$
\end{fact}
\begin{proof}
    \begin{align*}
        \TD\left(\rho_1,\rho_2\right) &= \frac{1}{2}\Tr\left(\sqrt{\left(\rho_1-\rho_2\right)^2}\right) \\
        &= \frac{1}{2}\Tr\left(\sqrt{\left(\rho_1-\left(\alpha\rho_1 + \beta\rho_1^{\perp}\right)\right)^2}\right) \\
        &= \frac{1}{2}\Tr\left(\sqrt{\left((1-\alpha)\rho_1- \beta\rho_1^{\perp}\right)^2}\right) \\
        &= \frac{1}{2}\Tr\left(\sqrt{(1-\alpha)^2\rho_1^2+ \beta^2\left(\rho_1^{\perp}\right)^2}\right) \\
    \end{align*}
    Here, the first equality is from definition of trace distance, the second equality is by definition of $\rho_2$, the third equality is by simplification, and the fourth equality is because $\rho_1\rho_1^{\perp} = 0$, $\rho_1^{\perp}\rho_1 = 0$.

    Since, $\rho_1\rho_1^{\perp} = 0$, $\rho_1^{\perp}\rho_1 = 0$, we can add $(1-\alpha)\beta\rho_1\rho_1^{\perp}+ (1-\alpha)\beta\rho_1^{\perp}\rho_1$ without changing the value.
    \begin{align*}
        \TD\left(\rho_1,\rho_2\right) &= \frac{1}{2}\Tr\left(\sqrt{(1-\alpha)^2\rho_1^2+ \beta^2\left(\rho_1^{\perp}\right)^2 + (1-\alpha)\beta\rho_1\rho_1^{\perp}+ (1-\alpha)\beta\rho_1^{\perp}\rho_1}\right) \\
        &= \frac{1}{2}\Tr\left(\sqrt{\left((1-\alpha)\rho_1+ \beta\rho_1^{\perp}\right)^2}\right) \\
        &= \frac{1}{2}\Tr\left((1-\alpha)\rho_1+ \beta\rho_1^{\perp}\right) \\
        &= \frac{1}{2}\left((1-\alpha)\Tr(\rho_1)+ \beta\Tr\left(\rho_1^{\perp}\right)\right) \\
        &= \frac{1}{2}\left((1-\alpha)+ \beta\right) \\
        &= \beta \\
    \end{align*}
    Here, the first line is since $\rho_1\rho_1^{\perp} = 0$, $\rho_1^{\perp}\rho_1 = 0$, the second line is by simplification, the third equality is since $(1-\alpha)\rho_1+ \beta\rho_1^{\perp}$ is positive semidefinite, the fourth equality is by linearity of trace, the fifth equality is by $\Tr(\rho) = 1$ and $\Tr(\rho^{\perp}) = 1$, and the last equality is by $\alpha + \beta = 1$.
\end{proof}

We define some new notation to look at a different charecterisation of the symmetric space.

\begin{definition}
    Let $v\in [N]^t$ for some $N,t\in \N$, then define $type(v)$ to be a vector in $[t+1]^N$ where the $i^{th}$ entry in $type(v)$ denotes the frequency of $i$ in $v$.
\end{definition}

\begin{definition}
\label{def:type_states}
    Let $T\in [t+1]^N$ for some $N,t\in \N$, then define $$\ket{type_T} = \beta\sum_{\substack{v\in[N]^t\\ type(v) = T}}\ket{v},$$ where $\beta\in\R$ is an appropriately chosen constant. Similarly, for $T\in \{0,1\}^N$, define $$\ket{bintype_T} = \beta\sum_{\substack{v\in[N]^t\\ type(v) \pmod{2} = T}}\ket{v},$$ where $\beta\in\R$ is an appropriately chosen constant.
\end{definition}

Note that if $hamming(T) = t$ and $T\in \{0,1\}^N$, then $\ket{type_T} = \ket{bintype_T}$.

\begin{lemma}
\label{lem:haar_type}
For all $N,t \in \N$, we have
    $$\TD\left(\E_{\substack{T\leftarrow \{0,1\}^N\\ hamming(T)=t}} \ketbra{type_T}{type_T}, \E_{\ket{\vartheta} \leftarrow \Haar(\C^N)} \ketbra{\vartheta}{\vartheta}^{\otimes t}\right)\leq O\left(\frac{t^2}{N}\right).$$
\end{lemma}

\begin{proof}
    From \Cref{fact:avg-haar-random}, 
    \begin{align*}
        \rho = \E_{\ket{\vartheta} \leftarrow \Haar(\C^N)} \ketbra{\vartheta}{\vartheta}^{\otimes t} &= \frac{\Pi^{N,t}_\sym}{\Tr(\Pi^{N,t}_\sym)} \\
        &=\frac{1}{t!\Tr(\Pi^{N,t}_\sym)} \sum_{\sigma \in S_t} P_{N}(\sigma) \\
        &= \frac{1}{t!\Tr(\Pi^{N,t}_\sym)} \sum_{\sigma \in S_t} \sum_{x_1,\ldots,x_t\in[N]}\ketbra{x_{\sigma^{-1}(1)},\ldots,x_{\sigma^{-1}(t)}}{x_1,\ldots,x_t} 
    \end{align*}
    where the second and the third line follow from~\Cref{fact:sym}. From~\Cref{def:type_states}, 
    \begin{align*}
        \sigma = \E_{\substack{T\leftarrow \{0,1\}^N\\ hamming(T)=t}} \ketbra{type_T}{type_T} &= \E_{\substack{T\leftarrow \{0,1\}^N\\ hamming(T)=t}} \left(\frac{1}{\sqrt{t!}}\sum_{\substack{v\in[N]^t\\ type(v) = T}}\ket{v}\right)\left(\frac{1}{\sqrt{t!}}\sum_{\substack{v'\in[N]^t\\ type(v') = T}}\bra{v'}\right) \\
        &= \frac{1}{t!} \E_{\substack{T\leftarrow \{0,1\}^N\\ hamming(T)=t}} \left(\sum_{\substack{v,v'\in[N]^t\\ type(v) = type(v') = T}}\ketbra{v}{v'}\right) \\
        &= \frac{1}{t!} \E_{\substack{T\leftarrow \{0,1\}^N\\ hamming(T)=t}} \left(\sum_{\substack{v\in[N]^t\\ type(v) = T}}\sum_{\substack{\sigma\in S_t\\ v' = \sigma(v)}}\ketbra{v}{v'}\right) \\
        &= \frac{1}{t! {N \choose t}} \sum_{\substack{v\in[N]^t\\ type(v) \in \{0,1\}^N }}\sum_{\sigma\in S_t}\ketbra{v}{\sigma(v)} \\
        &= \frac{1}{t! {N \choose t}} \sum_{\substack{x_1,\ldots,x_t\in[N]\\ x_1,\ldots,x_t \text{ are distinct}}}\sum_{\sigma\in S_t}\ketbra{x_1,\ldots,x_t}{x_{\sigma(1)},\ldots,x_{\sigma(t)}}
    \end{align*}
    where the third line follows by re-interpreting vector of same type as permutation of each other, the fourth line follows by taking expectation (since there are a total of ${N \choose t}$ strings of hamming weight of $t$ in $\{0,1\}^N$). The fifth line follows since $type(v)\in\{0,1\}^N$, hence all elements of $v$ are distinct.

    Define $$\sigma^{\perp} = \frac{1}{t! \left({N + t - 1\choose t} - {N \choose t}\right)} \sum_{\substack{x_1,\ldots,x_t\in[N]\\ x_1,\ldots,x_t \text{ are not distinct}}}\sum_{\sigma\in S_t}\ketbra{x_1,\ldots,x_t}{x_{\sigma(1)},\ldots,x_{\sigma(t)}}$$

    Hence, $\rho = \alpha\sigma + \beta\sigma^{\perp}$, where $\beta$ is probability of picking $x_1,\ldots,x_t\in[N]$ such that there is a colision which is less than $O\left(\frac{t^2}{N}\right)$. 

    Hence, from \Cref{fact:trace-sum}, we see $$\TD\left(\E_{\substack{T\leftarrow \{0,1\}^N\\ hamming(T)=t}} \ketbra{type_T}{type_T}, \E_{\ket{\vartheta} \leftarrow \Haar(\C^N)} \ketbra{\vartheta}{\vartheta}^{\otimes t}\right)\leq O\left(\frac{t^2}{N}\right).$$
\end{proof}

\noindent We prove \Cref{thm:prfs-to-prs} via a hybrid argument. Let $t$ be a polynomial in $\secparam$. %we specify several hybrids for sampling a $n$-qubit state $\ket{\psi}$.

\paragraph{Hybrid 1.}  Sample a random $k \leftarrow \{0,1\}^\secparam$. Let $\ket{\psi} = \ket{\psi_k}$ as defined in \Cref{eq:binary-phase} and output $\ket{\psi}^{\otimes t}$.

\paragraph{Hybrid 2.} For all $x \in \{0,1\}^n$, sample $\alpha_x \in \{\pm 1\}$ uniformly at random. Let
\[
    \ket{\psi} = 2^{-n/2} \sum_{x \in \{0,1\}^n} \alpha_x \ket{x}~.
\]
Output $\ket{\psi}^{\otimes t}$.

\paragraph{Hybrid 3.} Sample $w \in [N]^t$ uniformly at random. Let $T = type(w) \pmod{2}$.
Output $\ket{bintype_T}^{\otimes t}$.

\paragraph{Hybrid 4.} Sample $T \in \{0, 1\}^N$ with $hamming(T) = t$ uniformly at random. Output $\ket{type_T}^{\otimes t}$.

\paragraph{Hybrid 5.} Sample a Haar-random $n$-qubit state $\ket{\psi}$. Output $\ket{\psi}^{\otimes t}$.

\medskip

The Hybrids 1 and 2 are computationally indistinguishable since the PRF $F$ cannot be distinguished by a quantum adversary (that can make superposition queries to $F_k(\cdot) = F(k,\cdot)$) from a random function, from the quantum query security in~\Cref{def:qqprfs}. 

The rest of the proof can be seen as the following lemma.

\begin{lemma}
\label{lem:binaryphase}
Fix $t,n \in \N$. Let $\{ \alpha_x \}_x$ be independent and uniformly random $\pm 1$ values. 
Define
\[
\ket{\psi} = 2^{-n/2} \sum_{x \in \{0,1\}^n} \alpha_x \, \ket{x}~.
\]
Then
\begin{equation}
\label{eq:td1}
	\TD \Big (\,  \E  \ketbra{\psi}{\psi}^{\otimes t} \, , \, \E  \ketbra{\vartheta}{\vartheta}^{\otimes t} \, \Big ) \leq O(t^2/2^{n})
\end{equation}
where $\ket{\vartheta}$ is a Haar-random $n$-qubit state.
\end{lemma}

A version of this Lemma is also proven by Brakerski and Shmueli~\cite{BS19}, we give an alternate and a more straightforward proof. 

\begin{lemma}
    Hybrids 2 and 3 (from above) are identical.
\end{lemma}
\begin{proof}
    The output of Hybrid 2 is as follows:
    \begin{align*}
        \rho &= \E_{\alpha_x} \left(\frac{1}{N^{t/2}}\sum_{x_1,\ldots,,x_t\in[N]}\alpha_{x_1}\ldots\alpha_{x_t}\ket{x_1,\ldots,x_t}\right)\left(\frac{1}{N^{t/2}}\sum_{y_1,\ldots,y_t\in[N]}\alpha_{y_1}\ldots\alpha_{y_t}\bra{y_1,\ldots,y_t}\right) \\
        &= \frac{1}{N^{t}}\E_{\alpha_x} \left(\sum_{\substack{x_1,\ldots,x_t\in[N] \\ y_1,\ldots,y_t\in[N]}}\alpha_{x_1}\ldots\alpha_{x_t}\alpha_{y_1}\ldots\alpha_{y_t}\ketbra{x_1,\ldots,x_t}{y_1,\ldots,y_t}\right) \\
        &= \frac{1}{N^{t}}\left(\sum_{\substack{x_1,\ldots,x_t\in[N] \\ y_1,\ldots,y_t\in[N]}}\sum_{type(x)\bmod{2} = type(y)\bmod{2}}\ketbra{x_1,\ldots,x_t}{y_1,\ldots,y_t}\right)
    \end{align*}
    where the second line follows by rearranging terms and third line follows from the fact that taking expectation over any unpaired $\alpha_x$ results in zero.

    The output of Hybrid 3 is as follows:
    \begin{align*}
        \rho' &= \E_{w\in [N]^t} \left(\sum_{\substack{v\in[N]^t \\ type(v)\bmod{2} = type(w)\bmod{2}}}\ket{v}\right)\left(\sum_{\substack{v'\in[N]^t \\ type(v')\bmod{2} = type(w)\bmod{2}}}\bra{v'}\right) \\
        &= \frac{1}{N^{t}} \left(\sum_{\substack{v,v'\in[N]^t \\ type(v)\bmod{2} = type(v')\bmod{2}}}\ketbra{v}{v'}\right) 
    \end{align*}
    where the second line follows because there are $N^t$ options for $w$. From above we can see that the output of Hybrid 2 and Hybrid 3 are exactly the same.
\end{proof}

Let $w\in [N]^t$ be sampled uniformly at random and $T = type(w) \pmod{2}$, then $hamming(T) = t$ with probability $1-t^2/N$ as $w$ has no collisions with probabilty $1-t^2/N$. Hence, the trace distance between the outputs of Hybrids 3 and 4 is $O(t^2/N)$.

By \Cref{lem:haar_type}, the trace distance between the outputs of Hybrids 4 and 5 is $O(t^2/N)$.

Combining our bounds on above, we get that
\[
	\TD \Big (    \ketbra{\psi}{\psi}^{\otimes t} \, , \, \E  \ketbra{\vartheta}{\vartheta}^{\otimes t} \Big ) \leq O(t^2/N)
\]
as desired.

\subsection{(Classically-Accessible) Adaptively-Secure PRFS}
\label{sec:classical:aprfs}
We construct an adaptively secure PRFS where the adversary is only allowed classical access in the security experiment. 
\luowen{Wouldn't all of this be an immediate corollary of the next section? What is the point?} \hnote{Added this comment:} While the quantum-accessible APRFS constructed and analyzed in \Cref{sec:cons:qrfs} are also classically-accessible; their analyses are much more involved. Here we present a construction and analysis of a classically-accessible APRFS that is simpler, and may be a helpful starting point for applications that only require classically-accessible PRFS. 

\par To construct a $(d(\secparam),n(\secparam))$-APRFS $G$, we start with the following two primitives:
\begin{itemize}
    \item $(d(\secparam),\secparam)$-post-quantum secure pseudorandom function $F$ (\Cref{def:pqprfs}) and, 
    \item $n(\secparam)$-pseudorandom state generator $g$ (\Cref{def:vanilla-prs}). 
\end{itemize}

\paragraph{Construction.} We describe $G_{\secparam}$ as follows: on input $k \in \{0,1\}^{\secparam}$, $x \in \{0,1\}^{d(\secparam)}$, 
\begin{itemize}
    \item $k_x \leftarrow F(k,x)$,
    \item $\rho \leftarrow g(k_x)$,
    \item Output $\rho$. 
\end{itemize}

\begin{lemma}
Assuming the post-quantum security of the pseudorandom function $F$, the QPT procedure $G$ described above is a secure $(d(\secparam),n(\secparam))$-APRFS (Definition~\ref{def:aprfs}). 
\end{lemma}
\begin{proof}
We prove this by a standard hybrid argument. Let $A_{\secparam}$ be a QPT distinguisher that distinguishes the oracles ${\cal O}_{\sf PRFS}(k,\cdot)$, where $k \leftarrow \{0,1\}^{\secparam}$, and ${\cal O}_{\sf Haar}(\cdot)$ with probability $\eps$. We prove that $\eps$ is negligible. Let $q$ be the number of queries made by $A_{\secparam}$. For simplicity, we assume below that all the queries are distinct. The same proof generalizes to the setting when the queries are not distinct.  

\paragraph{Hybrid $H_{1}$.} Output $A_{\secparam}^{{\cal O}_{\sf PRFS}(k,\cdot)}$, where $k \leftarrow \{0,1\}^{\secparam}$. 

\paragraph{Hybrid $H_{2}$.} Output $A_{\secparam}^{{\cal O}_{H_2}(k,\cdot)}$, where $k \leftarrow \{0,1\}^{\secparam}$, where ${\cal O}_{H_2}(k,\cdot)$ is defined as follows.
\par On input $x \in \{0,1\}^{d(\secparam)}$, compute $\rho \leftarrow g(k_x)$, where $k_x$ is sampled uniformly at random. Output $\rho$. 
\par The computational indistinguishability of $H_1$ and $H_2$ follows from the post-quantum security of $F$.

\paragraph{Hybrid $H_{3.i}$ for $i \in [q]$.} Output $A_{\secparam}^{{\cal O}_{H_{3.i}}(k,\cdot)}$, where $k \leftarrow \{0,1\}^{\secparam}$, where ${\cal O}_{H_{3.i}}(k,\cdot)$ is defined as follows.
\par On input $x \in \{0,1\}^{d(\secparam)}$, do the following: if this is the $j^{th}$ query, where $j \leq i$, then output $\ket{\vartheta_x}$, where $\ket{\vartheta_x}$ is sampled from the Haar distribution and if $j > i$, output $\rho \leftarrow g(k_x)$, where $k_x$ is sampled uniformly at random.  
\begin{claim}
Assuming the security of PRS generator $g$, for every $i \in [q]$, $A_{\secparam}$ can distinguish hybrids $H_{3.i-1}$ (we set $H_{3.0}=H_2$) and $H_{3.i}$ only with negligible probability.
\end{claim}
\begin{proof}
Suppose $A_{\secparam}$ can distinguish the hybrids $H_{3.i-1}$  and $H_{3.i}$ with probability $\eps_3(\secparam)$. We construct an adversary ${\cal B}$ that violates the security of $g$ with probability $\eps_3(\secparam)$. Then we invoke the security of $g$ to conclude that $\eps_3(\secparam)$ has to be negligible in $\secparam$. First, we let ${\cal B}$ be inefficient and later, we remark how to make ${\cal B}$ efficient. 
\par ${\cal B}$ gets as input a state $\rho$. ${\cal B}$ then runs $A_{\secparam}$ while  simulating the oracle $A_{\secparam}$ has access to. For the $j^{th}$ query, say $x \in \{0,1\}^{d(\secparam)}$, where $j < i$, it outputs $\ketbra{\vartheta_x}{\vartheta_x}$, where  $\ket{\vartheta_x} \leftarrow \Haar_{n(\secparam)}$. For the $i^{th}$ query, it outputs $\rho$. For the $j^{th}$ query, say $x \in \{0,1\}^{d(\secparam)}$, where $j > i$, it outputs $\rho_{k_x}$, where  $\rho_{k_x}\leftarrow g(k_x)$ and $k_x \leftarrow \{0,1\}^{\secparam}$. The output of ${\cal B}$ is set to be the output of $A_{\secparam}$.
\par If $\rho \leftarrow g(k)$, for some $k \leftarrow \{0,1\}^{\secparam}$ then the output distribution of ${\cal B}$ is precisely the output distribution of $H_{3.i-1}$ and if $ \rho=\ketbra{\psi}{\psi}$, where $\ket{\psi} \leftarrow \Haar_{n(\secparam)}$, then the output distribution of ${\cal B}$ is precisely the output distribution of $H_{3.i}$. Thus, the probability that $A_{\secparam}$ distinguishes the hybrids $H_{3.i-1}$ and $H_{3.i}$ is precisely the distinguishing probability of ${\cal B}$. 
\par Note that ${\cal B}$ is not efficient since it needs to  sample Haar random states. Instead of sampling Haar random states, ${\cal B}$ instead uses $q(\secparam)$-state designs, where $q(\secparam)$ is the maximum number of queries $A_{\secparam}$ can make. $q(\secparam)$-state designs can be efficiently generated (in time polynomial in $q(\secparam)$)~\cite{ambainis2007quantum,dankert2009exact}. 
\par Since ${\cal B}$ is QPT, $\eps_3(\secparam)$, which is the probability that ${\cal B}$ distinguishes PRS from Haar random, has to be negligible in $\secparam$. 
\end{proof}

\paragraph{Hybrid $H_{4}$.} Output $A_{\secparam}^{{\cal O}_{\sf Haar}(\cdot)}$.
\par The hybrids $H_{3.q}$ and $H_4$ are identical. \\

\noindent Thus, it follows that $A_{\secparam}$ can distinguish $H_1$ and $H_4$ at most with negligible probability. This completes the proof.

\end{proof}

\subsection{(Quantum-Accessible) Adaptively-Secure PRFS}
\label{sec:cons:qrfs}
\noindent We now focus on achieving adaptive quantum query security (\Cref{def:qaprfs}). 

We present two constructions and two proofs. The first construction requires the existence of subexponentially-secure post-quantum one-way functions, and we use a technique called \emph{complexity leveraging} \hnote{add citation?} to prove security of the construction. The second construction only requires the existence of polynomially-secure post-quantum one-way functions, but the analysis requires more sophisticated tools, namely an extension of Zhandry's small-range distribution technique to handle random unitary oracles. We believe that presenting two different ways to analyze PRFS may be helpful for future works.

\subsubsection{First Construction}

Our first construction is similar to the construction in Section~\ref{sec:classical:aprfs} except that we use the binary phase PRS construction of~\cite{ji2018pseudorandom,BS19} (and which we analyze in \Cref{sec:simpleranalysis:prs}). 
\par To construct a $(d(\secparam),n(\secparam))$-APRFS $G$, where $d(\secparam)=\secparam^{\gamma}$ with $1 > \gamma > 0$ and $\frac{n(\secparam)}{2}-d(\secparam) = \omega(\log(\secparam))$, we start with two pseudorandom functions:
\begin{itemize}
    \item $(d(\secparam),\ell(\secparam))$-post-quantum secure pseudorandom function $F_1$ (\Cref{def:pqprfs}), 
    \item $(\ell(\secparam),n(\secparam),1,\frac{1}{2^{\secparam^{\delta}}})$-quantum-query secure pseudorandom function $F_2$ (\Cref{def:qqprfs}), where $1 > \delta > \gamma > 0$.
\end{itemize} 

\noindent We note that the same construction also works for any $d(\secparam) <  \secparam^{\gamma}$.

\paragraph{Construction.} We present the construction of $G_{\secparam}$ in Figure~\ref{fig:aprfs:cons}.

\begin{figure}[!htb]
 \begin{tabular}{|p{16cm}|}
    \hline \\
    {\bf Input}: $k \in \{0,1\}^{\secparam}$, $x \in \{0,1\}^{d(\secparam)}$.
    \begin{itemize}
    \item Initialize $\ket{x}_{\bf X}\ket{0}_{{\bf K}}$.
    \item Apply a unitary $V_{F_1}$ on $\ket{x}_{\bf X}\ket{0}_{{\bf K}}$, where $V_{F_1}$ is defined as follows: $V_{F_1}\ket{a}\ket{b} = \ket{a}\ket{F_1(k,a) \oplus b}$. Let $k_x=F_1(k,x)$.
    \item Apply a unitary $U_{F_2}$ on  $\ket{k_x}_{\bf K}\ket{0}_{\bf Y}\ket{-}_{\bf Anc}$ to obtain $\ket{k_x}_{\bf K} \ket{\psi_{k_x}}_{\bf Y} \ket{-}_{\bf Anc}$, where $\ket{\psi_{k_x}}_{\bf Y}=\sum_{y \in \{0,1\}^{n(\secparam)}} \frac{(-1)^{F_2(k_x,y)}}{\sqrt{2^{n(\secparam)}}} \ket{y}$, and  $U_{F_2}$ is described as follows:
    \begin{itemize}[noitemsep]
        %\item Apply $HX$ on ${\sf anc}$ register. 
        \item Apply $H^{\otimes n(\secparam)}$ on ${\bf Y}$. 
        \item Apply a unitary that maps $\ket{a}_{\bf K}\ket{b}_{\bf Y}\ket{c}_{\bf Anc}$ to $\ket{a}_{\bf K}\ket{b}_{\bf Y}\ket{c \oplus F_2(a,b)}_{\bf Anc}$.  
        %\item Apply $XH$ on ${\sf anc}$ register.
    \end{itemize}
    \item The resulting state is $\ket{x}_{\bf X}\ket{k_x}_{{\bf K}}\ket{\psi_{k_x}}_{\bf Y} \ket{-}_{\bf Anc}$. Trace out the register ${\bf Anc}$.

    \item Apply the unitary $V_{F_1}$ on the ${\bf X}$ and ${\bf K}$ registers (again) to obtain $\ket{x}_{\bf X}\ket{0}_{{\bf K}}\ket{\psi_{k_x}}_{\bf Y} \ket{0}_{\bf Anc}$. 
    \item Output $\ket{\psi_{k_x}}$. 
\end{itemize}
\\ \hline
    \end{tabular}
    \caption{Construction of $G_{\secparam}$.}
    \label{fig:aprfs:cons}
\end{figure}

\noindent We prove the security of the above construction in the lemma below. 

\begin{lemma}
Assuming the quantum-query $\nu(\secparam)$-security of the pseudorandom function $F_1$, where $\nu(\secparam)$ is a negligible function, the quantum-query $\frac{1}{2^{\secparam^{\delta}}}$-security of $F_2$, the QPT procedure $G$ described above is a secure $(d(\secparam),n(\secparam))$-APRFS (Definition~\ref{def:aprfs}). 
\end{lemma}
\begin{proof}
We prove this by a standard hybrid argument. Let $A_{\secparam}$ be a QPT distinguisher that distinguishes the oracles ${\cal O}_{\sf PRFS}(k,\cdot)$, where $k \leftarrow \{0,1\}^{\secparam}$, and ${\cal O}_{\sf Haar}$ with probability $\eps$. We prove that $\eps$ is negligible. Let $q$ be the number of queries made by $A_{\secparam}$. 

\paragraph{Hybrid $H_{1}$.} Output $A_{\secparam}^{{\cal O}_{\sf PRFS}(k,\cdot)}$, where $k \leftarrow \{0,1\}^{\secparam}$. Observe that ${\cal O}_{\sf PRFS}(k,\cdot)$ can be implemented as follows. 
\par On input a $d$-qubit register ${\bf X}$, do the following:
\begin{enumerate}
    \item Initialize a register ${\bf K}$ with $\ket{0}$.
    \item Apply the unitary $V_{F_1}$ on the registers ${\bf X}$ and ${\bf K}$, where $V_{F_1}$ is defined as follows: $V_{F_1}\ket{x}\ket{a} = \ket{x}\ket{F_1(k,x) \oplus a}$. 
    \item Initialize two registers ${\bf Y}$ and ${\bf Anc}$ with $\ket{0}$ and $\ket{-}$ respectively.
    \item Apply the unitary $U_{F_2}$ (described in the construction) on the registers ${\bf K}$, ${\bf Y}$ and ${\bf Anc}$. 
    %Controlled on ${\bf Anc}$ containing $k_x$, store the state $\ket{\psi_{k_x}}=\sum_{y \in \{0,1\}^{n(\secparam)}} (-1)^{F(k_x,y)} \ket{y}$ in the register ${\bf Y}$.
    \item Apply the unitary $V_{F_1}$ on the registers ${\bf X}$ and ${\bf K}$. 
    \item Output the state in the registers ${\bf X}$ and ${\bf Y}$. 
\end{enumerate}

\paragraph{Hybrid $H_{2}$.} Output $A_{\secparam}^{{\cal O}_{H_2}(k,\cdot)}$, where $k \leftarrow \{0,1\}^{\secparam}$, where $O_{H_2}$ is defined below.   \par Sample $\widehat{F_1}:\{0,1\}^{d(\secparam)} \rightarrow \{0,1\}^{\ell(\secparam)}$ uniformly at random. On input a $d$-qubit register ${\bf X}$, do the following:
\begin{enumerate}
    \item Initialize a register ${\bf K}$ with $\ket{0}$.
    \item Apply the unitary $V_{\widehat{F_1}}$ on the registers ${\bf X}$ and ${\bf K}$, where $V_{\widehat{F_1}}$ is defined as follows: $V_{\widehat{F_1}}\ket{x}\ket{a} = \ket{x}\ket{\widehat{F_1}(x) \oplus a}$. 
    \item Initialize two registers ${\bf Y}$ and ${\bf Anc}$ with $\ket{0}$ and $\ket{-}$ respectively..
    \item Apply the unitary $U_{F_2}$ (described in the construction) on the registers ${\bf K}$, ${\bf Y}$ and ${\bf Anc}$. 
    %Controlled on ${\bf Anc}$ containing $k_x$, store the state $\ket{\psi_{k_x}}=\sum_{y \in \{0,1\}^{n(\secparam)}} (-1)^{F(k_x,y)} \ket{y}$ in the register ${\bf Y}$.
    \item Apply the unitary $V_{\widehat{F_1}}$ on the registers ${\bf X}$ and ${\bf K}$. 
    \item Output the state in the registers ${\bf X}$ and ${\bf Y}$. 
\end{enumerate}

\begin{claim}
\label{clm:qprfs:1:2}
Assuming the quantum-query security of $F_1$, $A_{\secparam}$ can distinguish the hybrids $H_1$ and $H_2$ with probability at most $\nu(\secparam)$. 
\end{claim}
\begin{proof}
Suppose $A_{\secparam}$ can distinguish $H_1$ and $H_2$ with probability $> \nu(\secparam)$. We show that there exists a QPT adversary ${\cal B}$ that violates the quantum-query security of $F_1$ with probability $> \nu(\secparam)$. 
\par ${\cal B}$ runs $A_{\secparam}$ by simulating the oracle that $A_{\secparam}$ has access to. For every query made by $A_{\secparam}$, ${\cal B}$ with oracle access to ${\cal O}$ (which is either ${\cal O}_{\sf prf}(k,\cdot)$ or ${\cal O}_{\sf Rand}$;~\Cref{def:qqprfs}), does the following: it performs 6 steps, where steps 1,3,4 and 6 are the same as in $H_1$ (also identical to the steps 1,3,4 and 6 of $H_2$). We describe the steps 2 and 5 below. \begin{itemize}
    \item [2.] Send the registers ${\bf X}$ and ${\bf K}$ to ${\cal O}$.
    \item [5.] Send the registers ${\bf X}$ and ${\bf K}$ to ${\cal O}$. 
\end{itemize}
If ${\cal B}$ has oracle access to ${\cal O}_{\sf prf}(k,\cdot)$ then the output distribution of $A_{\secparam}$ is identical to the output distribution of $A_{\secparam}$ in $H_1$. If ${\cal B}$ has oracle access to ${\cal O}_{\sf Rand}$ then the output distribution of $A_{\secparam}$ is identical to the output of $A_{\secparam}$ in $H_2$. Thus, the distinguishing probability of ${\cal B}$ is $> \nu(\secparam)$, which is a contradiction to the quantum-query security of $F_1$.
\end{proof}

\noindent We fix some lexicographic ordering on $\{0,1\}^{d(\secparam)}$ and we use this ordering implicitly in the next few hybrids.

\paragraph{Hybrid $H_{3.i}$ for $i \in \{0,1\}^{d(\secparam)}$.} Output $A_{\secparam}^{{\cal O}_{H_{3.i}}(k,\cdot)}$, where $k \leftarrow \{0,1\}^{\secparam}$, where $O_{H_{3.i}}$ is defined below.   
\par Sample $\widehat{F_1}:\{0,1\}^{d(\secparam)} \rightarrow \{0,1\}^{\ell(\secparam)}$ uniformly at random. Also, sample $\widehat{F_2}:\{0,1\}^{\ell(\secparam)} \times \{0,1\}^{n(\secparam)} \rightarrow \{0,1\}$ uniformly at random. On input a $d$-qubit register ${\bf X}$, do the following:
\begin{enumerate}
    \item Initialize a register ${\bf K}$ with $\ket{0}$.
    \item Apply the unitary $V_{\widehat{F_1}}$ on the registers ${\bf X}$ and ${\bf K}$, where $V_{\widehat{F_1}}$ is defined as follows: $U_{\widehat{F_1}}\ket{x}\ket{a} = \ket{x}\ket{\widehat{F_1}(x) \oplus a}$. 
    \item Initialize two registers ${\bf Y}$ and ${\bf Anc}$ with $\ket{0}$ and $\ket{-}$ respectively.
    \item Apply the unitary $U^{{\sf Hyb}}$  on the registers ${\bf X},{\bf K}$, ${\bf Y}$ and ${\bf Anc}$, where $U^{{\sf Hyb}}$ maps $\ket{x}_{\bf X}\ket{k_x}_{\bf K}\ket{0}_{\bf Y}\ket{-}_{\bf Anc}$ to $\ket{x}_{\bf X}\ket{k_x}_{\bf K} \ket{\psi_{k_x}}_{\bf Y} \ket{-}_{\bf Anc}$ and $\ket{\psi_{k_x}}$ is defined below:
    \begin{itemize}
        \item If $x \leq i$, then  $\ket{\psi_{k_x}}_{\bf Y}=\sum_{y \in \{0,1\}^{n(\secparam)}} \frac{(-1)^{\widehat{F_2}(x,y)}}{\sqrt{2^{n(\secparam)}}} \ket{y}$,
        \item If $x >  i$, then  $\ket{\psi_{k_x}}_{\bf Y}=\sum_{y \in \{0,1\}^{n(\secparam)}} \frac{(-1)^{F_2(k_x,y)}}{\sqrt{2^{n(\secparam)}}} \ket{y}$
    \end{itemize}
   
    \item Apply the unitary $V_{\widehat{F_1}}$ on the registers ${\bf X}$ and ${\bf K}$. 
    \item Output the state in the registers ${\bf X}$ and ${\bf Y}$. 
\end{enumerate}

\begin{claim}
\label{clm:qprfs:3}
For every $i \in \{0,1\}^{d(\secparam)}$, assuming the quantum-query $\frac{1}{2^{\secparam^{\delta}}}$-security of $F_2$, $A_{\secparam}$ can distinguish the hybrids $H_{3.i}$ and $H_{3.i-1}$ only with  probability at most $\frac{1}{2^{\secparam^{\delta}}}$.
\end{claim}
\begin{proof}
Suppose $A_{\secparam}$ can distinguish $H_{3.i}$ and $H_{3.i-1}$ with probability $> \frac{1}{2^{\secparam^{\delta}}}$. We show that there exists a QPT adversary ${\cal B}$ that violates the quantum-query security of $F_2$ with probability $> \frac{1}{2^{\secparam^{\delta}}}$. We will first see how to construct an inefficient ${\cal B}$ and later, we will see how to make ${\cal B}$ run in quantum poly-time. 
\par ${\cal B}$ runs $A_{\secparam}$ by simulating the oracle $A_{\secparam}$ has access to. It first samples  $\widehat{F_1}:\{0,1\}^{d(\secparam)} \rightarrow \{0,1\}^{\ell(\secparam)}$ uniformly at random. For every query made by $A_{\secparam}$, ${\cal B}$ with oracle access to ${\cal O}$ (which is either ${\cal O}_{\sf prf}(k_i^*,\cdot)$, where $k_i^* \xleftarrow{\$} \{0,1\}^{\ell(\secparam)}$ or ${\cal O}_{\sf Rand}$;~\Cref{def:qqprfs}), does the following: it performs 6 steps, where steps 1,2,3,5 and 6 are the same as in $H_1$ (in turn identical to the steps 1,2,3,5 and 6 of $H_2$). We describe the step 4 below. \begin{itemize}
    \item [4.] Apply the following unitary on the registers ${\bf X},{\bf K}$, ${\bf Y}$ and ${\bf Anc}$, where the unitary maps $\ket{x}_{\bf X}\ket{k_x}_{\bf K}\ket{0}_{\bf Y}\ket{-}_{\bf Anc}$ to  $\ket{x}_{\bf X}\ket{k_x}_{\bf K} \ket{\psi_{k_x}}_{\bf Y} \ket{-}_{\bf Anc}$ and $\ket{\psi_{k_x}}$ is defined below\footnote{To see how to implement this unitary, let us take an example where $A_{\secparam}$ queries a pure state $\sum_{x \in \{0,1\}^{d(\secparam)}} \alpha_x \ket{x}_{\bf X}$. This argument can be naturally generalized to the case when $A_{\secparam}$ queries a mixed state. Using an appropriately defined controlled unitary, first create the state $\sum_{x \neq i} \alpha_x \ket{x}_{\bf X}\ket{k_x}_{\bf K}\ket{\psi_{k_x}}_{\bf Y} + \alpha_i \ket{i}_{\bf X}\ket{k_i}_{\bf K}\ket{0}_{\bf Y}$, where $\ket{\psi_{k_x}}$ for $x \neq i$ is computed as mentioned in the bullets in the proof of Claim~\ref{clm:qprfs:3}. Then, using the oracle ${\cal O}$, create the state $\ket{\psi_i}$. Using $\ket{\psi_i}$ and the controlled SWAP operation, create the state $\sum_{x \neq i} \alpha_x \ket{x}_{\bf X} \ket{k_x}_{\bf K} \ket{\psi_{k_x}}_{\bf Y} \ket{\psi_{i}}_{\bf Z} + \alpha_i \ket{i}_{\bf X} \ket{k_i}_{\bf K} \ket{\psi_i}_{\bf Y}\ket{0}_{\bf Z}$, where ${\bf Z}$ is some new register. Using the oracle ${\cal O}$ again, we can create the state $\sum_{x \neq i} \alpha_x \ket{x}_{\bf X} \ket{k_x}_{\bf K} \ket{\psi_{k_x}}_{\bf Y} \ket{0}_{\bf Z}\ket{0}_{\bf Anc} + \alpha_i \ket{i}_{\bf X}\ket{k_i}_{\bf K}\ket{\psi_i}_{\bf Y}\ket{0}_{\bf Z}\ket{{\cal O}(0)}_{\bf Anc}$. Finally, after  querying ${\cal O}$ on 0, we can suitably modify the previous state to obtain  $\sum_{x \neq i} \alpha_x \ket{x}_{\bf X}\ket{k_x}_{\bf K}\ket{\psi_{k_x}}_{\bf Y} \ket{0}_{\bf Z}\ket{-}_{\bf Anc} + \alpha_i \ket{i}_{\bf X}\ket{k_i}_{\bf K}\ket{\psi_i}_{\bf Y}\ket{0}_{\bf Z}\ket{-}_{\bf Anc}$. We can trace out the ${\bf Z}$ register to obtain the desired outcome. }: 
    \begin{itemize}
        \item If $x < i$, then  $\ket{\psi_{k_x}}_{\bf Y}=\sum_{y \in \{0,1\}^{n(\secparam)}} \frac{(-1)^{\widehat{F_2}(x,y)}}{\sqrt{2^{n(\secparam)}}} \ket{y}$
        \item If $x=i$, then $\ket{\psi_{k_x}}_{{\bf Y}}$ is obtained by querying ${\cal O}$ on $\sum_{y \in \{0,1\}^{n(\secparam)}} \frac{1}{\sqrt{2^{n(\secparam)}}} \ket{y} \ket{-}$ and then trace out the last qubit.
        \item If $x >  i$, then  $\ket{\psi_{k_x}}_{\bf Y}=\sum_{y \in \{0,1\}^{n(\secparam)}} \frac{(-1)^{F_2(k_x,y)}}{\sqrt{2^{n(\secparam)}} }\ket{y}$
    \end{itemize}
\end{itemize}
If ${\cal B}$ has oracle access to ${\cal O}_{\sf prf}(k_i^*,\cdot)$ then the output distribution of $A_{\secparam}$ is identical to the output distribution of $A_{\secparam}$ in $H_{3.i-1}$. If ${\cal B}$ has oracle access to ${\cal O}_{\sf Rand}$ then the output distribution of $A_{\secparam}$ is identical to the output of $A_{\secparam}$ in $H_{3.i}$. Thus, the distinguishing probability of ${\cal B}$ is at least $\frac{1}{2^{\secparam^{\delta}}}$.
\par Since ${\cal B}$ samples $\widehat{F_1}$ and $\widehat{F_2}$ uniformly at random, ${\cal B}$ does not run in polynomial time. However,  using~\Cref{lem:qwise}, we can replace both $\widehat{F_1}$ and $\widehat{F_2}$ with $2q(\secparam)$-wise independent hash functions, where $q(\secparam)$ is the number of queries made by $A_{\secparam}$, without changing the distinguishing probability. Once this change is made, ${\cal B}$ is now a QPT algorithm. This  contradicts the quantum-query security of $F_2$.
\end{proof}

\noindent Similarly, the following proof also holds. Let $i_{\rm min}$ be the minimum element in $\{0,1\}^{d(\secparam)}$ according to the lexicographic ordering described earlier. 

\begin{claim}
\label{clm:qprfs:2:3}
For every $i \in \{0,1\}^{d(\secparam)}$, assuming the quantum-query $\frac{1}{2^{\secparam^{\delta}}}$-security of $F_2$, $A_{\secparam}$ can distinguish the hybrids $H_{2}$ and $H_{3.i_{\rm min}}$ only with  probability at most $\frac{1}{2^{\secparam^{\delta}}}$.
\end{claim}

\paragraph{Hybrid $H_{4}$.} Output $A_{\secparam}^{{\cal O}_{H_4}(k,\cdot)}$, where $k \leftarrow \{0,1\}^{\secparam}$, where $O_{H_4}$ is defined below.   
\par Sample $\widehat{F_1}:\{0,1\}^{d(\secparam)} \rightarrow \{0,1\}^{\ell(\secparam)}$ uniformly at random. Also, sample $\widehat{F_2}:\{0,1\}^{\ell(\secparam)} \times \{0,1\}^{n(\secparam)} \rightarrow \{0,1\}$ uniformly at random. On input a $d$-qubit register ${\bf X}$, do the following:
\begin{enumerate}
    \item Initialize a register ${\bf K}$ with $\ket{0}$.
    \item Apply the unitary $V_{\widehat{F_1}}$ on the registers ${\bf X}$ and ${\bf K}$, where $V_{\widehat{F_1}}$ is defined as follows: $U_{\widehat{F_1}}\ket{x}\ket{a} = \ket{x}\ket{\widehat{F_1}(x) \oplus a}$. 
    \item Initialize two registers ${\bf Y}$ and ${\bf Anc}$ with $\ket{0}$ and $\ket{-}$ respectively.
    \item Apply the unitary $U_{\widehat{F_2}}$  on the registers ${\bf X},{\bf K}$, ${\bf Y}$ and ${\bf Anc}$, where $U_{\widehat{F_2}}$ is defined the same way as $U_{F_2}$ except that $\widehat{F_2}$ is used instead of $F_2$. 
   
    \item Apply the unitary $V_{\widehat{F_1}}$ on the registers ${\bf X}$ and ${\bf K}$. 
    \item Output the state in the registers ${\bf X}$ and ${\bf Y}$. 
\end{enumerate}

\noindent The following claim follows from the descriptions of $H_{3.i_{\rm max}}$, where $i_{\rm max}$ is the maximum element in $\{0,1\}^{d(\secparam)}$ according to the lexicographic ordering considered earlier, and $H_4$.  

\begin{claim}
\label{clm:qprfs:3:4}
The hybrids $H_{3.i_{\rm max}}$ and $H_4$ are identically distributed. 
\end{claim}

\paragraph{Hybrid $H_{5}$.} Output $A_{\secparam}^{{\cal O}_{{\sf Haar}}(\cdot)}$.   

\begin{claim}
\label{clm:4:5}
Suppose $q$ be the number of queries made by $A_{\secparam}$. Then,  $A_{\secparam}$ can distinguish the  hybrids $H_4$ and $H_5$ with probability at most $O\left(\frac{q \cdot 2^{d(\secparam)}}{2^{\frac{n(\secparam)}{2}}} \right)$. 
\end{claim}
\begin{proof}
We can define a sequence of $2^{d(\secparam)}$ intermediate hybrids. Hybrid $H_{4.i}$, for $i \in \{0,1\}^{d(\secparam)}$, is defined as follows: it behaves like hybrid $H_4$ except in Step 4. 
\begin{itemize}
    \item [4.] Apply the following unitary on the registers ${\bf X},{\bf K}$, ${\bf Y}$ and ${\bf Anc}$, where the unitary maps $\ket{x}_{\bf X}\ket{k_x}_{\bf K}\ket{0}_{\bf Y}\ket{-}_{\bf Anc}$ to  $\ket{x}_{\bf X}\ket{k_x}_{\bf K} \ket{\psi_{k_x}}_{\bf Y} \ket{-}_{\bf Anc}$ and $\ket{\psi_{k_x}}$ is defined below: 
    \begin{itemize}
        \item If $x \leq  i$, then  $\ket{\psi_{k_x}}_{\bf Y}=\sum_{y \in \{0,1\}^{n(\secparam)}} \frac{(-1)^{\widehat{F_2}(x,y)}}{\sqrt{2^{n(\secparam)}}} \ket{y}$
        \item If $x > i$, then  $\ket{\psi_{k_x}}_{\bf Y} \leftarrow \Haar_{n(\lambda)}$. 
    \end{itemize}
\end{itemize}
\noindent Hybrids $H_{4.i-1}$ and $H_{4.i}$ are $O(\frac{q}{2^{n(\secparam)\2}})$-close from~\Cref{lem:binaryphase}. Moreover, for the same reason, Hybrids $H_{4}$ and $H_{4.1}$ are also $O(\frac{q}{2^{n(\secparam)/2}})$-close. Finally, hybrids $H_{4.i_{\rm max}}$, where $i_{\rm max}$ is the maximum element, and $H_{5}$ are identically distributed. 
\par From this, it follows that the hybrids $H_4$ and $H_5$ are $O\left(\frac{q \cdot 2^{d(\secparam)}}{2^{n(\secparam)/2}}\right)$-statistically close. 
\end{proof}

\noindent By applying triangle inequality to~\Cref{clm:qprfs:1:2},~\Cref{clm:qprfs:3},~\Cref{clm:qprfs:2:3},~\Cref{clm:qprfs:3:4} and~\Cref{clm:4:5}, the following holds: 
\begin{eqnarray*}
\eps = \nu(\secparam) + \frac{2^{d(\secparam)}}{2^{\secparam^{\delta}}} + O\left(\frac{q \cdot 2^{d(\secparam)}}{2^{\frac{n(\secparam)}{2}}} \right)
\end{eqnarray*}
Using the facts that $d(\secparam)=\secparam^{\gamma}$, $\gamma < \delta$ and $\frac{n(\secparam)}{2} - d(\secparam)=\omega(\log(\secparam))$, it follows that $\eps(\secparam)$ is a negligible function in $\secparam$. 
\end{proof}

%\subsection{Alternate Proof of Quantum-Accessible Adaptively-Secure PRFS}
%\label{sec:cons:qrfs2}

\subsubsection{Second Construction}

We now present our second construction of QAPRFS. The advantage of this construction and analysis is that the security of~\Cref{fig:aprfs:cons} can be based only on the existence of polynomially-secure post-quantum one-way functions (as opposed to the sub-exponentially secure ones needed by the first construction). \luowen{and PRU?}
The security proof is a bit more involved and uses a quantum unitary version of Zhandry's small range distribution theorem, proven in \Cref{sec:quantum-small-range}, which could be of independent interest.

\par The construction is actually nearly identical, except we set the pseudorandom functions $F_1, F_2$ to the \emph{same} pseudorandom function $f$. In the previous subsection, we had to treat $F_1$ and $F_2$ separately since we required different levels of security from both. 
\luowen{Personally I feel like the construction here is much simpler than before (due to complexity leveraging and carefully choosing the parameters) although the security proof is more involved.} \hnote{I've added a paragraph at the beginning of Section 4.3 to explain why we have two different constructions}

\hnote{Need to be a bit more precise about what $f$ is... is it satisfying Def 2.5 or 2.6?}

%\subsection{Quantum-accessible adaptively-secure PRFS}

\begin{theorem}
Assuming the existence of quantum query secure one-way functions,~\Cref{fig:aprfs:cons} with $F_1 = F_2 = f$ is a $(d(\secparam),n(\secparam))$-QAPRFS. 
\end{theorem}
% \begin{theorem}
%   Assuming post-quantum OWFs, $n$-qubit $d$-bit-input quantum-accessible adaptively-secure PRFS generators exist for all polynomial $n = \omega(\log\lambda), d$.
% \end{theorem}
\begin{proof}
\begin{comment}
  On input a $\lambda$-bit random seed $k$, we construct a quantum-query secure pseudorandom function $f_k$ from $d + n$ bits to $n$ bits using the random seed.
  (Without loss of generality, we can assume the seed is $\lambda$ bits, or we invoke the PRG if the construction needs more randomness than $\lambda$ bits.)
  On input $x \in \{0, 1\}^d$, we invoke the binary phase construction using the pseudorandom function $f_k(x, \cdot)$, i.e. we are implementing the following isometry for every computational basis state:
  \[ G_k \ket{x} = \ket{x} \otimes \sum_{y \in \{0, 1\}^d} (-1)^{f_k(x, y)} \ket{y}. \]
  Note that the binary phase construction can be cleanly invoked so that it satisfies the correctness with respect to quantum queries property, as you could always uncompute the function value $f_k(x, y)$ knowing $\ket x, \ket y$.
  (There is nothing special about the binary phase construction except for the fact that the output could be produced without any auxiliary garbage. The original JLS construction \cite{ji2018pseudorandom} could also be extended to work with this construction.)
  
  We now argue that this satisfies the pseudorandomness requirement. %with respect to quantum access property.
\end{comment}  
  
  Assume for contradiction that there is an adversary that distinguishes the real world oracle from the ideal world oracle with noticeable advantage $\alpha(\lambda)$ with $q(\lambda)$ queries.
  We carry out the following hybrids, with $r = 1200q^3/\alpha$.
  \begin{description}
    \item[Hybrid 0] This is the real world oracle.
    \item[Hybrid 1] We change $f_k$ in the construction to a random function $f$, i.e. the oracle is now
      \[ \ket x \mapsto \ket{x} \otimes \sum_{y \in \{0, 1\}^d} (-1)^{f(x, y)} \ket{y}. \]
      This is computationally indistinguishable to the last hybrid by the quantum query security of PRF.
    \item[Hybrid 2] We interpret the random function as $\bit^d \to (\bit^n \to \bit^n)$ instead of $\bit^{d + n} \to \bit^n$, and change the random function $f$ to instead be sampled by a small-range distribution $\SR_r^{\bit^n \to \bit^n}(\bit^d)$, i.e. we sample $r$ random functions $\bit^n \to \bit^n$ and use a random one for every prefix $f(x, \cdot)$, where $x \in \{0,1\}^d$.
      By \Cref{thm:smallrange-unitary}, this change is statistically indistinguishable except with $300q^3/r$ advantage.
      \luowen{If we move this to the main body, then the small range thing becomes undefined... (They are only defined in \Cref{sec:quantum-small-range}.}
    \item[Hybrid 3.0 -- 3.$r$] Hybrid 3.0 is the same as Hybrid 2; or rather, we will design it in a way so that they are identically distributed.
      In particular, we are going to consider Hybrid 3.0 to be sampled equivalently (as Hybrid 2) as follows: for each $i \in [r]$, sample a function $g_i$ from $\bit^n \to \bit^n$; and upon invoking at $x \in \bit^d$, we consider invoking the isometry $U_{i_x}$, where isometry $U_i$ outputs $\sum_{y \in \bit^d} (-1)^{g_i(y)} \ket y$ for $i \in [r]$.
    
      For $i = 1, ..., r$, Hybrid 3.$i$ is the same as Hybrid 3.$(i - 1)$ except that the $i$-th entry of the small range distribution $U_i$ is changed from above to an isometry that outputs a Haar random state.
      (More formally, we are going to sample a Haar random state $\ket{\psi_i}$ and when $U_i$ is invoked, it simply outputs $\ket{\psi_i}$.)
      By \Cref{lem:binaryphase}, this change is statistically indistinguishable except with a negligible $O(q/2^{n/2})$ advantage.
    \item[Hybrid 4] We change the unitary distribution from $\SR_r^{U_n}(\bit^d)$ to $U_n^{\bit^d}$, where $U_n$ denote the isometry for outputting $n$-qubit Haar random states as specified before.
      This is equivalent to the ideal world oracle, and it is statistically indistinguishable from Hybrid 3.$r$ except with $300q^3/r$ advantage again by \Cref{thm:smallrange-unitary}.
  \end{description}
  By triangle inequality, we conclude that no efficient adversary is able to distinguish Hybrid 0 from Hybrid 4 with advantage more than $\frac{600q^3}r + \varepsilon = \alpha/2 + \varepsilon$ for some negligible quantity $\varepsilon$.
  This is a contradiction as $\alpha$ is noticeable.
\end{proof}

\paragraph{QAPRFS from PRUs.} We present another construction of QAPRFS from pseudorandom unitaries.
We first recall the definition of PRUs \cite{ji2018pseudorandom}, with added parameters about the size of the unitary.

\begin{definition}
  Let $\mathcal H(\lambda)$ be the Hilbert space over $n(\lambda)$ qubits.
  A family of unitary operators $\{U_k \in U(\mathcal H)\}_{k \in \mathcal K}$ is pseudorandom, if two conditions hold:
  \begin{enumerate}
    \item (Efficient generation) There is an efficient quantum algorithm $Q$, such that for all $k \in \{0, 1\}^\lambda$ and any $\ket\psi \in S(\mathcal H)$, $Q(k, \ket\psi) = U_k\ket\psi$.
    \item (Pseudorandomness) $U_k$ for a random $k$, given as an oracle, is computationally indistinguishable from a Haar random unitary operator.
      More precisely, for any efficient quantum algorithm $\mathcal A$, there exists a negligible function $\varepsilon(\cdot)$ such that
      \[
        \left|\Pr_{k \gets \{0, 1\}^\lambda}[A^{U_k}(1^\lambda) = 1] - \Pr_{U \gets \mu}[A^{U}(1^\lambda) = 1]\right| \le \varepsilon(\lambda).
      \]
  \end{enumerate}
\end{definition}

\begin{theorem}
  Assuming $n(\secparam)$-qubit pseudorandom unitaries (PRUs) with $n = \omega(\log\lambda)$, $(d(\secparam),n(\secparam))$-QAPRFS exist for any $d \le n$.
\end{theorem}
\begin{proof}
  We are going to consider the input length to be $n$ instead as $d < n$ can be handled by padding zeroes.
  \par Let $Q_{\secparam}$ be a pseudorandom unitary. The construction of $(d(\secparam),n(\secparam))$-QAPRFS $G_{\secparam}$ is defined as follows: $G(k,x) = Q_{\secparam}(k,x)$. 
  \par More formally, this is implemented as follows. We will abuse the notation and denote the circuit computing $G(k,\cdot)$ as $G_k(\cdot)$. Similarly, let $U_k$ be the unitary associated with $Q_{\secparam}(k,\cdot)$. We implement $G_k$ as follows: $G_k$ is an isometry that maps $\ket{x}$ to $\ket{x} \otimes U_k\ket{x}$. \\
  %To prove its security, we can replace the PRU with a Haar random unitary $R$, and argue that the unitary $U: \ket x \otimes \ket\bot \leftrightarrow \ket x \otimes R\ket x$ is indistinguishable to the ideal world oracle.
  
  \noindent We argue that $G_{\secparam}$ is a quantum accessible-secure PRFS.
  Assume for contradiction that there is an adversary that distinguishes the real world oracle from the ideal world oracle with noticeable advantage $\alpha(\lambda)$ with $q(\lambda)$ queries.
  We carry out the following hybrids, with $r = 1200q^3/\alpha$.
  \begin{description}
    \item[Hybrid 0] This is the real world oracle.
    \item[Hybrid 1] We change $U_k$ in the construction to a Haar-random unitary $U$, i.e. the oracle is now
      \[ \ket x \mapsto \ket{x} \otimes U\ket x. \]
      This is computationally indistinguishable to the last hybrid by security of PRU.
    \item[Hybrid 2]  \pnote{Edited:} Instead of applying $U$, we instead first apply a random permutation $\Pi:\{0,1\}^n \rightarrow \{0,1\}^n$ followed by a Haar random unitary $U'$. %We change the random unitary $U$ to be sampled $U'\Pi$, where $\Pi$ is a random permutation and $U'$ is another Haar random unitary.
      More formally, the oracle now computes the following:
      \[ \ket x \mapsto \ket{x} \otimes U'\Pi\ket x = \ket x \otimes U'\ket{\Pi(x)}. \]
      By unitary invariance of Haar, Hybrid 2 is identically distributed as Hybrid 1.
    \item[Hybrid 3] We replace $\Pi$ with a random function $f: \bit^n \to \bit^n$.
      In other words, the oracle is now
      \[ \ket x \mapsto \ket x \otimes U'\ket{f(x)}. \]
      Using (quantum) collision resistance of random functions \cite{Zhandry15}, we know that Hybrids 2 and 3 are statistically indistinguishable except with a negligible $O(q^3/2^n)$ advantage.
    \item[Hybrid 4] We change the random function $f$ to instead be sampled by a $\SR_r^{\bit^n \to \bit^n}(\bit^d)$.
      By \Cref{thm:smallrange-unitary}, this change is statistically indistinguishable except with $300q^3/r$ advantage.
    \item[Hybrid 5.1 -- 5.$r$] Hybrid 5.1 is the same as Hybrid 4.
      However, we are going to reinterpret the oracle as a unitary oracle.
      Formally, we are going to consider Hybrid 5.1 to be sampled equivalently (as Hybrid 4) as follows: for each $i \in [r]$, sample a Haar random state $\ket{\psi_i}$ that is orthogonal to the span of $\{\ket{\psi_j}\}_{j = 1, ..., i - 1}$; and upon invoking at $x \in \bit^d$, we consider invoking the isometry $U_{i_x}$, where isometry $U_i$ outputs $\ket{\psi_i}$ for $i \in [r]$.
    
      For $i = 2, ..., r$, Hybrid 5.$i$ is the same as Hybrid 5.$(i - 1)$ except that the $\ket{\psi_i}$ is sampled from Haar random instead of Haar random in a subspace.
      To compare the difference between these two hybrids, we can consider $\ket{\psi_i}$ to be first sampled from Haar random for Hybrid 5.$i$, and then project it to the orthogonal subspace to produce $\ket{\psi_i'}$ for Hybrid 5.$(i - 1)$.
      Note that any projector is $2$-Lipschitz, and thus by L\'evy's lemma \cite[Fact 2.2]{AQY21}, the probability that the overlap of a Haar random state with a subspace of polynomial dimension is more than $\frac{\log \lambda}{2^{n/2}}$ (which is negligible in $\lambda$) is negligible.
      Therefore, we establish that these hybrids are statistically indistinguishable except with a negligible probability by \Cref{lem:diamond-concat}.
    \item[Hybrid 6] We change the unitary distribution from $\SR_r^{U_n}(\bit^d)$ to $U_n^{\bit^d}$, where $U_n$ denote the unitary for $n$-qubit Haar random states.
      This is equivalent to the ideal world oracle, and it is statistically indistinguishable except with $300q^3/r$ advantage again by \Cref{thm:smallrange-unitary}.
  \end{description}
  By triangle inequality, we conclude that no efficient adversary is able to distinguish Hybrid 0 from Hybrid 6 with advantage more than $\frac{600q^3}r + \varepsilon$ for some negligible quantity $\varepsilon$, which is negligibly close to $\alpha/2$ by our choice of parameters.
  This is a contradiction as $\alpha$ is noticeable.
\end{proof}

\noindent We conclude by noting that both constructions above (a) use PRF/PRU as a black-box (and their proofs relativize to the setting of either random oracles or families of Haar unitary oracles as considered by Kretschmer~\cite{Kretschmer21}); and (b) only get super-logarithmic output length.
We leave as future work to construct these for logarithmic output length.

\section{On the Necessity of Computational Assumptions} 
\label{sec:computational-assumptions}

The following lemma shows that the security guarantee of a PRS generator (and thus of PRFS generators) can only hold with respect to computationally bounded distinguishers, provided that the output length is at least $\log \lambda$. 
%\luowen{This could be removed as well?} \note{HY: I think this is worth keeping as I haven't seen this written elsewhere, and we (will) reference this in the inro}

\begin{lemma}
  \label{lem:prs-inefficient-dist}
  Let $G$ be a PRS generator with output length $n(\lambda) \geq \log \lambda$. Then there exists a polynomial $t(\lambda)$ and a quantum algorithm $A$ (not efficient in general) such that
    \[
        \left | \Pr_{k \leftarrow \{0,1\}^\lambda} \left [ A_\lambda\left(G_\lambda(k)^{\otimes t(\lambda)}\right) = 1 \right] - \Pr_{\ket{\vartheta} \leftarrow \Haar_{n(\lambda)}} \left [ A_\lambda\left(\ketbra\vartheta\vartheta^{\otimes t(\lambda)}\right) = 1 \right] \right | \ge \frac{1}{3}
    \]
    for all sufficiently large $\lambda$.
\end{lemma}
\begin{proof}
  For notational convenience we abbreviate $n = n(\lambda)$ and $t = t(\lambda)$.
  We split the proof into two cases.
  \pnote{I'm not sure why we have two cases here. From the pseudorandomness property, with overwhelming probability over the choice of $k$, the PRS state is pure. We can directly consider Case 2 then. Am I missing something here?}
  \luowen{Yes, case 1 is basically mirroring what you are saying... In other words, if case 1 is true then this attack is even efficient (and case 2 is the real inefficient attack). However, this lemma is true if $G$ is insecure PRS... (In other words, this way I do not need to assume it is secure to give an attack.)}
  
  \emph{Case 1:} if there does not exist a negligible function $\nu(\cdot)$ such that
  \begin{equation}
    \label{eq:purity-test}
    \Pr_k\left[\min_{\ket\theta}\TD(G_\lambda(k), \ketbra\theta\theta) \le \nu(\lambda)\right] \ge \frac12.
  \end{equation}
  Then there exists some non-negligible function $\kappa(\cdot)$ such that with probability at least $\frac12$ over the choice of $k$, $\min_{\ket\theta}\TD(G_\lambda(k), \ketbra\theta\theta) \ge \kappa(\lambda)$. \pnote{I think the next 3 lines seem a bit unnecessary and there might be a simpler argument. If we do a purity test on two copies then the probability of succeeding is at most $\frac{1}{2} (1 - \kappa)$. This follows from the fact that getting the same pure state once more happens with probability at most $1 - \kappa$, which in turn follows from $\min_{\ket\theta}\TD(G_\lambda(k), \ketbra\theta\theta) \ge \kappa(\lambda)$. Now, if we repeat this $t = O(1/\kappa(\lambda))$ copies, we are done.} \luowen{Not sure if I understand... Could you make it more formal?}
  Let $\nu_{k, 1} \ge ... \ge \nu_{k, 2^n}$ and $\ket{\alpha_{k, 1}}, ..., \ket{\alpha_{k, 2^n}}$ be eigenvalues and eigenvectors for $G_\lambda(k)$.
  Then $\kappa \le \TD(G_\lambda(k), \ketbra{\alpha_{k, 1}}{\alpha_{k, 1}}) = \frac12(1 - \nu_{k, 1} + \nu_{k, 2} + \cdots + \nu_{k, 2^n}) = 1 - \nu_{k, 1}$.
  Thus by H\"older's inequality, $\Tr(G_\lambda(k)^2) \le 1 - \kappa$.
  Therefore, a purity test using $t = O(1/\kappa(\lambda))$ copies will correctly reject PRS states with probability at least $\frac13$ but never incorrectly reject any Haar random state.

  \emph{Case 2:} if there exists a negligible function $\nu(\cdot)$ such that \eqref{eq:purity-test} holds.
There exists a polynomial $t(\lambda)$ such that 
\[
    2^{\lambda} \leq \frac{1}{6} \cdot \dim \Pi^{2^{n},t}_\sym = \frac{1}{6} \cdot \binom{2^{n} + t - 1}{t}
\] 
for all sufficiently large $\lambda$. This is because by setting $t = \lambda + 1$, we can lower bound the dimension of $\Pi^{2^n,t}_\sym$ by $\binom{2\lambda}{\lambda+1}$ and
\[
  \binom{2\lambda}{\lambda} \geq \frac{\lambda}{\lambda+1} \frac{4^\lambda}{\sqrt{\pi\lambda}}\left(1 - \frac1{8\lambda}\right)
\]
which is much larger than $6 \cdot 2^\lambda$ for all sufficiently large $\lambda$.

Let $g \subseteq \bit^\lambda$ be the set of $k$'s such that $\min_{\ket\theta}\TD(G_\lambda(k), \ketbra\theta\theta) \le \nu(\lambda)$.
Note that $2^{\lambda}$ is an upper bound on the rank of the density matrix
\begin{equation}
    \label{eq:mixture-of-prs}
    \E_{k \leftarrow g} \ketbra{\psi_k}{\psi_k}^{\otimes t},
\end{equation}
where $\ket{\psi_k} = \arg\min_{\ket\theta}\TD(G_\lambda(k), \ketbra\theta\theta)$. Note that by \Cref{fact:avg-haar-random} the rank of the density matrix 
\begin{equation}
    \label{eq:haar-random-mixture}
    \E_{\ket{\vartheta} \leftarrow \Haar_{n(\lambda)}} \ketbra{\vartheta}{\vartheta}^{\otimes t} = \frac{\Pi^{2^n,t}_\sym}{\dim \Pi^{2^n,t}_\sym}
\end{equation}
is equal to $\dim \Pi^{2^n,t}_\sym$. 

For all $\lambda$, define the quantum circuit $A_\lambda$ that, given a state on $tn$ qubits, performs the two-outcome measurement $\{ P, I - P \}$ where $P$ is the projector onto the support of $  \E_{k \leftarrow g} \ketbra{\psi_k}{\psi_k}^{\otimes t}$, and accepts if the $P$ outcome occurs. 

By assumption of case 2, given the density matrix~\eqref{eq:mixture-of-prs} the circuit $A_\lambda$ will accept with probability at least $\frac12$. On the other hand, given the density matrix~\eqref{eq:haar-random-mixture} the circuit $A_\lambda$ will accept with probability
\[
    \Tr \left ( P \cdot \frac{\Pi^{2^n,t}_\sym}{\dim \Pi^{2^n,t}_\sym} \right) \leq \Tr \left (\frac{P}{\dim \Pi^{2^n,t}_\sym} \right) = \frac{\mathrm{rank}(P)}{\dim \Pi^{2^n,t}_\sym} \leq \frac{1}{6}~.
\]
Letting $A = \{A_\lambda\}_\lambda$ we obtained the desired Lemma statement.
\end{proof}

We remark that the attack given in \Cref{lem:prs-inefficient-dist} cannot be used on smaller output length, up to additive factors of superpolynomially smaller order in the output length.
Suppose $n = \log \lambda - \omega(\log\log\lambda)$ and for any $t = \lambda^{O(1)}$,
  \begin{align*}
    \log \binom{2^{n} + t - 1}{t}
      %&\le \log \binom{\lambda^{1 - \eps} + t - 1}{\lambda^{1 - \eps} - 1} \\
      &\le 2^n \cdot \log\frac{e(2^n + t - 1)}{2^n - 1} \\
      &= \frac\lambda{\omega(\log\lambda)} \cdot O(\log \lambda).
  \end{align*}
  This means that $\binom{2^{n} + t - 1}{t} = 2^{\lambda/\omega(\log \lambda)} \ll 2^\lambda$ and therefore the attack above does not necessarily apply.
  Indeed, Brakerski and Shmueli~\cite{BrakerskiS20} have shown that PRS generators with output length $n(\secparam) \leq c \log \lambda$ for some $c > 0$ can be achieved with statistical security. 

We conclude the section by remarking that the result of Kretschmer~\cite{Kretschmer21} can be easily generalized so that PRS generators with output length at least $\log \lambda + c$ (for some small constant $0 < c < 2$) imply ${\sf BQP} \neq {\sf PP}$ as well\footnote{For readers familiar with \cite{Kretschmer21}, it can be verified that a sufficient condition for that proof to go through is if $2^\lambda \cdot e^{-2^n/3}$ is negligible, which is satisfied if $n \ge \log\lambda + 2$.}. %\note{HY: this looks like essentially the same as Kretschmer's proof. should we just say it follows from the same proof or do you think there's something new and subtle here?} \luowen{It is essentially the same proof except with the calculations done for a much smaller output length; and we need to handle the case that the output might not be pure.}

\medskip
\newcommand{\goodevent}{{\sf GoodEvent}}
\newcommand{\bfx}{{\bf x}}
\newcommand{\distr}{{\cal D}}
\section{Tomography with Verification}
\label{sec:tomography}
%\noindent The notion of tomography is concerned with discerning certain properties, that can be expressed as classical strings, given multiple copies of a quantum state. Tomography has been extensively studied in various forms, including [CITE]. 
\emph{Quantum state tomography} (or just \emph{tomography} for short) is a process that takes as input multiple copies of a quantum state $\rho$ and outputs a string $u$ that is a classical description of the state $\rho$; for example, $u$ can describe an approximation of the density matrix $\rho$, or it could be a a more succinct description such as a \emph{classical shadow} in the sense of~\cite{Huang2020}. 
%\pnote{I think our notion is more general than what the previous statement claims. For example, maybe the tomography algorithm could recover some classical shadows that might not be sufficient to recover an approximation of the original density matrix but good enough for verification.} \hnote{right! edited.} 
In this paper, we use tomography as a tool to construct protocols based on pseudorandom states with only \emph{classical} communication. 

\par For our applications, we require tomography procedures satisfying a useful property called verification. 
%Informally speaking, state tomography allows us to recover a classical description of a quantum state $\rho$, given sufficiently many copies of $\rho$. We augment tomography with a verification procedure. 
Suppose we execute a tomography algorithm on multiple copies of a state to obtain a classical string $u$. The verification algorithm, given $u$ and the algorithm to create this state, checks if $u$ is consistent with this state or not. Verification comes in handy when tomography is used in cryptographic settings, where we would like to make sure that the adversary has generated the classical description associated with a quantum state according to some prescribed condition (this will be implictly incorporated in the verification algorithm).

\paragraph{Verifiable Tomography.} Let ${\cal C}=\{\Phi_{\secparam}\ :\ \secparam \in \mathbb{N}\}$ be a family of channels where each channel $\Phi_\secparam$ takes as input $\ell(\secparam)$ qubits for some polynomial $\ell(\cdot)$. A \emph{verifiable tomography scheme} associated with $\cal C$ is a pair $(\tomography,\verify)$ of QPT algorithms, which have the following input/output behavior:
\begin{itemize} 
\item $\tomography$: given as input a quantum state $\rho^{\otimes L}$ for some density matrix $\rho$ and some number $L$, output a classical string $u$ (called a \emph{tomograph} of $\rho$). 

%on input $L$ copies of a mixed state $\rho$, namely $\sigma^{\otimes L}$ \hnote{don't understand this ``namely'' part. can it be removed?}, output a classical string $u$. \hnote{how is $L$ quantified? is this for all possible dimensions of $\rho$? It seems like $L$ should be a function of the dimension.} \pnote{I'm not specifying what $L$ is here and instead it is incorporated in correctness definitions.. EDIT: actually, I think it might be better to write $\tomography(\sigma)$ and not specify the number of copies here.. we can mention it in the correctness definition} 
\item $\verify$: given as input a pair of classical strings $(\bfx,u)$ where $\bfx$ has length $\ell(\secparam)$, output $\valid$ or $\invalid$. 
\end{itemize}

\noindent We would like $(\tomography,\verify)$ to satisfy correctness which we describe next. 

\subsection{Correctness Notions for Verifiable Tomography}
We can consider two types of correctness. The first type of correctness, referred to as \emph{same-input correctness}, states that $\verify(\bfx,u)$ outputs $\valid$ if $u$ is obtained by running the $\tomography$ procedure on copies of the output of $\Phi_{\secparam}(\bfx)$. The second type of correctness, referred to as \emph{different-input correctness}, states that $\verify(\bfx',u)$ outputs $\invalid$ if $u$ is obtained by applying tomography to $\Phi_{\secparam}(\bfx)$, where $(\bfx',\bfx)$ do not satisfy a predicate $\Pi$.

\paragraph{Same-Input Correctness.} Consider the following definition. 

\begin{definition}[Same-Input Correctness]
\label{def:sameinput}
We say that $(\tomography,\verify)$ satisfies {\em $L$-same-input correctness}, for some polynomial $L(\cdot)$, such that for every $\bfx \in \{0,1\}^{\ell(\secparam)}$, if the following holds: 
    $$\prob\left[ \valid \leftarrow \verify\left(\bfx, \tomography\left( (\Phi_\secparam(\bfx))^{\otimes L(\secparam)} \right) \right)  \right] \geq 1 - {\sf negl}(\secparam),$$
\end{definition}

\noindent For some applications, it suffices to consider a weaker definition. Instead of requiring the correctness guarantee to hold for every input, we instead require that it holds over some input distribution. 

\begin{definition}[Distributional Same-Input Correctness]
\label{def:distr:sameinput}
 We say that $(\tomography,\verify)$ satisfies {\em $(L,\distr)$-distributional same-input correctness}, for some polynomial $L(\cdot)$ and distribution $\distr$ on $\ell(\secparam)$-length strings, if the following holds: 
    $$\prob\left[ \valid \leftarrow \verify\left(\bfx, \tomography\left( (\Phi_\secparam(\bfx))^{\otimes L(\secparam)} \right) \right)\ :\ \bfx \leftarrow  \distr \right] \geq 1 - {\sf negl}(\secparam)$$
\end{definition}

\paragraph{Different-Input Correctness.} Ideally, we would require that $\verify(\bfx,u)$ outputs $\invalid$ if $u$ is produced by tomographing  $\Phi_{\secparam}(\bfx')$, and $\bfx'$ is any string such that $\bfx' \neq \bfx$. However, for applications, we only require that this be the case when the pair $(\bfx,\bfx')$ satisfy a relation defined by a predicate $\Pi$. In other words, 
%we do not require this to hold for the case when $\bfx'$ can be any string but rather it belongs to a subset of strings. This subset can be specified using a predicate $\Pi$. In more detail,
we require $\verify(\bfx,u)$ outputs $\invalid$ only when $u$ is a tomograph of $\Phi_{\secparam}(\bfx')$ and $\Pi(\bfx',\bfx)=0$.
\par We define this formally below. 

\begin{definition}[Different-Input Correctness]
\label{def:diffinputcorrecntess}
We say that $(\tomography,\verify)$ satisfies {\em $(L,\Pi)$-different-input correctness}, for some polynomial $L(\cdot)$ and predicate $\Pi:\{0,1\}^{\ell(\secparam)} \times \{0,1\}^{\ell(\secparam)} \rightarrow \{0,1\}$, such that for every $\bfx,\bfx' \in \{0,1\}^{\ell(\secparam)}$ satisfying $\Pi(\bfx,\bfx')=0$, if the following holds: 
    $$\prob\left[ \invalid \leftarrow \verify\left(\bfx', \tomography\left( (\Phi_\secparam(\bfx))^{\otimes L(\secparam)} \right) \right)  \right] \geq 1 - {\sf negl}(\secparam)$$
\end{definition}

\noindent Analogous to~\Cref{def:distr:sameinput}, we  correspondingly define below the notion of $(L,\distr,\Pi)$-different-input correctness. 

\begin{definition}[Distributional Different-Input Correctness]
\label{def:distr:diffinput}
We say that $(\tomography,\verify)$ satisfies {\em $(L,\Pi,\distr)$-distributional different-input correctness}, for some polynomial $L(\cdot)$, predicate $\Pi:\{0,1\}^{\secparam} \times \{0,1\}^{\secparam} \rightarrow \{0,1\}$ and distribution $\distr$ supported on  $(\bfx,\bfx') \in \{0,1\}^{\ell(\secparam)}\times \{0,1\}^{\ell(\secparam)}$ satisfying $\Pi(\bfx,\bfx')=0$, if the following holds: 
     $$\prob\left[ \invalid \leftarrow \verify\left(\bfx', \tomography\left( (\Phi_\secparam(\bfx))^{\otimes L(\secparam)} \right) \right)\ :\ (\bfx,\bfx') \leftarrow \distr  \right] \geq 1 - {\sf negl}(\secparam)$$
\end{definition}

Sometimes we will use the more general {\em $(\varepsilon,L,\Pi,\distr)$-distributional different-input correctness} definition. In this case, the probability of $\verify$ outputting $\invalid$ is bounded below by $1-\varepsilon$ instead of $1-{\sf negl}(\secparam)$.
% \paragraph{Correctness.} There exists a polynomial $L(\cdot)$ such that for every $\secparam \in \mathbb{N}$, $\bfx \in \{0,1\}^{\ell(\secparam)}$, we require the following two quantities to hold: \pnote{we need to give more context about $\Pi$.}
% \begin{enumerate}
%     \item {\em Same-key correctness}: for every $\bfx \in \{0,1\}^{\ell(\secparam)}$,
%     $$\prob\left[ \valid \leftarrow \verify\left(\bfx, \tomography\left( (\Phi_\secparam(\bfx))^{\otimes L(\secparam)} \right) \right)  \right] \geq 1 - {\sf negl}(\secparam) $$
%     \item {\em $\Pi$-different-key correctness}: there exists a predicate $\Pi:\{0,1\}^{\ell(\secparam)} \times \{0,1\}^{\ell(\secparam)} \rightarrow \{0,1\}$ such that for every $\bfx,\bfx' \in \{0,1\}^{\ell(\secparam)}$ with $\Pi(\bfx,\bfx') = 0$,
%      $$\prob\left[ \invalid \leftarrow \verify\left(\bfx', \tomography\left( (\Phi_\secparam(\bfx))^{\otimes L(\secparam)} \right) \right)  \right] \geq 1 - {\sf negl}(\secparam) $$
% \end{enumerate}

\subsection{Verifiable Tomography Procedures} 
\label{sec:tomography:procedures}
We will consider two different instantiations of $(\tomography,\verify)$ where the first instantiation will be useful for bit commitments and the second instantiation will be useful for pseudo one-time pad schemes. 
\par In both the instantiations, we use an existing tomography procedure stated in the lemma below. 

\begin{lemma}[Section 1.5.3, \cite{Lowe21}]
\label{lem:lowe}
There exists a tomography procedure ${\cal T}$ that given $s N^2$ copies of an  $N$-dimensional density matrix $\rho$, outputs a matrix $M$ such that $\mathbb{E} \| M - \rho\|_F^2 \leq \frac{N}{s}$ where the expectation is over the randomness of the tomography procedure. Moreover, the running time of ${\cal T}$ is polynomial in $s$ and $N$. 
\end{lemma}

\noindent We state and prove a useful corollary of the above lemma. 

\newcommand{\negl}{\mathsf{negl}}
\begin{corollary}
\label{cor:tomography}
There exists a tomography procedure ${\cal T}_{\sf imp}$ that given $4s N^2 \secparam$ copies of an $N$-dimensional density matrix $\rho$, outputs a matrix $M$ such that the following holds: 
$$\prob\left[ \|M - \rho\|_F^2 \leq \frac{9N}{s} \right] \geq 1 -  \negl(\secparam)$$
Moreover, the running time of ${\cal T}_{\sf imp}$ is polynomial in $s,N$ and $\secparam$. 
\end{corollary}
\begin{proof}
Set $\eps=\frac{N}{s}$. We define ${\cal T}_{\sf imp}$ as follows: \\

\noindent ${\cal T}_{\sf imp}\left( \rho^{\otimes 4sN^2 \secparam} \right)$: on input $4s N^2 \secparam$ copies of an $N$-dimensional density matrix $\rho$, do the following:
\begin{itemize}
    \item For every $i \in [\secparam]$, compute $M_i \leftarrow {\cal T}(\rho^{\otimes 4s N^2})$, 
    \item Output $M_{i^*}$, where $|\{j:\|M_j - M_{i^*}\|_F^2 \leq 4\eps\}| > \frac{\secparam}{2}$. If no such $i^* \in [\secparam]$ exists, output $\bot$. 
\end{itemize}

\noindent To prove that ${\cal T}_{\sf imp}$ satisfies the condition mentioned in the statement of the corollary, we first consider the following event. \\

\noindent {\sf GoodEvent}: For every $i \in [\secparam]$, compute $M_i \leftarrow {\cal T}(\rho^{\otimes 4s N^2})$. There exists a set $S \subseteq \{M_1,\ldots,M_{\secparam}\}$ such that for every $M \in S$, $\|M - \rho\|_F^2 \leq  \eps$ and moreover, $|S| > \frac{\secparam}{2}$. \\

\noindent Consider the following two claims.

\begin{claim}
\label{clm:goodevent}
$\prob\left[ {\sf GoodEvent} \right] \geq 1 - \negl(\secparam)$. 
\end{claim}
\begin{proof}
Applying Markov's inequality to~\Cref{lem:lowe}, we have $\prob\left[  M \leftarrow {\cal T}\left( \rho^{\otimes 4s N^2} \right)\ \text{ and }\ \|M - \rho\|_F^2 \geq \eps \right] \leq \frac{1}{4}$. Let ${\bf X}_i$ be a random variable that is set to 1 if the $i^{th}$ execution of ${\cal T}(\rho^{\otimes 4s N^2})$ outputs $M_i$ such that $\|M_i - \rho \|_F^2 \geq \eps$. Observe that $\mathbb{E}[{\bf X}_i] \leq \frac{1}{4}$. Let $\mu=\mathbb{E}[\sum_{i=1}^\secparam {\bf X}_i] \leq \frac{\secparam}{4}$. By Chernoff bound\footnote{For any set of iid Bernoulli random variables ${\bf X}_1,\ldots,{\bf X}_N$, for any $\delta > 0$, the following holds: $\prob\left[ \sum_{i=1}^\secparam {\bf X}_i \geq (1+\delta)\mu \right] \leq \frac{1}{e^{\delta \mu \cdot \frac{\delta}{2+\delta}}}$}, there exists $\delta \geq 1$ such that $\prob[\sum_{i=1}^\secparam {\bf X}_i \geq \frac{\secparam}{2}] = \prob[\sum_{i=1}^\secparam {\bf X}_i \geq (1+\delta)\mu] \leq \frac{1}{e^{\frac{\secparam}{4} \cdot \frac{1}{2+1} }} = \negl(\secparam)$. 
\par This proves that with overwhelming probability, there exists a subset $S$ of $\{M_1,\ldots,M_{\secparam}\}$ of size $> \frac{\secparam}{2}$ such that for every $M \in S$, $\|M - \rho\|_F^2 \leq \eps$.  
\end{proof}

\begin{claim}
\label{clm:condition:goodevent}
$\prob\left[ \|M - \rho\|_F^2 \leq 9\eps \ \big|\ {\sf GoodEvent} \right] = 1$. 
\end{claim}
\begin{proof}
Since we are conditioning on ${\sf GoodEvent}$, there exists a subset $S$ of $\{M_1,\ldots,M_{\secparam}\}$ such that for every $M \in S$, $\|M - \rho\|_F^2 \leq \eps$. Moreover, $|S| > \frac{\secparam}{2}$. Suppose ${\cal T}_{{\sf imp}}$ outputs $M \in S$ then we are done. This is because, for every $M,M' \in S$, it holds that $\|M - M'\|_F^2 \leq 4\eps$.
\par So we might as well assume that $M \notin S$. By the description of ${\cal T}_{\sf imp}$, it follows that there exists a subset $S'$ of $\{M_1,\ldots,M_{\secparam}\}$ such that $|S'| > \frac{\secparam}{2}$ and $\|M - M'\|_F^2 \leq 4\eps$ for every $M' \in S'$. Since $S \cap S' \neq \emptyset$, it follows that there exists an $M' \in S$ such that $\|M - M'\|_F^2 \leq 4\eps$. By definition of $S$ and by the fact that Frobenius norm is a matrix norm, it follows that $\|M - \rho\|_F^2 \leq 9\eps$.  
\end{proof}

\noindent Thus, we have the following: 

\begin{align*}
\prob\left[ \|M - \rho\|_F^2 \leq 9\eps \right] & = \prob\left[ \|M - \rho\|_F^2 \leq 9\eps\ |\ \goodevent \right] \prob[\goodevent]\\
&  \qquad +  \prob\left[ \|M - \rho\|_F^2 \leq 9\eps\ |\ \neg \goodevent \right] \prob[\neg \goodevent] \\ 
& \geq   \prob\left[ \|M - \rho\|_F^2 \leq 9\eps\ |\ \goodevent \right] \cdot (1 - \negl(\secparam))\ \  & \text{(\Cref{clm:goodevent})} \\
& = (1 - \negl(\secparam))\ \  & \text{(\Cref{clm:condition:goodevent})}
\end{align*}

\end{proof}

\subsubsection{First Instantiation}
\noindent We will work with a verifiable tomography procedure that will be closely associated with a PRFS. In particular, we will use a $(d(\secparam),n(\secparam))$-PRFS  $\{G_{\secparam}\left(\cdot,\cdot \right)\}$ satisfying recognizable abort property (\Cref{def:classicalgen:strongrec}). Let $\widehat{G}$ be the QPT algorithm associated with $G$ according to~\Cref{def:classicalgen:strongrec}. Note that the output length of $\widehat{G}$ is $n+1$. We set $d(\secparam) = \lceil \frac{\log(\secparam)}{\log(\log(\secparam))} \rceil$ and $n(\secparam)=\lceil 3\log(\secparam) \rceil$.
\par We will describe the algorithms $(\tomography,\verify)$ in Figure~\ref{fig:firstinstantiation}. The set of channels ${\cal C}=\{\Phi_{\secparam}: \secparam \in \mathbb{N}\}$ is associated with $(\tomography,\verify)$, where $\Phi_{\secparam}$ is defined as follows: 
\begin{itemize}
    \item Let the input be initialized on register ${\bf A}$.
    \item Controlled on the first register containing the value $(P_x,k,x,b)$, where  $P_x$ is an $n$-qubit Pauli, $k \in \{0,1\}^{\secparam}, b \in \{0,1\}$, do the following: compute $\left(I\otimes P_x^b\right) \widehat{G}_{\secparam}(k,x) \left(I\otimes P_x^b\right)$ and store it in the register ${\bf B}$. 
    \item Trace out ${\bf A}$ and output ${\bf B}$. 
\end{itemize}

% \pnote{describe the input and output space of $\Phi_{\secparam}$... the description below does not cover all possibilities} $$\Phi_\secparam\left(P\|k\|b\right) = \bigotimes_{x\in \{0,1\}^d} \left(\left(I\otimes P_x^b\right) \widehat{G}_{\secparam}(k,x) \left(I\otimes P_x^b\right)\right),$$ 
% where $P$ is an $m$-qubit Pauli that can be expressed as a tensor product of $n$-qubit Paulis. That is, $P = \bigotimes_{x\in \{0,1\}^d} P_{x}$, where $P_x$ is an $n$-qubit Pauli. \\

\begin{figure}[!htb]
\begin{tabular}{|p{16cm}|}
\hline 
\ \\
\noindent \underline{$\tomography(\rho^{\otimes L})$}: On input $L$ copies of an  $2^{(n+1)}$-dimensional density matrix $\rho$, compute ${\cal T}_{\sf imp}(\rho^{\otimes L})$ to obtain $M$, where ${\cal T}_{\sf imp}$ is given in Corollary~\ref{cor:tomography}. Output $M$.  \\
\ \\

\noindent \underline{$\verify({\bf x},M)$}:
\begin{enumerate}
    \item Run $\rho^{\otimes L} \leftarrow \left( \Phi_\secparam\left({\bf x}\right) \right)^{\otimes L}$, where $L = 3^{8}2^{3(n+1)+2}\secparam$. 
    \item Compute $\widehat{M} \leftarrow \tomography\left( \rho^{\otimes L} \right)$.
    \item If $\bra{\bot}M\ket{\bot}>\frac{1}{9}$ for any $x\in\{0,1\}^{d}$, output $\invalid$.
    \item If $\|M - \widehat{M}\|_{F}^2 \leq \frac{4}{729}$ output $\valid$. Output $\invalid$ otherwise.
\end{enumerate}
\\ 
\hline
\end{tabular}
\caption{First instantiation of $\tomography$}
\label{fig:firstinstantiation}
\end{figure}

\noindent The channel $\Phi_\secparam$ can be represented as a quantum circuit of size polynomial in $\secparam$ as the PRFS generator $\hat{G}$ runs in time polynomial in $\secparam$. 

\paragraph{Distributional Same-Input Correctness.} We prove below that $(\tomography,\verify)$ satisfies distributional same-input correctness. For every $x \in \{0,1\}^{d(\secparam)}$, for every $n$-qubit Pauli $P_x$ and $b \in\{0,1\}$, define the distribution $\distr_{P_x,x,b}$ as follows: sample $k \xleftarrow{\$} \{0,1\}^{\secparam}$ and output $\bfx=\left(P_x,k,x,b \right)$. 

\begin{lemma}
\label{lem:samekey:first} Let $L= O(2^{3n} \secparam)$. The verifiable tomography scheme $(\tomography,\verify)$ described in~\Cref{fig:firstinstantiation} satisfies $(L,\distr_{P_x,x,b})$-distributional same-input correctness for all $P_x, x, b$. 
\end{lemma}
\begin{proof}
Define ${\cal K}_{\sf good} \subseteq \{0,1\}^{\secparam}$ such that for every $k \in {\cal K}_{\sf good}$ and $x \in \{0,1\}^d$, $\widehat{G}_{\secparam}(k,x)$ can be written as $\eta_{k,x} \left(\ketbra{0}{0} \otimes \ketbra{\psi_{k,x}}{\psi_{k,x}}\right) + (1 - \eta_{k,x}) \ketbra{\bot}{\bot}$, where $\eta_{k,x} \geq 1 - \negl(\secparam)$ for all $x\in\{0,1\}^{d}$. From the fact that $\{G_{\secparam}(\cdot,\cdot)\}$ is a PRFS, it follows that $|{\cal K}_{\sf good}| \geq (1 - \negl(\secparam))2^{\secparam}$. 
\par Fix $k \in {\cal K}_{\sf good}$. Let $x \in \{0,1\}^{d(\secparam)}$, $P_x$ be an $n$-qubit Pauli and $b \in\{0,1\}$. Set $\bfx=(P_x,k,x,b)$. Let $M \leftarrow \tomography\left( \rho^{\otimes L} \right)$, where $\rho = \Phi_{\secparam}(\bfx)$. We now argue that the probability that $\verify\left(\bfx,M\right)$ outputs $\valid$ is negligibly close to 1. 
\par Let $\widehat{M} \leftarrow \tomography\left( \rho^{\otimes L} \right)$, where $\rho=\Phi_{\secparam}(\bfx)$, be generated during the execution of $\verify\left(\bfx,M\right)$. Conditioned on the event that both $\|M - \rho\|_F^2 \leq \frac{1}{729}$ and $\|\widehat{M} - \rho\|_F^2 \leq \frac{1}{729}$ holds, we argue that $\bra{\bot} M \ket{\bot} \leq \frac{1}{2}$. Consider the following cases. 

\begin{itemize}
    \item Case $b=0$: In this case, $\rho= \left(\eta_{k,x} \left(\ketbra{0}{0} \otimes \ketbra{\psi_{k,x}}{\psi_{k,x}}\right) + (1 - \eta_{k,x}) \ketbra{\bot}{\bot}\right)$. Since $(\bra{0}\bra{\psi_{k,x}})\ket{\bot} = 0$ (from~\Cref{def:classicalgen:strongrec}), it follows that $\bra{\bot} \rho \ket{\bot} = (1-\eta_{k,x}) \leq \negl(\secparam)$. Since, we have that $\|M - \rho\|_F^2 \leq \frac{1}{729}$, by~\Cref{fact:frob_inequality}, we get $\bra{\bot}M\ket{\bot} \leq \negl(\secparam) + \frac{1}{27} + \sqrt{\frac{2-2\negl(\secparam)}{729}}\leq \frac{1}{9}$.

    %\item Case $b=0$: In this case, $\rho_x= \left(\eta_{k,x} \ketbra{0}{0} \otimes \ketbra{\psi_{k,x}}{\psi_{k,x}} + (1 - \eta_{k,x}) \ketbra{\bot}{\bot}\right)$. Since $(\bra{0}\bra{\psi_{k,x}})\ket{\bot} = 0$ (from~\Cref{def:classicalgen:strongrec}), it follows that $\bra{\bot} \rho_x \ket{\bot} = (1-\eta_{k,x}) \leq \negl(\secparam)$. Since, we have that $\|M_x - \rho_x\|_F^2 \leq \frac{1}{729}$, by~\Cref{fact:frob_inequality}, we get $\bra{\bot}M_x\ket{\bot} \leq \frac{1}{27} \leq \frac{1}{2}$.
    \item Case $b=1$: In this case, $\rho_x= (I\otimes P_x)\left(\eta_{k,x} \left(\ketbra{0}{0} \otimes \ketbra{\psi_{k,x}}{\psi_{k,x}}\right) + (1 - \eta_{k,x}) \ketbra{\bot}{\bot}\right)(I\otimes P_x)$. Since $(\bra{0}\bra{\psi_{k,x}})\ket{\bot} = 0$ (from~\Cref{def:classicalgen:strongrec}), it follows that $\bra{\bot} \rho_x \ket{\bot} = (1-\eta_{k,x}) \bra{\bot}(I\otimes P_x)\ket{\bot}\bra{\bot}(I\otimes P_x)\ket{\bot}$. Since, for any unitary $A$ and any state $\ket{\phi}$, we have $\bra{\phi}A\ket{\phi} \leq \braket{\phi|\phi}$, we get $\bra{\bot} \rho_x \ket{\bot} \leq (1-\eta_{k,x}) \leq \negl(\secparam)$. Similar to the above case, we get $\bra{\bot}M\ket{\bot} \leq \negl(\secparam) + \frac{1}{27} + \sqrt{\frac{2-2\negl(\secparam)}{729}}\leq \frac{1}{9}$.
\end{itemize}

From~\Cref{cor:tomography}, it follows that (a) $\prob[\|M - \rho\|_F^2 \leq \frac{92^{(n+1)}}{3^{8} 2^{(n+1)}} \leq \frac{1}{729}] \geq 1-\negl(\secparam)$ and similarly, (b) $\prob[\|\widehat{M} - \rho\|_F^2 \leq \frac{1}{729}]\geq 1 - \negl(\secparam)$, where the probability is over the randomness of $\tomography$. Thus, it follows that $\prob[\|M - \widehat{M}\|_{F}^2 \leq \frac{4}{729}] \geq 1-\negl(\secparam)$.

\end{proof}

% \pnote{I think the lemma below is incorrect.. there might be bad keys for which the $\ket{\bot}$ component might be significant.}

% \begin{lemma}
% \label{clm:tomography}
% Let $k \in \{0,1\}^{\secparam}$ and $b \in \{0,1\}$. Let $P$ be an $m$-qubit Pauli of the form $P = \bigotimes_{x\in \{0,1\}^d} P_{x}$. Compute $\rho^{\otimes L} \leftarrow \left( \Phi_\secparam\left(P\|k\|b\right) \right)^{\otimes L}$, where $L = O(3^{12} 2^{4d})$. Compute $\tomography\left( \rho^{\otimes L} \right)$ to get $M$. Then, $\verify\left(P\|k\|b,M\right)$ outputs $\valid$ with probability 1. \pnote{change the statement}
% \end{lemma}
% \begin{proof}
% Let the output of $\tomography\left( \rho^{\otimes L} \right)$ in the execution of $\verify\left(P\|k\|b,M\right)$ be $\widehat{M}$. From the guarantee of the above tomography procedure, we have that $\|M - \rho \|_F^2 \leq \frac{1}{729}$ and $\|\widehat{M} - \rho\|_F^2 \leq \frac{1}{729}$. By triangle inequality, we have $\|M - \widehat{M}\|_{F}^2 \leq \frac{4}{729}$. \pnote{I think it is technically incorrect to write triangle inequality since it is a matrix norm and not a distance metric. Maybe we can say by definition?}
% \end{proof}

\paragraph{Distributional Different-Input Correctness.} We prove below that $\left(\tomography,\verify \right)$ satisfies $\left(\varepsilon,L,\Pi,\distr_{x}\right)$-different-input correctness, where $\Pi$ and $\distr_x$ are defined as follows:
$$
  \Pi\left(\left(P_0,k_0,x_0,b_0\right),\left(P_1,k_1,x_1,b_1\right)\right) =
  \begin{cases}
                                   0 & P_0=P_1, x_0=x_1\text{ and } b_0\neq b_1, \\
                                   1 & \text{otherwise.}
  \end{cases}
$$
The sampler for $\distr_{x}$ is defined as follows: sample $P_x \xleftarrow{\$} \mathcal{P}_n$, $k_0,k_1 \xleftarrow{\$} \{0,1\}^{\secparam}$ and output $(\left(P_x,k_0,x,0 \right),\allowbreak (\left(P_x,k_1,x,1 \right))$. 
\noindent We first prove an intermediate lemma that will be useful for proving distributional different-input correctness. Later on, this lemma will also be useful in the application of bit commitments.

% While there is a direct proof for different-input correctness, we first prove a stronger statement about the behaviour of $\verify$ on different inputs that we will use in the application section too. \ag{Is this line okay or should we just say "To prove different-input correctness, we first prove a stronger statement about the behaviour of $\verify$ on different inputs"?}

\begin{lemma}
\label{lem:diff_key_binding}
Let $P_x\in\mathcal{P}_n$ and there exists a density matrix $M$ such that  $\verify(P_x\|k_0\|x\|0,M) = \valid$ and $\verify(P_x\|k_1\|x\|1,M) = \valid$, for some $k_0,k_1\in\{0,1\}^{\secparam}$. Then $$\Tr\left(P_x\ketbra{\psi_{k_1,x}}{\psi_{k_1,x}}P_x\ketbra{\psi_{k_0,x}}{\psi_{k_0,x}}\right)\geq \frac{542}{729}.$$
\end{lemma}

\begin{proof}
Since $\verify(P_x\|k_0\|x\|0,M) = \valid$, therfore $$\bra{\bot}M\ket{\bot}\leq \frac{1}{9}$$ and $$\|M-M_0\|_F^2\leq \frac{4}{729},$$ where $M_0 = \tomography\left( \left( \Phi_\secparam\left(P_x\|k_0\|x\|0\right) \right)^{\otimes L} \right)$. 

Similarly, since $\verify(P_x\|k_1\|x\|1,M) = \valid$, $$\bra{\bot}(I\otimes P_x)M(I\otimes P_x)\ket{\bot}\leq \frac{1}{9}$$ and $$\|M-M_1\|_F^2\leq \frac{4}{729},$$ where $M_1 = \tomography\left( \left( \Phi_\secparam\left(P_x\|k_1\|x\|1\right) \right)^{\otimes L} \right)$. 

Since, $M_0 = \tomography\left( \left( \Phi_\secparam\left(P_x\|k_0\|x\|0\right) \right)^{\otimes L} \right)$ and $M_1 = \tomography\left( \left( \Phi_\secparam\left(P_x\|k_1\|x\|1\right) \right)^{\otimes L} \right)$, $$\|M_0 - \left(\eta_0\left(\ketbra{0}{0}\otimes\ketbra{\psi_{k_0,x}}{\psi_{k_0,x}}\right) + (1-\eta_0)\ketbra{\bot}{\bot}\right)\|_F^2\leq \frac{1}{729},$$ and $$\|M_1 - (I\otimes P_x)\left(\eta_1\left(\ketbra{0}{0}\otimes\ketbra{\psi_{k_1,x}}{\psi_{k_1,x}}\right) + (1-\eta_1)\ketbra{\bot}{\bot}\right)(I\otimes P_x)\|_F^2\leq \frac{1}{729}.$$
By triangle inequality, $$\|M - \left(\eta_0\left(\ketbra{0}{0}\otimes\ketbra{\psi_{k_0,x}}{\psi_{k_0,x}}\right) + (1-\eta_0)\ketbra{\bot}{\bot}\right)\|_F^2\leq \frac{1}{81},$$ and $$\|M - (I\otimes P_x)\left(\eta_1\left(\ketbra{0}{0}\otimes\ketbra{\psi_{k_1,x}}{\psi_{k_1,x}}\right) + (1-\eta_1)\ketbra{\bot}{\bot}\right)(I\otimes P_x)\|_F^2\leq \frac{1}{81}.$$
By Fact~\ref{fact:frob_inequality}, $$\bra{\bot}\left(\eta_0\ketbra{0}{0}\otimes\ketbra{\psi_{k_0}}{\psi_{k_0}} + (1-\eta_0)\ketbra{\bot}{\bot}\right)\ket{\bot}\leq 10/27,$$ and $$ \bra{\bot}(I\otimes P_x)\left(\eta_1\ketbra{0}{0}\otimes\ketbra{\psi_{k_1,x}}{\psi_{k_1,x}} + (1-\eta_1)\ketbra{\bot}{\bot}\right)(I\otimes P_x)\ket{\bot}\leq 10/27.$$ 
Simplifying, we get $$\eta_0\geq \frac{17}{27},$$ and $$\eta_1\geq \frac{17}{27}.$$ 
Also, by triangle inequality, 
\begin{multline*}
    \|(I\otimes P_x)\left(\eta_1\ketbra{0}{0}\otimes\ketbra{\psi_{k_1,x}}{\psi_{k_1,x}} + (1-\eta_1)\ketbra{\bot}{\bot} \right)(I\otimes P_x) \\
    - \left(\eta_0\ketbra{0}{0}\otimes\ketbra{\psi_{k_0,x}}{\psi_{k_0,x}} + (1-\eta_0)\ketbra{\bot}{\bot}\right)\|_F^2\leq 4/81.
\end{multline*}
Or, 
\begin{multline*}
    \|\eta_1(I\otimes P_x)\ketbra{0}{0}\otimes\ketbra{\psi_{k_1,x}}{\psi_{k_1,x}}(I\otimes P_x) - \eta_0\ketbra{0}{0}\otimes\ketbra{\psi_{k_0,x}}{\psi_{k_0,x}}  \\
    - \left((1-\eta_0)\ketbra{\bot}{\bot} - (1-\eta_1)(I\otimes P_x) x\ketbra{\bot}{\bot}(I\otimes P_x)\right)\|_F^2\leq 4/81.
\end{multline*}
By~\Cref{fact:frab:square}, since $\bra{\psi_{0}}(\ket{0}\ket{\psi_k}) = 0$,
\begin{multline*}
    \|\eta_1(I\otimes P_x)\ketbra{0}{0}\otimes\ketbra{\psi_{k_1,x}}{\psi_{k_1,x}}(I\otimes P_x) - \eta_0\ketbra{0}{0}\otimes\ketbra{\psi_{k_0,x}}{\psi_{k_0,x}}\|_F^2 \\
    + \|\left((1-\eta_0)\ketbra{\bot}{\bot} - (1-\eta_1)(I\otimes P_x) \ketbra{\bot}{\bot}(I\otimes P_x)\right)\|_F^2\leq 4/81.
\end{multline*}
Or, 
$$\|\eta_1(I\otimes P_x)\ketbra{0}{0}\otimes\ketbra{\psi_{k_1,x}}{\psi_{k_1,x}}(I\otimes P_x) - \eta_0\ketbra{0}{0}\otimes\ketbra{\psi_{k_0,x}}{\psi_{k_0,x}}\|_F^2 \leq 4/81.$$

By~\Cref{fact:frab:square}, 
$$\eta_1^2 + \eta_0^2 - \eta_1\eta_0\Tr\left(P_x\ketbra{\psi_{k_1,x}}{\psi_{k_1,x}}P_x\ketbra{\psi_{k_0,x}}{\psi_{k_0,x}}\right) \leq 4/81.$$

Hence, $$\Tr\left(P_x\ketbra{\psi_{k_1,x}}{\psi_{k_1,x}}P_x\ketbra{\psi_{k_0,x}}{\psi_{k_0,x}}\right)\geq \frac{542}{729}.$$
\end{proof}

With the above lemma in mind, we can prove the different-input correctness.
\begin{lemma}
\label{lem:diffkey:first}
$(\tomography,\verify)$ described in~\Cref{fig:firstinstantiation} satisfies $\left(O(2^{-n}),L,\Pi,\distr_{x}\right)$-different-input correctness, where $L = O(2^{3n} \secparam)$.
\end{lemma}
\begin{proof}
Define $\distr'_x$ as follows: sample $P_x \xleftarrow{\$} \mathcal{P}_n$, $k_0,k_1 \xleftarrow{\$} {\cal K}_{\sf good}$ such that $k_0\neq k_1$ and output $\left(\left(P_x,k_0,x,0 \right),(\left(P_x,k_1,x,1 \right)\right)$. 
Let
\begin{multline*}
    p = \prob[ \valid \leftarrow \verify\left(P_x||k_0||x||0 , \tomography \left(\Phi_\secparam\left(P_x||k_1||x||1 \right)\right)^{\otimes L(\secparam)} \right)\ :\\
    \ \left(\left(P_x,k_0,x,0 \right),\left(P_x,k_1,x,1 \right)\right) \leftarrow \distr_{x}  ], 
\end{multline*}
and
\begin{multline*}
    p' = \prob[ \valid \leftarrow \verify\left(P_x||k_0||x||0 , \tomography \left(\Phi_\secparam\left(P_x||k_1||x||1 \right)\right)^{\otimes L(\secparam)} \right) \ :\\
    \ \left(\left(P_x,k_0,x,0 \right),\left(P_x,k_1,x,1 \right)\right) \leftarrow \distr'_{x}  ]
\end{multline*}
Then, since $|{\cal K}_{\sf good}| \geq (1 - \negl(\secparam))2^{\secparam}$, we know that $$p = p'\cdot\prob[ k_0,k_1\in  {\cal K}_{\sf good}:\ \left(\left(P_x,k_0,x,0 \right),\left(P_x,k_1,x,1 \right)\right) \leftarrow \distr_{x}  ] + \negl(\secparam),$$
or 
$$p = p'(1-\negl(\secparam)) + \negl(\secparam).$$
By~\Cref{lem:samekey:first}, we know that $$\prob[\valid\leftarrow\verify\left(P_x||k_1||x||1,\tomography \left(\Phi_\secparam\left(P_x||k_1||x||1 \right)\right)^{\otimes L(\secparam)}\right) :\ k_1\xleftarrow{\$}{\cal K}_{\sf good}] \geq 1-\negl(\secparam).$$
Hence, $$p' = \prob\left[\substack{ \verify(P_x||k_0||x||0,M_x) = \valid,\\ \ \\ \verify(P_x||k_1||x||1,M_x) = \valid  }:\substack{ \left(\left(P_x,k_0,x,0 \right),\left(P_x,k_1,x,1 \right)\right) \leftarrow \distr'_{x}\\ \ \\ M_x = \tomography \left(\Phi_\secparam\left(P_x||k_1||x||1 \right)\right)^{\otimes L(\secparam)}}\right]\cdot (1-\negl(\secparam))+\negl(\secparam).$$
By~\Cref{lem:diff_key_binding}, $$p'\leq \prob\left[\Tr\left(P_x\ketbra{\psi_{k_1,x}}{\psi_{k_1,x}}P_x\ketbra{\psi_{k_0,x}}{\psi_{k_0,x}}\right)\geq \frac{542}{729}:\left(\left(P_x,k_0,x,0 \right),\left(P_x,k_1,x,1 \right)\right) \leftarrow \distr'_{x}\right]\cdot (1-\negl(\secparam))+\negl(\secparam),$$
Or 
$$p'\leq \prob\left[\left|\bra{\psi_{k_1,x}}P_x\ket{\psi_{k_0,x}}\right|^2\geq \frac{542}{729}:\left(\left(P_x,k_0,x,0 \right),\left(P_x,k_1,x,1 \right)\right) \leftarrow \distr'_{x}\right]\cdot (1-\negl(\secparam))+\negl(\secparam).$$
We use the following fact \cite[Fact 6.9]{AQY21}: Let $\ket{\psi}$ and $\ket{\phi}$ be two arbitrary $n$-qubit states. Then, $$\E_{P_x\xleftarrow{\$}\mathcal{P}_n} \left[\left|\bra{\psi}P_x\ket{\phi}\right|^{2}\right]=2^{-n}.$$
For any $k_0,k_1,x$ by the above fact, $\E_{P_x\xleftarrow{\$}\mathcal{P}_n} \left[\left|\bra{\psi_{k_0,x}}P_x\ket{\psi_{k_1,x}}\right|^{2}\right]=2^{-n}.$ Using Markov’s inequality we get that for all $\delta>0$, $$\Pr_{P_x\xleftarrow{\$}\mathcal{P}_n}\left[\left|\bra{\psi_{k_0,x}}P_x\ket{\psi_{k_1,x}}\right|^2\geq\delta\right]\leq \delta^{-1}2^{-n}.$$
Hence, $$p'\leq\frac{729}{542}2^{-n}\cdot (1-\negl(\secparam))+\negl(\secparam),$$ and $$p\leq\frac{729}{542}2^{-n}\cdot (1-\negl(\secparam))+\negl(\secparam).$$
Hence, the scheme satisfies satisfies $\left(O(2^{-n}),L,\Pi,\distr_{x}\right)$-different-input correctness.

\end{proof}

\subsubsection{Second Instantiation}
\label{sec:secondinstantiation}
Similar to the first instantiation, we will start with a $(d(\secparam)+1,n(\secparam))$-PRFS  $\{G_{\secparam}\left(\cdot,\cdot \right)\}$ satisfying recognizable abort property (\Cref{def:classicalgen:strongrec}). We set $d(\secparam) = \lceil \log(\secparam) \rceil$ and $n(\secparam)=\lceil \log(\secparam) \rceil$.

\par We will describe the algorithms $(\tomography,\verify)$ in Figure~\ref{fig:secondinstantiation}. The set of channels ${\cal C}=\{\Phi_{\secparam}:\secparam \in \mathbb{N}\}$ is associated with $(\tomography,\verify)$, where $\Phi_{\secparam}$ is defined as follows: 
\begin{itemize}
    \item Let the input be initialized on register ${\bf A}$.
    \item Controlled on the first register containing the value $(k,i,b) \in \{0,1\}^{\ell(\secparam)}$, where $k \in \{0,1\}^{\secparam},i \in \{0,1\}^d, b \in \{0,1\}$, do the following: compute $ G_{\secparam}(k,i\|b)$ and store the result in the register ${\bf B}$. 
    \item Trace out ${\bf A}$ and output ${\bf B}$. 
\end{itemize}

%$$\Phi_\secparam\left(k\|i\|b \right) = G_{\secparam}(k,i\|b)$$

\begin{figure}[!htb]
\begin{tabular}{|p{16cm}|}
\hline 
\ \\
\noindent \underline{$\tomography(\rho^{\otimes L})$}: On input $L$ copies of an  $2^n$-dimensional density matrix $\rho$, compute ${\cal T}_{\sf imp}(\rho^{\otimes L})$ to obtain $M$, where ${\cal T}_{\sf imp}$ is given in Corollary~\ref{cor:tomography}. Output $M$.  \\
\ \\

% \paragraph{Tomography Procedure.} We will use a tomography procedure  $\tomography$ satisfying the following property: for any state $\rho$, $\tomography\left( \rho^{\otimes O(l^4/\varepsilon^2)} \right)$ outputs the $l \times l$ matrix $M$ such that $\|M - \rho\|_{F}^2 \leq \varepsilon$. \pnote{how do we instantiate this?} \ag{Can we use Lowe's thesis, \href{https://uwspace.uwaterloo.ca/bitstream/handle/10012/17663/LOWE_ANGUS.pdf?sequence=3&isAllowed=y}{this}, 1.5.3 } \pnote{yes, we should add a reference}\\

% \noindent We claim that a tomography procedure satisfying the above property admits a natural verification procedure $\verify$ that will be described below. First, we define the class of channels associated with $(\tomography,\verify)$. 

\noindent \underline{$\verify(\bfx,M)$}: 
\begin{enumerate}
    \item Run $\rho^{\otimes L} \leftarrow \left( \Phi_\secparam\left(\bfx\right) \right)^{\otimes L}$, where $L=2^{3n+11} \secparam$.
    \item Run $\tomography\left( \rho^{\otimes L} \right)$ to get $\widehat{M}$.
    \item If $\|M - \widehat{M}\|_{F}^2 \leq \frac{9}{128}$, output $\valid$. Output $\invalid$ otherwise.
\end{enumerate}
\\ 
\hline
\end{tabular}
\caption{Second instantiation of $\tomography$}
\label{fig:secondinstantiation}
\end{figure}

\paragraph{Same-Input Correctness.} We prove below that $(\tomography,\verify)$ satisfies same-input correctness. 

\begin{lemma}
\label{lem:samekey:second}
$(\tomography,\verify)$ described in~\Cref{fig:secondinstantiation} satisfies $L$-same-input correctness.
%Let $k \in \{0,1\}^{\secparam}$ and $b \in \{0,1\}$. Let $P$ be an $m$-qubit Pauli of the form $P = \bigotimes_{x\in \{0,1\}^d} P_{x}$. Compute $\rho^{\otimes L} \leftarrow \left( \Phi_\secparam\left(P\|k\|b\right) \right)^{\otimes L}$, where $L = O(3^{12} 2^{4d})$. Compute $\tomography\left( \rho^{\otimes L} \right)$ to get $M$. Then, $\verify\left(P\|k\|b,M\right)$ outputs $\valid$ with probability 1. \pnote{change the statement}
\end{lemma}
\begin{proof}
Suppose $M \leftarrow \tomography\left(\rho^{\otimes L} \right)$, where $\rho = \Phi_{\secparam}(\bfx)$ for some $\bfx \in \{0,1\}^{\ell(\secparam)}$. We prove that $\verify\left( \bfx,M\right)$ outputs  $\valid$ with overwhelming probability. Let $\widehat{M} \leftarrow \tomography\left(\rho^{\otimes L} \right)$ be generated during the execution of $\verify\left( \bfx,M\right)$.\\

\noindent Let us condition on the event that $\|M-\rho\|_F^2 \leq \frac{9}{512}$ and $\|\widehat{M}-\rho\|_F^2 \leq \frac{9}{512}$. Using triangle inequality, we have that $\|M-\widehat{M}\|_F^2 \leq \frac{9}{128}$.\\

\noindent All that is left is to prove that $\|M-\rho\|_F^2 \leq \frac{9}{512}$ and $\|\widehat{M}-\rho\|_F^2 \leq \frac{9}{512}$ holds with overwhelming probability. We can invoke~\Cref{cor:tomography} since the dimension of $\rho$ is $2^n$ and the number of copies of $\rho$ used in tomography is $4\cdot 2^{n+9} \cdot (2^n)^2 \cdot \secparam$. In more detail, we have the following: 
\begin{eqnarray*}
\prob\left[\|M-\rho\|_F^2 \leq \frac{9}{512} \text{ and } \|\widehat{M}-\rho\|_F^2 \leq \frac{9}{512}\right] & = &  \prob\left[\|M-\rho\|_F^2 \leq \frac{9}{512}\right] \prob\left[ \|\widehat{M}-\rho\|_F^2 \leq \frac{9}{512}\right] \\ 
& \geq & (1 - \negl(\secparam))(1- \negl(\secparam))\ \ \text{(\Cref{cor:tomography})}\\
& \geq & (1 - \negl(\secparam))
\end{eqnarray*}
%\pnote{use corollary 6.2} Let the output of $\tomography\left( \rho^{\otimes L} \right)$ in the execution of $\verify\left(P\|k\|b,M\right)$ be $\widehat{M}$. From the guarantee of the above tomography procedure, we have that $\|M - \rho \|_F^2 \leq \frac{1}{729}$ and $\|\widehat{M} - \rho\|_F^2 \leq \frac{1}{729}$. By triangle inequality, we have $\|M - \widehat{M}\|_{F}^2 \leq \frac{4}{729}$. \pnote{I think it is technically incorrect to write triangle inequality since it is a matrix norm and not a distance metric. Maybe we can say by definition?} \ag{That works. Though we generally call it triangle inequality in case of norms too.}
\end{proof}

\paragraph{Distributional Different-Input Correctness.} We prove below that $\left(\tomography,\verify \right)$ satisfies different-input correctness. 

\begin{lemma}
\label{lem:diffkey:second}
Assuming the security of $\{G_{\secparam}(\cdot,\cdot)\}$, $(\tomography,\verify)$ described in~\Cref{fig:secondinstantiation} satisfies $(L,\Pi,\distr_{i})$-different-input correctness, for every $i \in \{0,1\}^d$, where $\Pi:\{0,1\}^{\ell(\secparam)} \times \{0,1\}^{\ell(\secparam)} \rightarrow \{0,1\}$ and a distribution $\distr_{i}$ are defined as follows:
\begin{itemize}[noitemsep]
    \item $\Pi((k,i,b),(k',i',b'))=0$ if and only if $k=k'$, $i=i'$ and $b \neq b'$,
    \item $\distr_{i}$ is a distribution that samples $k \xleftarrow{\$} \{0,1\}^{\secparam}$ and outputs $((k,i,0),(k,i,1))$. 
\end{itemize} 
\end{lemma}
\begin{proof}
Define ${\cal K}_{\sf good} \subseteq \{0,1\}^{\secparam}$ as follows: for every $k \in \{0,1\}^{\secparam}$, $k \in {\cal K}_{\sf good}$ if and only if for every $i \in \{0,1\}^d$, $\|G_{\secparam}(k,i\|0) - G_{\secparam}(k,i\|1)\|_F^2 > \frac{81}{512}$. 
\par We show the following. 

\begin{claim}
For every $k \in {\cal K}_{\sf good}$, for every bit $b$,
$$\prob\left[\invalid \leftarrow \verify((k,i,b),M) \ :\ M \leftarrow \tomography((\Phi_{\secparam}((k,i,1-b))^{\otimes L} )) \right] \geq 1-\negl(\secparam)$$
\end{claim}
\begin{proof}
We show this for the case when $b=0$; the argument for the case when $b=1$ symmetrically follows. Let $\widehat{M} \leftarrow \tomography((\Phi_{\secparam}(k,i,0))^{\otimes L})$ be generated during the execution of $\verify((k,i,0),M)$.\\

\noindent Let us condition on the event that $\|\widehat{M} - G_{\secparam}(k,i,0) \|_F^2 \leq \frac{9}{512}$ and $\|\widehat{M} - G_{\secparam}(k,i,1) \|_F^2 \leq \frac{9}{512}$. We prove that $\|M - \widehat{M}\|_F^2 > \frac{9}{512}$. Suppose not, then the following holds:
\begin{eqnarray*}
\|G_{\secparam}(k,i\|0) - G_{\secparam}(k,i\|1)\|_F & \leq & \|\widehat{M} - G_{\secparam}(k,i,0) \|_F + \|\widehat{M} - G_{\secparam}(k,i,1) \|_F + \|M - \widehat{M}\|_F \\
& \leq & \sqrt{\frac{9}{512}} + \sqrt{\frac{9}{512}} + \sqrt{\frac{9}{512}} \leq 3 \sqrt{\frac{9}{512}}
\end{eqnarray*}
Thus, $\|G_{\secparam}(k,i\|0) - G_{\secparam}(k,i\|1)\|_F^2 \leq \frac{81}{512}$. This contradicts the fact that $\|G_{\secparam}(k,i\|0) - G_{\secparam}(k,i\|1)\|_F^2 > \frac{81}{512}$. 
\par To summarise, conditioned on the event that $\|M - G_{\secparam}(k,i,0) \|_F^2 \leq \frac{9}{512}$ and $\|\widehat{M} - G_{\secparam}(k,i,1) \|_F^2 \leq \frac{9}{512}$, it holds that $\|M - \widehat{M}\|_F^2 > \frac{9}{512}$. Thus, $\verify$ outputs $\invalid$. \\

\noindent As we showed in the proof of~\Cref{lem:samekey:second}, it follows that $\|\widehat{M} - G_{\secparam}(k,i,0) \|_F^2 \leq \frac{9}{512}$ and $\|\widehat{M} - G_{\secparam}(k,i,1) \|_F^2 \leq \frac{9}{512}$ holds with probability at least $(1-\negl(\secparam))$. \\

\par Combining the above observations, we have the following: for every $k \in {\cal K}_{\sf good}$, for every bit $b$, we show that $\verify((k,i,0),M)$ outputs $\invalid$ with probability at least $(1-\negl(\secparam))$. This completes the proof. 
\end{proof}

\noindent All that is left is to show that $|{\cal K}_{\sf good}|$ is large enough.  

\begin{claim}
$|{\cal K}_{\sf good}| \geq (1-\negl(\secparam))2^{\secparam}$.
\end{claim}
\begin{proof}
We invoke the security of $\{G_{\secparam}(\cdot,\cdot)\}$ to show this. Suppose the statement of the claim is false (that is $|{\cal K}_{\sf good}| \leq (1-\delta) 2^{\secparam}$, where $\delta$ is non-negligible), we prove that $\{G_{\secparam}(\cdot,\cdot)\}$ is insecure. \par Consider the following QPT distinguisher $D$ that distinguishes whether it is given (classical) oracle access to $G_{\secparam}(\cdot,\cdot)$ or whether it is given access to an oracle, call it ${\cal O}_{\sf Haar}$, that on any input $x \in \{0,1\}^d$, outputs an iid Haar random state. $D$ queries the oracle on all the inputs $\{0,1\}^{d+1}$ and for every input $x \in \{0,1\}^d$, it obtains $L$ copies of two states $\rho_{x\|0}$ and $\rho_{x\|1}$. It then computes $M_{x\|0} \leftarrow \tomography\left( \rho_{x\|0}^{\otimes L} \right)$ and $M_{x\|1}  \leftarrow \tomography \left( \rho_{x\|1}^{\otimes L} \right)$. It outputs 1 if there exists $x \in \{0,1\}^d$ such that $\|M_{x\|0} - M_{x\|1} \|_F^2 \leq \frac{15}{\sqrt{512}}$, otherwise it outputs 0.
\par We consider two cases below. 
\begin{itemize}
    \item $D$ has oracle access to ${\cal O}_{\sf Haar}$: for any $x \in \{0,1\}^d$,  from~\Cref{fact:frobenius:haar}, it follows that $\prob[\| \rho_{x\|0} - \rho_{x\|1} \|_F^2 \leq \frac{15}{\sqrt{512}}] \leq \negl(\secparam)$. By union bound, it follows that the probability that there exists an $x \in \{0,1\}^d$ such that $\| \rho_{x\|0} - \rho_{x\|1} \|_F^2 \leq \frac{15}{\sqrt{512}}$ is negligible in $\secparam$. Thus, $D$ outputs 1 with negligible probability. 
    \item $D$ has oracle access to $G_{\secparam}(k,\cdot)$, where $k \leftarrow \{0,1\}^{\secparam}$: Let us condition on the event that PRFS key $k \notin {\cal K}_{\sf good}$. This means that there exists an input $x \in \{0,1\}^d$ such that $\|G_{\secparam}(k,i\|0) - G_{\secparam}(k,i\|1)\|_F^2 \leq \frac{81}{512}$. Moreover, from~\Cref{cor:tomography}, it follows that with probability at least $(1-\negl(\secparam))$, $\|M_{x\|0} - G_{\secparam}(k,x\|0)\|_F^2 \leq \frac{9}{512}$ and $\|M_{x\|1} - G_{\secparam}(k,i\|1)\|_F^2 \leq \frac{9}{512}$. Thus, with probability at least $1 - \negl(\secparam)$, we have the following:
\begin{eqnarray*}
\|M_{x\|0} - M_{x\|1}\|_F & \leq & \|M_{x\|0} - G_{\secparam}(k,x,0) \|_F + \|M_{x\|1} - G_{\secparam}(k,x,1) \|_F + \| G_{\secparam}(k,x,0) - G_{\secparam}(k,x,1) \|_F \\
& \leq & \sqrt{\frac{9}{512}} + \sqrt{\frac{9}{512}} + \sqrt{\frac{81}{512}} = \frac{15}{\sqrt{512}}
\end{eqnarray*}
\noindent Since the probability that $k \notin {\cal K_{\sf good}}$ is at least $\delta$, we have that the probability that $D$ outputs 1 is at least $\delta(1-\negl(\secparam))$, which is non-negligible in $\secparam$. 
\end{itemize}
Since the difference in the probability that $D$ outputs 0 in both the above cases is at least non-negligible, we have that $D$ violates the security of PRFS. 
\end{proof}

\end{proof}

% \paragraph{$\Pi$-Perfect Soundness.}\ \pnote{Incorporate this in the completeness} This is associated with a predicate $\Pi:\{0,1\}^{\ell} \times \{0,1\}
% ^{\ell} \rightarrow \{0,1\}$. For every $k \in \{0,1\}^{\ell(\secparam)}$, $k' \in \{0,1\}^{\ell(\secparam)}$ such that $\Pi(k,k') = 0$, we require the following:
% $$\prob\left[  \invalid \leftarrow \verify\left(k', \tomography\left( (\Phi_\secparam(k))^{\otimes \poly(\secparam)} \right) \right)  \right] \geq 1 - {\sf negl}(\secparam)  $$

% \noindent Some commonly considered predicates:
% \begin{itemize}
%     \item $\Pi(k,k')=1$ if and only if $k=k'$. 
%     \item In the context of commitments, $k$ should be interpreted as the concatenation of the PRFS key, input to the PRFS and the bit to be committed. In this case, we would define $\Pi$ to be the following: $\Pi(k,k') = 1$ if and only if the last bit of $k$ and $k'$ agree. 
% \end{itemize}

% \noindent We can define a computational variant of the above soundness guarantee. Here, we would require that any QPT adversary cannot come up with $k$ and $k'$ such that the following holds: 
% $$\prob\left[  \invalid \leftarrow \verify\left(k', \tomography\left( (\phi_z^{(\secparam)}(k))^{\otimes \poly(\secparam)} \right) \right)  \right] \geq 1 - {\sf negl}(\secparam)  $$

% \noindent \pnote{Still not 100\% sure how the soundness guarantee should be defined.}
\section{Applications}
\label{sec:apps}
\noindent In this section, we show how to use PRFS to constrtuct a variety of applications: 
\begin{enumerate}
    \item Bit commitments with classical communication and,
    \item Pseudo one-time pad schemes with classical communication.
\end{enumerate}
\noindent To accomplish the above applications, we use verifiable tomography from~\Cref{sec:tomography}.  

\subsection{Commitment scheme}
\label{sec:commitment}
\noindent We construct bit commitments with classical communication from pseudorandom function-like quantum states. We recall the definition by~\cite{AQY21}. 
\par A (bit) commitment scheme is given by a pair of (uniform) QPT algorithms $(C,R)$, where $C=\{C_{\secparam}\}_{\secparam \in \mathbb{N}}$ is called the \emph{committer} and $R=\{R
_{\secparam}\}_{\secparam \in \mathbb{N}}$ is called the \emph{receiver}. There are two phases in a commitment scheme: a commit phase and a reveal phase. %Both phases are executed between $C=\{C_{\secparam}\}_{\secparam \in \mathbb{N}}$, a committer and $R=\{R
%_{\secparam}\}_{\secparam \in \mathbb{N}}$, the receiver. 
\begin{itemize}
    \item In the (possibly interactive) commit phase between $C_{\secparam}$ and $R_{\secparam}$, the committer $C_{\secparam}$ commits to a bit, say $b$. We denote the execution of the commit phase to be $\sigma_{CR} \leftarrow \commit \langle C_{\secparam}(b),R_{\secparam} \rangle$, where $\sigma_{CR}$ is a joint state of $C_{\secparam}$ and $R_{\secparam}$ after the commit phase. 
    \item In the %(non-interactive) 
    reveal phase $C_{\secparam}$ interacts with $R_{\secparam}$ and the output is a trit $\mu \in \{0,1,\bot \}$ indicating the receiver's output bit or a rejection flag. We denote an execution of the reveal phase where the committer and receiver start with the joint state $\sigma_{CR}$ by $\mu \leftarrow \reveal \langle C_\lambda, R_\lambda, \sigma_{CR} \rangle$. 
    
    %pair $(b',\bf d)$ where $b'$ is the the committed bit $b$ along with a quantum state (referred to as decommitment) 
\end{itemize}

\noindent We require that the above commitment scheme satisfies the correctness, computational hiding, and statistical binding properties below.  

\begin{definition}[Correctness]
We say that a commitment scheme $(C,R)$ satisfies correctness if $$\prob\left[ b^{\ast} = b \ :\ \substack{ \sigma_{CR}\xleftarrow {}\commit \langle C_{\secparam}(b),R_{\secparam}\rangle,\\ \ \\ b^{\ast} \leftarrow \reveal\langle C_{\secparam},R_{\secparam},\sigma_{CR} \rangle  } \right] \geq 1-\nu(\secparam),$$
where $\nu(\cdot)$ is a negligible function. 
\end{definition}

\begin{definition}[Computational Hiding]
We say that a commitment scheme $(C,R)$ satisfies computationally hiding if for any malicious QPT receiver $\left\{R^{\ast}_{\secparam}\right\}_{\secparam\in\mathbb{N}}$, for any QPT distinguisher $\left\{D_{\secparam}\right\}_{\secparam\in\mathbb{N}}$, the following holds: 
$$\left|\Pr_{(\tau,\sigma_{C{R^{\ast}}})\xleftarrow{}\commit \langle C_{\secparam}(0),R_{\secparam}^{\ast}\rangle}\left[D_{\secparam}(\sigma_{R^{\ast}}) = 1\right] - \Pr_{(\tau,\sigma_{CR^{\ast}})\xleftarrow{}\commit\langle C_{\secparam}(1),R_{\secparam}^{\ast}\rangle}\left[D_{\secparam}(\sigma_{R^{\ast}}) = 1\right] \right|\leq \varepsilon(\secparam),$$ for some negligible $\varepsilon(\cdot)$.
\end{definition}

\begin{definition}[Statistical Binding]
We say that a commitment scheme $(C,R)$ satisfies statistical binding if for every QPT sender $\{C^*_{\secparam}\}_{\secparam \in \mathbb{N}}$, there exists a (possibly inefficient) extractor ${\cal E}$ such that the following holds: 
$$\prob\left[ \mu \neq b^* \wedge \mu \neq \bot \ :\ \substack{ (\tau,\sigma_{C^*R})\xleftarrow {}\commit \langle C^*_{\secparam},R_{\secparam}\rangle,\\ \ \\ b^* \leftarrow {\cal E}(\tau),\\ \ \\ \mu \leftarrow \reveal\langle C_{\secparam}^* ,R_{\secparam} ,\sigma_{C^* R} \rangle  } \right] \leq \nu(\secparam),$$
where $\nu(\cdot)$ is a negligible function and $\tau$ is the transcript of the $\commit$ phase. 
%\hnote{In the AQY paper, the deifnition of statistical binding involves Real vs Ideal experiment. Is this definition equivalent?} \pnote{I think so. The reason why our definition was complicated in the previous paper was that when we apply the extractor on the receiver's state the state could change and since the receiver's state could be entangled with the sender' state, we had to argue that the change is undetectable.. however, in this case, the transcript is classical and so applying the extractor does not affect the sender's state. That's why the definition is much simpler.} \hnote{Ah, gotcha. Perhaps make this into a remark around here?} \pnote{sure, will do} 
\end{definition}

\begin{remark}[Comparison with~\cite{AQY21}]
In the binding definition of~\cite{AQY21}, given the fact that the sender's and the receiver's state could potentially be entangled with each other, care had to be taken to ensure that after the extractor was applied on the receiver's state, the sender's state along with the decision bit remains (indistinguishable) to the real world. In the above definition, however, since the communication is entirely classical, any operations performed on the receiver's end has no consequence to the sender's state. As a result, our definition is much simpler than~\cite{AQY21}.   
\end{remark}

\subsubsection{Construction}
Towards constructing a commitment scheme with classical communication, we use a verifiable tomography from~\Cref{fig:firstinstantiation}.

\paragraph{Construction.} We present the construction in Figure~\ref{fig:commitment}. In the construction, we require  $d(\secparam)=\lceil \log\frac{ 3\secparam}{n} \rceil \geq 1$.

% As stated earlier, we use pseudorandom function-like states in conjunction with a specific instantiation of verifiable tomography to obtain bit commitments. We state the two tools below. 
% \begin{itemize}
%     \item $(d(\secparam),n(\secparam))$-Pseudorandom function-like states,  $\{G_{\secparam}\left(\cdot,\cdot \right)\}$, where $n \geq 1$ and $d(\secparam)=\lceil \log\frac{ 3\secparam}{n} \rceil \geq 1$.
%     \item \noindent Tomography procedure $(\tomography,\verify)$ 
% \end{itemize}

%\noindent A tomography procedure satisfying the above property was achieved in a work by O'Donnell and Wright~\cite{ODW16}.  \hnote{Can't use OW, because it's not time-efficient. Can use, say, the simple method of measuring each qubit using all possible Pauli strings. For $d$-dimensional systems this requires $\tilde{O}(d^4/\eps^2)$ samples.}

\noindent

\begin{figure}[!htb]
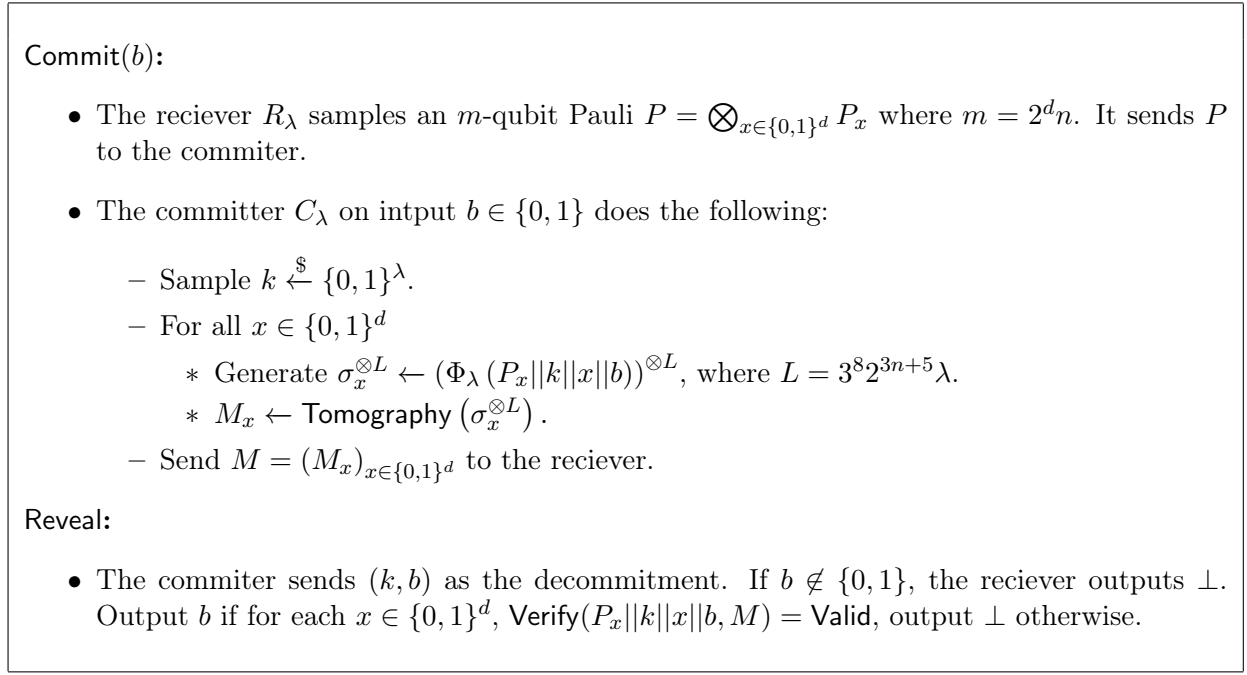

   \begin{tabular}{|p{16cm}|}
   \hline \\
\par \textbf{$\commit(b)$:} 
\begin{itemize}
    \item The reciever $R_{\secparam}$ samples an $m$-qubit Pauli $P = \bigotimes_{x\in\{0,1\}^d} P_x$ where $m=2^dn$. It sends $P$ to the commiter.
    \item The committer $C_{\secparam}$ on intput $b\in\{0,1\}$ does the following:
    \begin{itemize}
        \item Sample $k\xleftarrow{\$}\{0,1\}^{\secparam}$.
        \item For all $x\in\{0,1\}^d$
        \begin{itemize}
            \item Generate $\sigma_x^{\otimes L}\xleftarrow{}  \left(\Phi_\secparam\left(P_x||k||x||b\right)\right)^{\otimes L}$, where $L=3^{8} 2^{3n+5} \secparam$.
            \item $M_x \xleftarrow{} \tomography\left(\sigma_x^{\otimes L}\right).$ 
        \end{itemize}
        \item Send $M = \left(M_x\right)_{x\in\{0,1\}^d}$ to the reciever.
    \end{itemize}
\end{itemize}
\par \textbf{$\reveal$:} 
\begin{itemize}
    \item The commiter sends $\left(k,b\right)$ as the decommitment. If $b\not\in\{0,1\}$, the reciever outputs $\perp$. Output $b$ if for each $x\in\{0,1\}^d$, $\verify(P_x||k||x||b,M) = \valid$, output $\perp$ otherwise.
\end{itemize}
\\ 
\hline
\end{tabular}
\caption{Commitment scheme}
\label{fig:commitment}
\end{figure}

\noindent We prove that the construction in Figure~\ref{fig:commitment} satisfies correctness, computational hiding and statistical binding properties. 

\begin{lemma}[Correctness]
The commitment scheme in Figure~\ref{fig:commitment} satisfies correctness. 
\end{lemma}
\begin{proof}
This follows from Lemma~\ref{lem:samekey:first}.
\end{proof}

\begin{lemma}[Computational Hiding]
The commitment scheme in Figure~\ref{fig:commitment} satisfies computational hiding.
\end{lemma}

\begin{proof}
We prove the security via a hybrid argument. Fix $\secparam \in \mathbb{N}$. Consider a QPT adversary $R^{\ast}_{\secparam}$. 

\paragraph{Hybrid $H_{1,b}$, for all $b\in\{0,1\}$.} This corresponds to $C$ commiting to the bit $b$.

\paragraph{Hybrid $H_{2,b}$, for all $b\in\{0,1\}$.} This hybrid is the same as before except that for all $x\in \{0,1\}^d$, $\Phi_\secparam\left(P||k||x||b\right)$ replaced with $\left(\left(I\otimes P_x^b\right) \left(\ketbra{0}{0}\otimes\ketbra{\vartheta_{x}}{\vartheta_{x}}\right) \left(I\otimes P_x^b\right)\right)$ where $\ket{\vartheta_1},...,\ket{\vartheta_{2^d}}\xleftarrow{}\mathscr{H}_n$. 

The hybrids $H_{1,b}$ and $H_{2,b}$ are computationally indistinguishable because of the security of $PRFS$. $H_{2,0}$ and $H_{2,1}$ are identical by the unitary invariance property of Haar  distribution. Hence, $H_{1,0}$ and $H_{1,1}$ are computationally indistinguishable.

\end{proof}

\begin{lemma}[Statistical Binding]\label{lem:binding}
The commitment scheme in Figure~\ref{fig:commitment} satisfies $O(2^{-0.5\secparam})$-statistical binding.
\end{lemma}
\begin{proof}[Proof of Lemma~\ref{lem:binding}]
Let $C^{\ast} = \left\{C^{\ast}_{\lambda}\right\}_{\lambda\in\N}$ be a malicous committer. Execute the commit phase between $C^{\ast}_{\secparam}$ and $R_{\secparam}$. Let $\tau$ be the classical transcript and let $\sigma_{C^{\ast}R}$ be the joint state of $C^{\ast}R$. We first provide the description of an extractor.

\paragraph{Description of ${\cal E}$.} On the input $\tau = (P,M)$, the extractor does the following:
\begin{enumerate}
    \item For all $k'||b'\in\{0,1\}^{\lambda}\times\{0,1\}$, run for all $x\in\{0,1\}^d$,  $\verify(P_x||k'||x||b',M)$. 
    \item If for all $x\in\{0,1\}^d$, $\verify(P||k'||x||b',M) = \valid$, output $b'$. 
    \item Else output $\bot$.
\end{enumerate}

\begin{fact}\label{fact:almost_ortho}
Let $\mathcal{P}_m$ be the $m$-qubit Pauli group. Then, $$\Pr_{P\xleftarrow{\$}\mathcal{P}_m}\left[\exists k_0,k_1 : \forall x\in\{0,1\}^d,  \left|\bra{\psi_{k_0,x}}P_x\ket{\psi_{k_1,x}}\right|^2\geq\delta\right]\leq \delta^{-2^d}2^{2\secparam-m}.$$
\end{fact}
\begin{proof}
We use the following fact \cite[Fact 6.9]{AQY21}: Let $\ket{\psi}$ and $\ket{\phi}$ be two arbitrary $n$-qubit states. Then, $$\E_{P_x\xleftarrow{\$}\mathcal{P}_n} \left[\left|\bra{\psi}P_x\ket{\phi}\right|^{2}\right]=2^{-n}.$$
For any $k_0,k_1,x$ by the above fact, $\E_{P_x\xleftarrow{\$}\mathcal{P}_n} \left[\left|\bra{\psi_{k_0,x}}P_x\ket{\psi_{k_1,x}}\right|^{2}\right]=2^{-n}.$ Using Markov’s inequality we get that for all $\delta>0$, $$\Pr_{P_x\xleftarrow{\$}\mathcal{P}_n}\left[\left|\bra{\psi_{k_0,x}}P_x\ket{\psi_{k_1,x}}\right|^2\geq\delta\right]\leq \delta^{-1}2^{-n}.$$ Since, all $P_x$'s are independent, $$\Pr_{P\xleftarrow{\$}\mathcal{P}_m}\left[\forall x\in\{0,1\}^d,  \left|\bra{\psi_{k_0,x}}P_x\ket{\psi_{k_1,x}}\right|^2\geq\delta\right]\leq \left(\delta^{-1}2^{-n}\right)^{2^d}.$$ Using a union bound over all $k_0,k_1$, $$\Pr_{P\xleftarrow{\$}\mathcal{P}_m}\left[\exists k_0,k_1 : \forall x\in\{0,1\}^d,  \left|\bra{\psi_{k_0,x}}P_x\ket{\psi_{k_1,x}}\right|^2\geq\delta\right]\leq \delta^{-2^d}2^{2\secparam-m}.$$
\end{proof}

Let the transcript be $(P,M)$ where $P$ is chosen uniformly at random. Let $$p = \prob\left[ \mu \neq b^* \wedge \mu \neq \bot \ :\ \substack{ (\tau,\sigma_{C^*R})\xleftarrow {}\commit \langle C^*_{\secparam},R_{\secparam}\rangle,\\ \ \\ b^* \leftarrow {\cal E}(\tau),\\ \ \\ \mu \leftarrow \reveal\langle \tau,\sigma_{C^* R} \rangle  } \right]$$
Then
$$p = \Pr_{P\xleftarrow{\$}\mathcal{P}_m} \left[\exists k_0,k_1,b_0,b_1\ :\forall x\in\{0,1\}^d\ \substack{ \verify(P_x||k_0||x||b_0,M_x) = \valid,\\ \ \\ \verify(P_x||k_1||x||b_1,M_x) = \valid,\\ \ \\ b_0\neq b_1  } \right].$$
Without loss of generality we can assume $b_0 = 0$ and $b_1 = 1$, 
$$p = \Pr_{P\xleftarrow{\$}\mathcal{P}_m} \left[\exists k_0,k_1\ :\forall x\in\{0,1\}^d\ \substack{ \verify(P_x||k_0||x||0,M_x) = \valid,\\ \ \\ \verify(P_x||k_1||x||1,M_x) = \valid  } \right].$$

By Lemma~\ref{lem:diff_key_binding}, 
$$p\leq \Pr_{P\xleftarrow{\$}\mathcal{P}_m} \left[\exists k_0,k_1\ :\forall x\in\{0,1\}^d, \  Tr(P_x\ketbra{\psi_{k_1,x}}{\psi_{k_1,x}}P_x\ketbra{\psi_{k_0,x}}{\psi_{k_0,x}})\geq 542/729  \right]$$
    
By Fact~\ref{fact:almost_ortho}, 
$$p\leq \frac{729}{542}^{2^d} \left(2^{2\secparam-m}\right).$$
For $m\geq 3\secparam$, the protocol satisfies $O(2^{-0.5\secparam})$-statistical binding.

\end{proof}

\newcommand{\ct}{\mathsf{ct}}
\subsection{Encryption scheme}
\label{sec:encryption}
We construct a psuedo one-time pad scheme with classical communication from psuedorandom function-like quantum states. We first present the definition below. 
\begin{definition}[Psuedo One-Time Pad]
We say that a pair of QPT algorithms $(\enc,\dec)$ is a psuedo one-time pad if the following properties are satisfied: there exists a polynomial $M(\secparam)$ such that
\paragraph{Correctness:} There exists a negligible function $\varepsilon(\cdot)$ such that for every $\secparam$, every $x\in\{0,1\}^{M(\secparam)}$, $$\Pr_{\substack{k\xleftarrow{}\{0,1\}^\secparam\\  
    \ct \xleftarrow{}\enc_{\secparam}(k,x)}}\left[\dec_{\secparam}(k,\ct) = x\right]\geq 1-\varepsilon(\secparam).$$
    
\paragraph{Security:} For every $\secparam$, for every QPT adversary $A_{\secparam}$, there exists a negligible function $\varepsilon(\cdot)$ such that $x_1,x_2\in\{0,1\}^{M(\secparam)}$, $$\left|\Pr_{\substack{k\xleftarrow{}\{0,1\}^\secparam\\  
    \ct \xleftarrow{}\enc_{\secparam}(k,x_1)}}\left[A_{\secparam}(\ct) = 1\right] - \Pr_{\substack{k\xleftarrow{}\{0,1\}^\secparam\\  
    \ct \xleftarrow{}\enc_{\secparam}(k,x_2)}}\left[A_{\secparam}(\ct) = 1\right] \right|\leq \varepsilon(\secparam).$$

\end{definition}
\subsubsection{Construction}
Towards constructing an encryption scheme with classical communication, we use the verifiable tomography $(\tomography,\verify)$ described in~\Cref{fig:secondinstantiation} satisfying $L$-same-input correctness and $(L,\Pi,\distr)$-distributional different-input correctness. Let the set of channels associated with $(\tomography,\verify)$ be ${\cal C}=\{\Phi_{\secparam}\ :\ \secparam 
\in \mathbb{N}\}$. Recall that ${\cal C}$ is associated with $(d,n)$-PRFS. Refer to~\Cref{fig:secondinstantiation} for the description of $L,\Pi,\distr$ and $\Phi_{\secparam}$.

\paragraph{Construction.} We present the construction in Figure~\ref{fig:encryption}.
\begin{figure}[!htb]
   \begin{tabular}{|p{16cm}|}
   \hline \\
\par \textbf{$\enc_{\secparam}(k,x)$:} 
\begin{itemize}
   \item Parse $x=x_1 \cdots x_M$. 
    \item For $i\in[M]$, generate $L(\secparam)$ copies of $\sigma_{i} \leftarrow \left(\Phi_\secparam\left((k,i,x_i)\right)\right)$.
    \item Generate $u_i \xleftarrow{} \tomography\left(\sigma_{i}^{\otimes L(\secparam)}\right).$
    \item Output the ciphertext  $\ct=\left\{u_i\right\}_{i\in[M]}$.
\end{itemize}
\par \textbf{$\dec_{\secparam}(k,\ct)$:} 
\begin{itemize}
    \item Parse $\ct$ as $\left\{u_i\right\}_{i\in[M]}$.
    \item For $i\in [M]$, run $\left(\verify(k||i||0,u_i)\right)$. If $\valid$, set $x_i = 0$, else set $x_i = 1$.
    \item Output $x = x_1 \cdots x_M$.
\end{itemize}
\\ 
\hline
\end{tabular}
\caption{Encryption scheme}
\label{fig:encryption}
\end{figure}

\begin{theorem}
$(\enc_{\secparam},\dec_{\secparam})$ satisfies the correctness property of a psuedo one-time pad.
\end{theorem}
\begin{proof}
Let $\ct \leftarrow \enc_{\secparam}(k,x)$. Parse $\ct$ as $\{u_i\}_{i \in [l]}$.
\par For every $i \in [M]$, the following holds:
\begin{itemize}
    \item From~\Cref{lem:samekey:second}, $\verify \left( k||i||x_i,u_i \right)$ outputs $\valid$ with probability $1-\negl(\secparam)$. 
    \item From~\Cref{lem:diffkey:second}, $\verify \left( k||i||(1-x_i),u_i \right)$ outputs $\valid$ with probability at most $\negl(\secparam)$.
\end{itemize}
Thus, for any given $i$, the decryption algorithm can correctly determine $x_i$ with probability at least $1-\negl(\secparam)$. By union bound, it then follows that the probability that the decryption algorithm correctly determines all bits of $x$ is negligibly close to 1. 
% Fix a message $x\in\{0,1\}^{l}$. Let $x_i = 0$ for some $i$. Then by Lemma~\ref{lem:samekey:second}, $\verify\left(k||i||0,\tomography\left(\left(\Phi_\secparam\left(k,i,x_i\right)\right)^{\otimes\poly (\secparam)}\right)\right)$ outputs $\valid$ with probability greater than $1 - \negl(\secparam)$. Conversely, let $x_i = 1$. Then by Lemma~\ref{lem:diffkey:second}, $\verify\left(k||i||0,\tomography\left(\left(\Phi_\secparam\left(k||i||x_i\right)\right)^{\otimes\poly (\secparam)}\right)\right)$ outputs $\valid$ with probability less than $\negl(\secparam)$. Hence, probability that we output the correct bit at the $i^{th}$ position is greater than $1-\negl (\secparam)$. Hence by union bound, probability that we output all bits correctly is greater than $(1-\negl (\secparam))^{l} = 1-\negl' (\secparam)$.
\end{proof}

\begin{theorem}
$(\enc_{\secparam},\dec_{\secparam})$ satisfies the  security property of a pseudo one-time pad. 
\end{theorem}
\begin{proof}
We prove the security via a hybrid argument. Fix $\secparam$ and the messages $x_0,x_1\in\{0,1\}^{M}$. Consider a QPT adversary $A_{\secparam}$. 

\paragraph{Hybrid $H_{1,b}$, for $b\in\{0,1\}$.} Sample $k\xleftarrow{}\{0,1\}^{\secparam}$. Compute $\ct \xleftarrow{}\enc_{\secparam}(k,x_b)$. Output $\ct$. 

\paragraph{Hybrid $H_2$.} Sample $L(\secparam)$ copies of $n$-qubit Haar-random states $\ket{\vartheta_1},...,\ket{\vartheta_l}\xleftarrow{}\mathscr{H}_n$. For every $i \in [M]$, compute $u_i \leftarrow \tomography\left(\ket{\vartheta_1}^{\otimes L}\right)$. The output of this hybrid is $(u_1,\ldots,u_l)$.
\par The hybrids $H_{1,b}$, for $b \in \{0,1\}$, and $H_{2}$ are computationally indistinguishable from the security of PRS. 
\end{proof}

\section*{Acknowledgements}

The authors would like to thank the anonymous TCC 2022 reviewers for their helpful comments. 
The authors would also like to thank Fermi Ma for his suggestions that improved the bounds and the analysis in Section~\ref{sec:simpleranalysis:prs}. 

\ifllncs
  \bibliographystyle{splncs04}
  \bibliography{tcs}

@mastersthesis{Lowe21,
author={{Lowe, Angus}},
title={Learning Quantum States Without Entangled Measurements},
year={2021},
publisher="UWSpace",
eprinttype={hdl},
eprint={10012/17663},
school={University of Waterloo}
}

@inproceedings{Gav12,
  author    = {Dmitry Gavinsky},
  title     = {Quantum Money with Classical Verification},
  booktitle = {Proceedings of the 27th Conference on Computational Complexity, {CCC}
               2012, Porto, Portugal, June 26-29, 2012},
  pages     = {42--52},
  publisher = {{IEEE} Computer Society},
  year      = {2012},
  doi       = {10.1109/CCC.2012.10},
  bibsource = {dblp computer science bibliography, https://dblp.org}
}

@misc{MY21,
  eprinttype={arxiv},
  eprint={2112.06369},
  
  author = {Morimae, Tomoyuki and Yamakawa, Takashi},
  
  title = {Quantum commitments and signatures without one-way functions},
  
  publisher = {arXiv},
  
  year = {2021},
  
  copyright = {Creative Commons Attribution 4.0 International}
}

@inproceedings{ambainis2007quantum,
  author    = {Andris Ambainis and
               Joseph Emerson},
  title     = {Quantum t-designs: t-wise Independence in the Quantum World},
  booktitle = {22nd Annual {IEEE} Conference on Computational Complexity {(CCC} 2007),
               13-16 June 2007, San Diego, California, {USA}},
  pages     = {129--140},
  publisher = {{IEEE} Computer Society},
  year      = {2007},
  doi       = {10.1109/CCC.2007.26},
  bibsource = {dblp computer science bibliography, https://dblp.org}
}

@article{goldreich1986construct,
author = {Goldreich, Oded and Goldwasser, Shafi and Micali, Silvio},
title = {How to Construct Random Functions},
year = {1986},
issue_date = {Oct. 1986},
publisher = {Association for Computing Machinery},
address = {New York, NY, USA},
volume = {33},
number = {4},
issn = {0004-5411},
doi = {10.1145/6490.6503},
journal = {J. ACM},
month = {8},
pages = {792–807},
numpages = {16}
}

@article{dankert2009exact,
  title = {Exact and approximate unitary 2-designs and their application to fidelity estimation},
  author = {Dankert, Christoph and Cleve, Richard and Emerson, Joseph and Livine, Etera},
  journal = {Phys. Rev. A},
  volume = {80},
  issue = {1},
  pages = {012304},
  numpages = {6},
  year = {2009},
  month = {Jul},
  publisher = {American Physical Society},
  doi = {10.1103/PhysRevA.80.012304},
}

@article{roy2009unitary,
author={Roy, Aidan
and Scott, A. J.},
title={Unitary designs and codes},
journal={Designs, Codes and Cryptography},
year={2009},
month={Oct},
day={01},
volume={53},
number={1},
pages={13-31},
issn={1573-7586},
doi={10.1007/s10623-009-9290-2}
}

@article{haastad1999pseudorandom,
author = {H{\aa}stad, Johan and Impagliazzo, Russell and Levin, Leonid A. and Luby, Michael},
title = {A Pseudorandom Generator from any One-way Function},
journal = {SIAM Journal on Computing},
volume = {28},
number = {4},
pages = {1364-1396},
year = {1999},
doi = {10.1137/S0097539793244708}
}

@article{BHH16,
author={Brand{\~a}o, Fernando G. S. L.
and Harrow, Aram W.
and Horodecki, Micha{\l}},
title={Local Random Quantum Circuits are Approximate Polynomial-Designs},
journal={Communications in Mathematical Physics},
year={2016},
month={9},
day={01},
volume={346},
number={2},
pages={397-434},
issn={1432-0916},
doi={10.1007/s00220-016-2706-8},
}

@inproceedings{Zhandry12,
  author    = {Mark Zhandry},
  title     = {How to Construct Quantum Random Functions},
  booktitle = {53rd Annual {IEEE} Symposium on Foundations of Computer Science, {FOCS}
               2012, New Brunswick, NJ, USA, October 20-23, 2012},
  pages     = {679--687},
  publisher = {{IEEE} Computer Society},
  year      = {2012},
  doi       = {10.1109/FOCS.2012.37},
  bibsource = {dblp computer science bibliography, https://dblp.org}
}

@misc{harrow2013church,
  eprinttype={arxiv},
  eprint={1308.6595},
  
  author = {Harrow, Aram W.},
  
  title = {The Church of the Symmetric Subspace},
  
  publisher = {arXiv},
  
  year = {2013},
  
  copyright = {arXiv.org perpetual, non-exclusive license}
}

@inproceedings{ji2018pseudorandom,
  author    = {Zhengfeng Ji and
               Yi{-}Kai Liu and
               Fang Song},
  editor    = {Hovav Shacham and
               Alexandra Boldyreva},
  title     = {Pseudorandom Quantum States},
  booktitle = {Advances in Cryptology - {CRYPTO} 2018 - 38th Annual International
               Cryptology Conference, Santa Barbara, CA, USA, August 19-23, 2018,
               Proceedings, Part {III}},
  series    = {Lecture Notes in Computer Science},
  volume    = {10993},
  pages     = {126--152},
  publisher = {Springer},
  year      = {2018},
  doi       = {10.1007/978-3-319-96878-0_5},
  bibsource = {dblp computer science bibliography, https://dblp.org}
}

@book{nielsen_chuang_2010, place={Cambridge}, title={Quantum Computation and Quantum Information: 10th Anniversary Edition}, DOI={10.1017/CBO9780511976667}, publisher={Cambridge University Press}, author={Nielsen, Michael A. and Chuang, Isaac L.}, year={2010}}

@inproceedings{AharonovKN98,
  author    = {Dorit Aharonov and
               Alexei Y. Kitaev and
               Noam Nisan},
  editor    = {Jeffrey Scott Vitter},
  title     = {Quantum Circuits with Mixed States},
  booktitle = {Proceedings of the Thirtieth Annual {ACM} Symposium on the Theory
               of Computing, Dallas, Texas, USA, May 23-26, 1998},
  pages     = {20--30},
  publisher = {{ACM}},
  year      = {1998},
  doi       = {10.1145/276698.276708},
  bibsource = {dblp computer science bibliography, https://dblp.org}
}

@inproceedings{BS19,
  author    = {Zvika Brakerski and
               Omri Shmueli},
  editor    = {Dennis Hofheinz and
               Alon Rosen},
  title     = {(Pseudo) Random Quantum States with Binary Phase},
  booktitle = {Theory of Cryptography - 17th International Conference, {TCC} 2019,
               Nuremberg, Germany, December 1-5, 2019, Proceedings, Part {I}},
  series    = {Lecture Notes in Computer Science},
  volume    = {11891},
  pages     = {229--250},
  publisher = {Springer},
  year      = {2019},
  doi       = {10.1007/978-3-030-36030-6_10},
  bibsource = {dblp computer science bibliography, https://dblp.org}
}

@misc{AQY21,
  eprinttype={arxiv},
  eprint={2112.10020v2},
  
  author = {Ananth, Prabhanjan and Qian, Luowen and Yuen, Henry},
  
  title = {Cryptography from Pseudorandom Quantum States},
  
  publisher = {arXiv},
  
  year = {2021},
  
  copyright = {arXiv.org perpetual, non-exclusive license}
}

@article{lo1997quantum,
  title = {Is Quantum Bit Commitment Really Possible?},
  author = {Lo, Hoi-Kwong and Chau, H. F.},
  journal = {Phys. Rev. Lett.},
  volume = {78},
  issue = {17},
  pages = {3410--3413},
  numpages = {0},
  year = {1997},
  month = {Apr},
  publisher = {American Physical Society},
  doi = {10.1103/PhysRevLett.78.3410}
}

@article{mayers1997unconditionally,
  title = {Unconditionally Secure Quantum Bit Commitment is Impossible},
  author = {Mayers, Dominic},
  journal = {Phys. Rev. Lett.},
  volume = {78},
  issue = {17},
  pages = {3414--3417},
  numpages = {0},
  year = {1997},
  month = {Apr},
  publisher = {American Physical Society},
  doi = {10.1103/PhysRevLett.78.3414}
}

@article{Huang2020,
author={Huang, Hsin-Yuan
and Kueng, Richard
and Preskill, John},
title={Predicting many properties of a quantum system from very few measurements},
journal={Nature Physics},
year={2020},
month={10},
day={01},
volume={16},
number={10},
pages={1050-1057},
issn={1745-2481},
doi={10.1038/s41567-020-0932-7}
}

@inproceedings{BrakerskiS20,
  author    = {Zvika Brakerski and
               Omri Shmueli},
  editor    = {Daniele Micciancio and
               Thomas Ristenpart},
  title     = {Scalable Pseudorandom Quantum States},
  booktitle = {Advances in Cryptology - {CRYPTO} 2020 - 40th Annual International
               Cryptology Conference, {CRYPTO} 2020, Santa Barbara, CA, USA, August
               17-21, 2020, Proceedings, Part {II}},
  series    = {Lecture Notes in Computer Science},
  volume    = {12171},
  pages     = {417--440},
  publisher = {Springer},
  year      = {2020},
  doi       = {10.1007/978-3-030-56880-1_15},
  bibsource = {dblp computer science bibliography, https://dblp.org}
}

@inproceedings{Kretschmer21,
  author    = {William Kretschmer},
  editor    = {Min{-}Hsiu Hsieh},
  title     = {Quantum Pseudorandomness and Classical Complexity},
  booktitle = {16th Conference on the Theory of Quantum Computation, Communication
               and Cryptography, {TQC} 2021, July 5-8, 2021, Virtual Conference},
  series    = {LIPIcs},
  volume    = {197},
  pages     = {2:1--2:20},
  publisher = {Schloss Dagstuhl - Leibniz-Zentrum f{\"{u}}r Informatik},
  year      = {2021},
  doi       = {10.4230/LIPIcs.TQC.2021.2},
  bibsource = {dblp computer science bibliography, https://dblp.org}
}

@article{Zhandry15,
  author    = {Mark Zhandry},
  title     = {A note on the quantum collision and set equality problems},
  journal   = {Quantum Inf. Comput.},
  volume    = {15},
  number    = {7{\&}8},
  pages     = {557--567},
  year      = {2015},
  doi       = {10.26421/QIC15.7-8-2},
  bibsource = {dblp computer science bibliography, https://dblp.org}
}

@inproceedings{C:Zhandry12,
  author    = {Mark Zhandry},
  editor    = {Reihaneh Safavi-Naini and
  Ran Canetti},
  title     = {Secure Identity-Based Encryption in the Quantum Random Oracle Model},
  booktitle = {Advances in Cryptology - {CRYPTO} 2012 - 32nd Annual Cryptology Conference,
  Santa Barbara, CA, USA, August 19-23, 2012. Proceedings},
  series    = {Lecture Notes in Computer Science},
  volume    = {7417},
  pages     = {758--775},
  publisher = {Springer},
  year      = {2012},
  doi       = {10.1007/978-3-642-32009-5_44},
  bibsource = {dblp computer science bibliography, https://dblp.org}
}
\else
  \printbibliography
\fi

\newpage 

\appendix

% no longer used
%\section{Approximate state generation}
%\input{approximatestategen}

%\section{Separation}
%\input{bqp_neq_pp}
\section{Small-Range Distributions over Unitary Operators}
\label{sec:quantum-small-range}

Let $\mathcal U$ be a distribution over unitary operators over a finite Hilbert space, $r$ be a positive integer and $\mathcal X$ be a finite set.
(Looking ahead, we are going to sample a unitary from a distribution and potentially invoke this unitary multiple times. This is not a channel as all invocations are going to use the same unitary.)
We define $\mathcal U^{\mathcal X} := \bigoplus_{x \in \mathcal X} U_x$ to be the unitary that maps $\ket x \otimes \ket y \to \ket x \otimes U_x \ket y$, with each $U_x$ being an independent sample from $\mathcal U$.
%(Note that this is incomparable with channels, as once we have sampled a unitary we could invoke it multiple times.) \pnote{invoke what multiple times?}
We define small-range distributions $\SR_r^{\mathcal U}(\mathcal X)$ sampled as follows:
\begin{itemize}
  \item For each $i \in [r]$, sample a $U_i$ from $\mathcal U$.
  \item For each $x \in \mathcal X$, sample a random $i_x \in [r]$, so that the unitary maps $\ket x \otimes \ket y \to \ket x \otimes U_{i_x} \ket y$ for any state $\ket y$.
\end{itemize}

\begin{theorem}[{\cite[Corollary VII.5]{Zhandry12}}]
  \label{thm:bitstring-small-range}
  Let $D$ be a distribution over bit strings \pnote{of what length?} \luowen{Does not matter}, and let $\mathcal U$ be the unitary distribution for random oracles \pnote{what does unitary distribution for random oracles mean?} \luowen{keep reading, in particular, read immediately after this}, i.e.\,it corresponds to $D$ where for each bit string $m$, we associate a unitary that maps $\ket y \to \ket{y \oplus m}$ in the computational basis.
  The output distributions of a quantum algorithm making $q$ quantum queries to an oracle either drawn from $\SR_r^D(\mathcal X)$ from $D^{\mathcal X}$ are $\ell q^3/r$-close, where $\ell = 8\pi^2/3 < 27$.
\end{theorem}

The goal of this section is to generalize this theorem to the setting where $D$ is a (possibly infinite-support) distribution of arbitrary quantum unitary oracles instead of a distribution over (finite-length) bit strings.

Generalizing the proof of Zhandry's theorem to this setting seems difficult.
For example, the proof uses the fact that the query oracle operates in the computational basis, which a generic unitary might not satisfy.
Therefore, we instead are going to use a downsampling trick, which allows us to invoke Zhandry's theorem as a black box.
Furthermore, we are going to crucially rely on the fact that the bound given by Zhandry's theorem is independent of the size of $D$ or $\mathcal X$.

We invoke some inequalities about the diamond norm.
The diamond norm intuitively captures the best distinguishing advantage from a single invocation of one of two channels.
We refer the readers to the work of Aharonov et al.\,\cite{AharonovKN98} for the formal definitions.

\begin{lemma}[{\cite[lemma 12.6]{AharonovKN98}}]
  \label{lem:diamond2op}
  Let $V, W$ be operators.
  If $\operatornorm{V}, \operatornorm{W} \le 1$, then \[\diamondnorm{V (\cdot) V^\dagger - W (\cdot) W^\dagger} \le 2 \operatornorm{V - W}.\]
\end{lemma}
\begin{lemma}[{\cite[lemma 13]{AharonovKN98}}]
  \label{lem:diamond-concat}
  Let $T_1, T_2, T_1', T_2'$ be super-operators with norm $\le 1$, such that $\diamondnorm{T_j' - T_j} \le \varepsilon_j$ for $j = 1, 2$.
  Then $\diamondnorm{T_2'T_1' - T_2T_1} \le \varepsilon_1 + \varepsilon_2$.
\end{lemma}

\begin{lemma}[Finite-support downsampling lemma]
  \label{lem:downsampling-nonuniform}
  For any distribution $\mathcal U$ over unitary operators over a finite Hilbert space and any $\varepsilon > 0$, there exists a distribution $\mathcal F$ \pnote{over what? set of unitaries?} with finite support such that any $q$-query (otherwise unbounded) quantum algorithm cannot distinguish a unitary from either distribution with advantage more than $\varepsilon q$.
\end{lemma}
\begin{proof}
  Let $\mathcal N$ be the epsilon-net for the unitary group with $|\mathcal N| < \infty$ such that for any unitary $U$, there exists an approximation $U' \in \mathcal N$ such that $\operatornorm{U - U'} \le \frac\varepsilon{2}$.
  We construct $\mathcal F$ by mapping every unitary from $\mathcal U$ to its approximation in $\mathcal N$.
  It follows from \Cref{lem:diamond2op,lem:diamond-concat} that for any fixed $U, U'$, any $q$-query quantum algorithm's distinguishing advantage is at most $\varepsilon q$.
  Therefore the lemma follows by averaging over sampling from $\mathcal U$.
\end{proof}

\begin{lemma}[Uniform downsampling lemma]
  \label{lem:downsampling}
  For any distribution $\mathcal U$ over unitary operators over a finite Hilbert space and any $\varepsilon > 0$, there exists an integer $n > 0$ and a family of unitaries $\mathcal T = (U_1, ..., U_{2^n})$ such that any $q$-query quantum algorithm cannot distinguish a unitary from a uniformly sampled $U_i$ from $\mathcal T$, or $\mathcal U$ with advantage more than $\varepsilon q$.
\end{lemma}
\begin{proof}
  Let $\mathcal F$ be the distribution guaranteed by \Cref{lem:downsampling-nonuniform} with distinguishing advantage $\varepsilon q/2$.
  Consider a random family $\mathcal T$ where each entry is an independent sample from $\mathcal F$.
  By the law of large numbers, the expected total variation distance to $\mathcal F$ goes to 0 as $n$ goes to infinity, which implies the existence of a sequence of families whose distance is at most $\varepsilon/2 \le \varepsilon q/2$.
  Therefore, by triangular inequality, the overall distinguishing advantage is at most $\varepsilon q$. \pnote{TVD from what to ${\cal F}$? can you write mathematical expressions for these statements?}
  %\luowen{This proof would only work for finite-support distribution since I am using TV distance. For an infinite-support uniform-ish distribution, the TV distance will always be 1 for any finite number of samples. It would be great if it generalizes.} update: fixed via eps-nets
\end{proof}

This lemma shows that we can downsamples an arbitrary distribution to a uniform distribution (with potentially huge support) followed by postprocessing, which is compatible with Zhandry's theorem.
We do not make explicit the support of the downsampled distribution, and this is also not necessary as the loss from Zhandry's theorem is actually independent from the size of the support.
Indeed, we combine these to show that we can generalize Zhandry's theorem over arbitrary unitary distribution with (asymptotically) the same loss.

\begin{theorem}
  \label{thm:smallrange-unitary}
  The output distributions of a quantum algorithm making $q$ queries to an oracle either drawn from $\SR_r^{\mathcal U}(\mathcal X)$ from $\mathcal U^{\mathcal X}$ are $300q^3/r$-close.
\end{theorem}
\begin{proof}
  Let $\varepsilon = \frac{q^2}{r|\mathcal X|}$, and let $\mathcal T$ be the distribution with length $n$ guaranteed by \Cref{lem:downsampling}.
  We show the theorem by the following sequence of hybrids:
  \begin{description}
    \item[Hybrid 0] The oracle is drawn from $\mathcal U^{\mathcal X}$.
    \item[Hybrid 1] The oracle is drawn from $\mathcal T^{\mathcal X}$.
      Since the oracle is essentially a direct sum of $|\mathcal X|$ independent unitary samples, the distinguishing advantage between the output distributions of this and the last hybrid is at most $\varepsilon q \cdot |\mathcal X| = q^3/r$.
    \item[Hybrid 2] We instead sample a classical random oracle $(\{0, 1\}^n)^{\mathcal X}$, and simulate the distribution of $\mathcal T^{\mathcal X}$ using two queries to the random oracle.
      In particular, we first query the random oracle to get $i \in \bit^n$, apply $U_i$, and then uncompute $i$ in superposition.
      By construction, this is perfectly indistinguishable to the last hybrid.
    \item[Hybrid 3] We change the classical random oracle to be instead sampled by $\SR_r^{\{0, 1\}^n}(\mathcal X)$.
      By \Cref{thm:bitstring-small-range} and the construction, the distance to the last hybrid is at most $27(2q)^3/r$.
    \item[Hybrid 4] We switch the oracle to $\SR_r^{\mathcal T}(\mathcal X)$.
      This is perfectly indistinguishable to the last hybrid.
    \item[Hybrid 5] We switch the oracle to $\SR_r^{\mathcal U}(\mathcal X)$.
      The distance for this is again at most $q^3/r$.
  \end{description}
  By triangle inequality, the two distributions are $< 300q^3/r$ indistinguishable.
\end{proof}

\end{document}